\documentclass[journal]{IEEEtran}

\usepackage{amsmath,amsfonts, amssymb, amsthm}
\usepackage{array}
\usepackage{textcomp}
\usepackage{stfloats}
\usepackage{url}
\usepackage{verbatim}
\usepackage{graphicx}
\usepackage{balance}
\usepackage[OT1]{fontenc}
\usepackage{mathtools}    %
\usepackage{algorithm}
\usepackage{algpseudocode}
\usepackage{xcolor}
\usepackage{cases}
\usepackage{hyperref}
\usepackage[font={small}]{caption}
\usepackage{subcaption}
\usepackage{cite}
\usepackage{mathrsfs}
\usepackage{siunitx}
\usepackage{booktabs}
\usepackage{tikz}
\usepackage{enumitem}
\usepackage{multirow}
\usepackage{adjustbox} %

\usepackage[resetlabels]{multibib}
\newcites{Supp}{References}

\usetikzlibrary{positioning, arrows, calc, 3d,}
\mathtoolsset{showonlyrefs}

\def\BibTeX{{\rm B\kern-.05em{\sc i\kern-.025em b}\kern-.08em
    T\kern-.1667em\lower.7ex\hbox{E}\kern-.125emX}}

\DeclareMathOperator*{\argmin}{argmin}

\newcommand{\be}{\begin{equation}}
\newcommand{\ee}{\end{equation}}

\newtheorem{theorem}{Theorem}
\newtheorem*{theorem*}{Theorem}

\newtheorem*{corollary*}{Corollary}
\newtheorem{lemma}{Lemma}
\newtheorem*{lemma*}{Lemma}
\newtheorem{defn}{Definition}
\newtheorem{prop}{Proposition}
\newtheorem*{prop*}{Proposition}

\begin{document}
\title{CLAMP: Majorized Plug-and-Play for Coherent 3D Lidar Imaging}

\author{
    Tony G. Allen,
    \thanks{Tony G. Allen is with Oak Ridge National Laboratory, 1 Bethel Valley Road, Oak Ridge, TN 37830, USA, but this work was completed while at the Air Force Research Laboratory, WPAFB, OH 45433, USA, and Purdue University, West Lafayette, IN 47907, USA
    (email: \href{mailto:allentg@ornl.gov}{allentg@ornl.gov}).} 
    David J. Rabb,
    \thanks{David J. Rabb is with the Sensors Directorate, Air Force Research Laboratory, WPAFB, OH 45433, USA (email: \href{mailto:david.rabb@us.af.mil}{david.rabb@us.af.mil}).}
    \\Gregery T. Buzzard,~\IEEEmembership{Senior Member,~IEEE,}
    \thanks{Gregery T. Buzzard is with the Department of Mathematics, Purdue University, West Lafayette, IN 47907, USA (email: \href{mailto:buzzard@purdue.edu}{buzzard@purdue.edu}) and was partially supported by NSF CCF-1763896.}
    Charles A. Bouman,~\IEEEmembership{Fellow,~IEEE.}
    \thanks{Charles A. Bouman is with the School of Electrical and Computer Engineering, Purdue University, West Lafayette, IN 47907, USA (email: \href{mailto:bouman@purdue.edu}{bouman@purdue.edu}) and was partially supported by the Showalter Trust.}
    \thanks{(\textit{Corresponding author: Tony G. Allen})}
}

\markboth{IEEE TRANSACTIONS ON COMPUTATIONAL IMAGING, VOL. 11, 2025}%
{Allen, \MakeLowercase{\textit{et al.}}: CLAMP: Majorized Plug-and-Play for Coherent 3D Lidar Imaging}
    
{\maketitle}

\begin{abstract}
    Coherent lidar uses a chirped laser pulse for 3D imaging of distant targets.   
    However, existing coherent lidar image reconstruction methods do not account for the system's aperture, resulting in sub-optimal resolution. 
    Moreover, these methods use majorization-minimization for computational efficiency, but do so without a theoretical treatment of convergence.
    
    In this paper, we present Coherent Lidar Aperture Modeled Plug-and-Play (CLAMP) for multi-look coherent lidar image reconstruction.
    CLAMP uses multi-agent consensus equilibrium (a form of PnP) to combine a neural network denoiser with an accurate surrogate forward model of coherent lidar.
    Additionally, CLAMP introduces a computationally efficient FFT-based method to account for the system's aperture to improve resolution of reconstructed images.
    Furthermore, we formalize the use of majorization-minimization in consensus optimization problems and prove convergence to the exact consensus equilibrium solution. 
    Finally, we apply CLAMP to synthetic and measured data to demonstrate its effectiveness in producing high-resolution, speckle-free, 3D imagery.
    
\end{abstract}


\section{Introduction}

\IEEEPARstart{C}{oherent} lidar is a method for three-dimensional (3D) imaging in which a target is flood-illuminated by a frequency-modulated laser source.
In multi-look lidar, multiple snapshots of the reflected wave are recorded on a 2D detector array by a spatial-heterodyne interferometric system~\cite{spencerSpatialHeterodyne}. 
This process yields a collection of 2D measurements over the duration of the chirped waveform. 
A 3D complex-valued volume can be recovered by a demodulation and filtering process~\cite{farhat-1980, marronThreedimensionalFineresolutionImaging1992a,marronThreedimensionalLenslessImaging1992}.  This volume can be further processed to produce a 3D image of the surfaces in the scene. 

There are several challenges in creating high-quality 3D images from multi-look lidar data. 
Recent efforts have focused primarily on correcting phase aberrations caused by non-concurrent measurements~\cite{staffordPhaseGradientAlgorithm2016}, object motion~\cite{banetEffectsofspeckledecorellation}, or atmospheric turbulence~\cite{farrissIterativePhaseEstimation2021,farrissSharpnessbasedCorrectionMethods2018, banetImageSharpening3D2021, banet3DMultiplaneSharpness2021}.
However, beyond addressing phase aberrations, reconstruction algorithms also need to reduce speckle, noise, and aperture-induced blur.

The most straightforward method to reconstruct multi-look coherent lidar images, often called speckle averaging, involves averaging the squared magnitude of back-projected, demodulated data from multiple looks.
Assuming the speckle decorrelates between looks, this averaging reduces speckle in the final reconstruction.
However, this method does not incorporate regularization or priors.

In complex-valued image reconstruction problems, researchers generally have used model-based iterative reconstruction (MBIR) methods that regularize the complex-valued reflectance of the target, assuming the scene reflectance to have locally correlated phase and magnitude~\cite{oktemSparsitybasedThreedimensionalImage2019, sar3d-pnp}.
However, in many scenarios, the phase of the reflectance is well-modeled as random, without spatial or temporal correlation~\cite{goodmanSpecklePhenomenaOptics2020, potter-cetin-sparse, munsonImageReconstructionFrequencyoffset1984}.This observation can help simplify the estimation problem by reducing its effective dimensionality.

In a related 3D near-field radar imaging problem, the authors of~\cite{mimo3d} leverage this observation to regularize only the magnitude of the reflectance, leaving the phase unregularized.
In contrast, the papers of Pellizzari et al.~\cite{pellizzariPhaseErrorEstimation2017,pellizzariSyntheticApertureLadar2017,pellizzariCoherentPlugandPlayDigital2020,pellizzariImagingDistributedvolumeAberrations2019a, pellizzariOpticallyCoherentImage2017} introduce an approach for 2D coherent lidar that uses a Bayesian framework to directly estimate the speckle-free, real-valued reflectivity of the target.
The real-valued reflectivity roughly corresponds to the image seen using broad-spectrum illumination; since the reflectivity is relatively smooth, it can be more strongly regularized to further reduce the effective dimensionality of the problem.
While estimating the reflectivity introduces a seemingly intractable, nonlinear relationship between measured data and the estimated image, Pellizzari et al. overcomes this problem by using the Expectation-Maximization (EM) algorithm~\cite{baum-em-66, baum-em-70, dempsterMaximumLikelihoodIncomplete1977}.

A primary benefit of these methods is that they enable the use of advanced real-valued image priors through Plug-and-Play (PnP) algorithms~\cite{venkatakrishnanPlugandPlayPriorsModel2013, sreehariPlugandPlayPriorsBright2016a,boumanFoundationsComputationalImaging2022}.
Despite the success of advanced priors in 2D image reconstruction~\cite{ulugbek-spm2023}, relatively little work has been done using 3D priors in lidar image reconstruction or related modalities.
It is important to use 3D priors instead, as the images involve complex 3D structures and spatial dependencies that cannot be adequately captured by 2D models.
Early research made a simplifying assumption that there is only a single surface depth at each 2D pixel~\cite{McLaughlin2016, RappGoyal2017}, while more recent works used
Total Variation (TV)~\cite{rudinTV} and mixed-norm ($\ell_{2,1}$) regularization~\cite{mclaughlinL21TV}, deep point cloud denoiser~\cite{Tachella2019Nature}, or a 3D U-Net denoiser~\cite{mimo3d} to jointly exploit range and cross-range correlations and account for multiple surface depths at each cross-range pixel.
However, these methods have not been applied to 3D coherent lidar imaging.

Another challenge in lidar reconstruction is to efficiently correct the blur caused by the finite extent of the system's aperture.
The aforementioned MBIR methods~\cite{pellizzariPhaseErrorEstimation2017,pellizzariSyntheticApertureLadar2017,pellizzariCoherentPlugandPlayDigital2020,pellizzariImagingDistributedvolumeAberrations2019a, pellizzariOpticallyCoherentImage2017} make a simplifying assumption that the forward model operator is an orthogonal matrix. This greatly reduces computational complexity, but effectively disregards the system aperture.
On a similar problem in polarimetric synthetic-aperture radar,~\cite{tuckerSpeckleSuppressionMultiChannel2020} uses the same EM algorithm approach, but addresses the aperture problem by using a greedy graph-coloring-based matrix probing method to do an approximate matrix inversion.
\begin{figure*}[ht]
    \centering
    \includegraphics[width=0.99\textwidth]{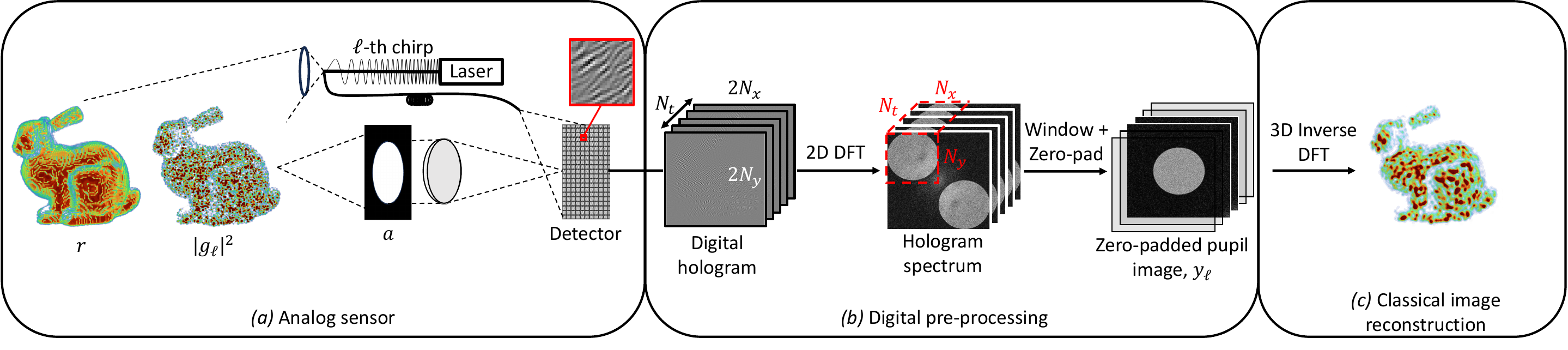}
    \caption{The measurement process and conventional data processing pipeline for a 3D coherent lidar imaging system using image plane recording geometry.
    \textit{(a)} Multiple chirped laser pulses, or looks, illuminate a target.
    Each look is assumed to generate an independent speckle pattern, $g_{\ell}$.
    Reflected light propagates to an optical system with aperture, $a$, where a lens focuses an image onto a focal plane array. 
    The image is interfered with a delayed local reference beam to form a hologram, which is sampled in time to create $N_t$ frames of size $2N_x \times 2N_y$ for each look.
    \textit{(b)} The 2D DFT of each hologram frame is taken to separate the frequency components of the target and reference beam.
    The pupil images are then windowed to size $N_x \times N_y \times N_t$, centered, and zero-padded by a factor, $q$, to form a 3D stack of complex fields, $y_{\ell}$ for the ${\ell}^{th}$ look. 
    \textit{(c)} A noisy 3D image of the target can be reconstructed from $y_{\ell}$ by the inverse 3D DFT and averaging over $\ell = 1,\ldots, L$.
    }
    \label{fig:lidarSystemDiagram}
\end{figure*}

Additionally, from a theoretical point of view, a variety of lidar reconstruction algorithms combine the EM algorithm with PnP or ADMM~\cite{boydDistributedOptimizationStatistical2010}; however, convergence results for this type of combined algorithm are limited.
In~\cite{luUnifiedAlternatingDirection2016}, the authors prove an ergodic convergence rate when using general majorization-minimization~\cite{mairalIncrementalMajorizationMinimizationOptimization2015,mairalOptimizationFirstOrderSurrogate, langeOptimizationTransferUsing} principles within ADMM\@, but they do not prove non-ergodic convergence of the iterates. 

In this paper, we present Coherent Lidar Aperture Modeled Plug-and-Play (CLAMP), a majorized PnP algorithm for multi-look coherent lidar imaging that accounts for the imaging system's aperture, which builds on our previous publication~\cite{allenMultiAgentConsensusEquilibrium}.
Using the Multi-Agent Consensus Equilibrium (MACE) framework~\cite{buzzardPlugandPlayUnpluggedOptimizationFree2018}, CLAMP combines multi-look coherent lidar data with a deep 3D image prior model to produce a 3D reconstruction with reduced noise, speckle, and increased resolution relative to traditional reconstruction methods.
Specifically, we make the following novel contributions:

\begin{itemize} 
    \item Accurate lidar aperture modeling using a computationally efficient FFT-based method;
    \item Incorporation of a deep neural network prior in 3D multi-look coherent lidar reconstruction;
    \item Formulation of a unified framework for majorization-minimization within MACE with theoretical guarantees of convergence to an exact MACE solution;
    \item Results on simulated and measured data that show CLAMP yields 3D images with reduced speckle and improved resolution.
\end{itemize}

A key component of CLAMP is the use of majorization-minimization, or surrogate function optimization, in the MACE framework.
We use convex surrogate functions in place of the proximal agents in MACE and show that this method converges to the exact MACE solution.
This surrogate-based MACE applies to many imaging problems and modalities and encompasses several common approximation methods and existing algorithms as special cases.

\section{Background}\label{sec:background}
\noindent In this section, we provide an overview of the multi-look coherent lidar imaging problem, develop the models used in our approach, and introduce the MACE framework.
We strive to provide a comprehensive background on the problem and the methods used in our approach, but for a more detailed discussion of 3D coherent lidar imaging, we refer the reader to~\cite{farhat-1980, marronThreedimensionalFineresolutionImaging1992a, marronThreedimensionalLenslessImaging1992, farrissIterativePhaseEstimation2021}.

\subsection{Coherent Lidar Imaging Measurement System}\label{sec:lidar}

\noindent Figure~\ref{fig:lidarSystemDiagram}(a) illustrates a typical coherent lidar imaging system with an image-plane recording geometry.
In this system, a series of $L$ chirped laser pulses, or looks, illuminates a target with reflectivity, $r$.
Each look generates a speckle pattern, $g_{\ell}$, which is a complex-valued field that represents the target's reflectance.
Throughout this work, we assume that the target is stationary and that the speckle realizations and measurement noise samples are statistically independent from look to look.
For each look, $\ell = 1, \ldots, L$, the corresponding complex reflectance, $g_{\ell}$ propagates to the aperture, $a$, where a thin lens forms an image of the target on a focal plane array.
The resulting image is mixed with a delayed, off-axis, local reference beam, which results in an interference pattern, or hologram, on the focal plane array.

As illustrated in Figure~\ref{fig:lidarSystemDiagram}(b), the hologram is sampled in time to create a series of $N_{t}$ frames of size $2N_x \times 2N_y$ for each look.
Intuitively, each of these hologram frames corresponds to a slightly different laser wavelength due to the chirp slope.
This stack of 2D holograms will can be used to recover 3D image of the complex reflectance of the target.
For information on the principles of 3D holographic imaging, we refer the reader to~\cite{farhat-1980, marronThreedimensionalLenslessImaging1992,marronThreedimensionalFineresolutionImaging1992a}.

Figure~\ref{fig:lidarSystemDiagram}(b) and (c) illustrate the conventional processing steps to reconstruct a 3D image of the target from this stack of $N_t$ holograms.
First, the 2D discrete Fourier transform (DFT) is taken for each frame in a look.
Since the frames are real valued, each DFT results in a complex image that contains two circular disks corresponding to the complex conjugate circular apertures in the pupil plane.
One of the two circular apertures is then embedded in a bounding square array (shown with a red dotted line) and windowed to give the pupil image.
This image is zero-padded to increase the effective sampling rate in the space and time domains and yields a larger 3D stack of pupil images, $y_{\ell}$.
Specifically, we zero-pad the pupil image by a factor of $q$ in each of the three dimensions; so the final array, $y_{\ell}$, is $q^3$ times the size of the un-padded images.
The last step in conventional reconstruction is to take the inverse 3D DFT of each $y_\ell$ and average over $\ell$, which yields a candidate reconstruction in cross-range and depth.

The final step in conventional reconstruction limits reconstruction quality in two key ways: lost spatial and temporal resolution due to the missing information in the frequency domain, and noise due to speckle.

In contrast, our algorithm replaces the conventional final 3D DFT and averaging step with an iterative model-based strategy which employs a physics-based forward model in conjunction with neural network denoisers to ensure the reconstructed image is high-resolution and free of speckle and noise.

\subsection{Forward Model}\label{sec:forward}

Assuming the target's depth is small compared to the range of the lidar system and that the fractional bandwidth of the chirp is small~\cite{farrissIterativePhaseEstimation2021}, the complex field in the pupil plane, $\tilde{y}_{\ell}: \mathbb{R}^3 \to \mathbb{C}$, for $\ell = 1, \ldots, L$, can be determined by the Fresnel propagation integral~\cite{goodmanIntroductionFourierOptics} of the complex reflectance, $\tilde{g}_{\ell}: \mathbb{R}^3 \to \mathbb{C}$ as
\begin{align}
    \tilde{y}_{\ell}(\xi, \nu, t) &= \frac{e^{2\pi i d / \lambda(t)}}{i\lambda d} \tilde{a}(\xi, \nu, t) \exp\left\{\frac{2\pi i}{2\lambda(t)d}(\xi^2 + \nu^2) \right\} \nonumber \\
    & \iiint_{-\infty}^{\infty} \left\{ \tilde{g}(x,y,z)  \exp\left\{\frac{2\pi i}{2\lambda d}(x^2 + y^2) \right\} \right\} \nonumber \\
    & \exp\left\{-\frac{2\pi i}{\lambda d}(\xi x + \nu y + \frac{1}{\lambda(t)}z) \right\} \,dx\,dy\,dz,
\end{align}
where $(\xi, \nu)$ are the coordinates in the pupil plane, $(x,y,z)$ are the coordinates of the target, $d_z$ is the distance from the object plane to the pupil plane, $\tilde{a}(\xi, \nu, t)$ is the aperture, and $\lambda(t)$ is the instantaneous wavelength of the chirp at time $t$.
By a small angle approximation, the quadratic phase term outside the integral can be ignored, and since we only consider the magnitude of the complex field, the quadratic phase term inside the integral can be absorbed into $\tilde{g}_{\ell}$ as in~\cite{pellizzariSyntheticApertureLadar2017}.
This results in a Fourier relationship,
\begin{equation}\label{eq:continuous-forward}
    \tilde{y}_{\ell}(\xi, \nu, t) = \frac{e^{2\pi i d / \lambda(t)}}{i\lambda d} \tilde{a}(\xi, \nu, t) \cdot \mathscr{F}\left[\tilde{g}_{\ell}\right]\left(\frac{\xi}{\lambda d}, \frac{\nu}{\lambda d}, \frac{1}{\lambda(t)}\right),
\end{equation}
where $\mathscr{F}$ is the continuous 3D Fourier transform. 

The measured and processed data, $y_{\ell} \in \mathbb{C}^{n}$ where $n=q^{3}N_{x}N_{y}N_{t}$, is a normalized, discrete sampling of $\tilde{y}_{\ell}$ in lexicographical order.
A discrete version of~\eqref{eq:continuous-forward} can be used to relate $y_{\ell}$ to the complex reflectance image, $g_{\ell} \in \mathbb{C}^{n}$, by
\begin{equation}
    \label{eq:forward-model}
    y_{\ell} = Ag_{\ell} + \eta_{\ell},
\end{equation}
where $A$ is a linear operator that models the imaging system, and $\eta_{\ell} \sim \text{CN}(0,\sigma_\eta^2 I)$ is circularly-symmetric complex Gaussian white noise, which is the dominant noise source~\cite{spencerShotNoise}.
The operator $A$ can be written as
\begin{equation}\label{eq:A}
    A = \mathcal{D}(a) F,
\end{equation}
where $\mathcal{F}$ is the orthonormal 3D DFT, $\mathcal{D}(\cdot)$ represents a diagonal matrix with entries given by its argument, and the vector $a$ is a binary vector that encodes the aperture of the imaging system after zero-padding.

Our ultimate goal will be to recover $r\in \mathbb{R}^{n}$, the real-valued speckle-free reflectivity at each 3D voxel.
While the data $y_{\ell}$ directly relates to $g_{\ell}$ by~\eqref{eq:forward-model}, images reconstructed from estimates of $g_{\ell}$ are degraded by speckle.
To address this, we need a forward model that relates $r$ to $g_{\ell}$.

We will assume that each look produces a sample of $g_{\ell}$ that is conditionally independent given $r$.
Using a standard model of fully developed speckle~\cite{goodmanFundamentalPropertiesSpeckle1976, goodmanSpecklePhenomenaOptics2020}, we can model the conditional distribution of $g_{\ell}$ given $r$ as being a circularly-symmetric complex Gaussian random variable with conditional distribution 
\begin{equation}\label{eq:g-given-r}
    g_{\ell}|r \sim \text{CN}(0,\mathcal{D}(r)) \ ,
\end{equation}
where $\mathcal{D}(r)$ is a diagonal covariance with entries $r$.
We can then directly model the data $y_{\ell}$ in terms of $r$ by composing the forward model~\eqref{eq:forward-model} with the speckle model of~\eqref{eq:g-given-r} to yield
\begin{equation}\label{eq:y-given-r}
    y_{\ell}|r \sim \text{CN}(0, A\mathcal{D}(r)A^H + \sigma_{\eta}^2I) \ ,
\end{equation}
where $y_{\ell}$ for $\ell =1,\dots, L$ are assumed conditionally independent given $r$.

\section{CLAMP}\label{sec:clamp} 
\noindent In this section, we present the Coherent Lidar Aperture Modeled Plug-and-Play (CLAMP) algorithm for reconstructing 3D images from multi-look coherent lidar measurements.

\subsection{MACE Formulation}\label{sec:mace}
CLAMP is built on the MACE framework. For a more detailed discussion on MACE, we refer the reader to~\cite{buzzardPlugandPlayUnpluggedOptimizationFree2018}.

Our goal is to reconstruct $r$ from a collection of $L$ coherent lidar looks, $\{ y_{\ell} \}_{\ell =1}^L$.
In order to do this, MACE balances multiple agents, such as data-fitting operators or denoisers, to form a single reconstruction~\cite{buzzardPlugandPlayUnpluggedOptimizationFree2018}.
In our case, we will use $L$ data-fitting agents, $F_{\ell}$, and a single denoising agent, $H$.
Each agent operates on an image, $w_j$, and returns a new image that is improved according to the directive of the agent.

The first $L$ agents are forward agents, $F_{\ell}$, that enforce fidelity to the measurements, $y_{\ell}$, for $\ell = 1, \ldots, L$.
Specifically, $F_{\ell}$ is a proximal map given by 
\begin{equation}\label{eq:prox}
    F_{\ell}(w_{{\ell}}) = \underset{r}{\operatorname{argmin}}\left\{ f_{\ell}(r) + \frac{1}{2\sigma^2} \left\lVert r - w_{\ell} \right\rVert^2 \right\}.
\end{equation}
where $w_{{\ell}}$ is a candidate reconstruction, and $f_{\ell} (r) = - \log p\left(y_{\ell} | r \right)$ is the negative log likelihood associated with the ${\ell}^{th}$ look.
These operators have an intuitive interpretation of taking a candidate reconstruction as input and returning an image that better fits the data from look $\ell$.
The final agent is a deep neural network image denoiser that enforces a prior on the 3D image by reducing noise and speckle. 
The consensus equilibrium of these $L+1$ agents yields an image that is both consistent with all measurements and regularized in all three dimensions.
The implementation of these agents is described in Section~\ref{sec:em-agents} and~\ref{sec:cnn}.

To formulate the MACE solution, we stack the images into a single state vector, $\mathbf{w} = [w_1, \ldots, w_{L+1},]$, and stack the agents into a single operator, $\mathbf{F}$, given by
\begin{equation}
\mathbf{F}(\mathbf{w}) = \left[F_1(w_1), \ldots, F_L (w_{L}), H (w_{L+1}) \right].
\end{equation}
In order to produce a single reconstruction that is a compromise between the agents, we average the components of $\mathbf{w}$ to form a new image, $\overline{\mathbf{w}} = \frac{1}{L+1} \sum_{j=1}^{L+1} \upsilon_j w_j$ where $\upsilon_j$ sums to one.
In this work, we will use $\upsilon_j = 1/2L$ for all $j=1,\ldots, L$ and $\upsilon_{L+1} = 1/2$, which gives equal weight to the data-fitting agents and the prior agent.
Finally, the MACE solution is determined by the solution to the equilibrium equation
\begin{equation}\label{eq:mace}
    \mathbf{F}(\mathbf{w}^*) = \mathbf{G}(\mathbf{w}^*) \ ,
\end{equation}
where $\mathbf{G}(\mathbf{w}) = [\overline{\mathbf{w}}, \ldots, \overline{\mathbf{w}}]$ is a weighted average of the components of $\mathbf{w}$. Intuitively, $\mathbf{w}^*$ is the image that is the consensus solution of the agents in $\mathbf{F}$.
The final reconstruction, $r^*$, can be obtained from any single component of $\mathbf{w}^*$, however, we typically use $r^* = \overline{\mathbf{w}}^*$ for numerical robustness.

The MACE equation of~\eqref{eq:mace} is commonly solved as a fixed point of the operator $\mathbf{T}= \left( 2 \mathbf{G} - \mathbf{I} \right)\left( 2 \mathbf{F} - \mathbf{I} \right)$ by the iterations
\begin{equation}\label{eq:mann}
    \mathbf{w} \gets (1-\rho)\mathbf{w} + \rho \mathbf{T}\mathbf{w},
\end{equation}
where $\rho \in (0,1)$.
This approach (summarized in Algorithm~\ref{alg:MACE}) is guaranteed to converge to a solution of~\eqref{eq:mace} when $\mathbf{T}$ is non-expansive.

\begin{algorithm}[ht]
    \caption{MACE}\label{alg:MACE}
    \begin{algorithmic}[1]
    \State\textbf{Input}: Initialize $\mathbf{w} \in \mathbb{R}^{n(L+3)}$, $\rho \in (0,1)$
    \While{not converged}
        \State$\mathbf{r} \gets \mathbf{F}(\mathbf{w})$
        \State$\mathbf{x} \gets 2\mathbf{r} - \mathbf{w}$
        \State$\mathbf{w} \gets \mathbf{w} + 2\rho \left( \mathbf{G}( \mathbf{x}) - \mathbf{r} \right)$
    \EndWhile
    \State\textbf{Output}: $r^* = \overline{\textbf{w}}$
    \end{algorithmic}
\end{algorithm}

\subsection{EM Surrogate Forward Agents}~\label{sec:em-agents}
Unfortunately, direct application of the MACE iterations in~\eqref{eq:mann} is computationally infeasible.
The proximal map of the forward model in \eqref{eq:y-given-r} quickly becomes intractable for even moderately large images due to the non-linear relationship between $r$ and $y_{\ell}$ introduced by speckle interference.

In order to overcome this problem, we will take the approach first proposed by Pellizzari et al.~\cite{pellizzariPhaseErrorEstimation2017,pellizzariSyntheticApertureLadar2017,pellizzariCoherentPlugandPlayDigital2020,pellizzariImagingDistributedvolumeAberrations2019a, pellizzariOpticallyCoherentImage2017} and use the EM algorithm to compute surrogate functions to the exact negative log likelihood functions.
These surrogate functions are given by
\begin{equation}\label{eq:em-surrogate}
    \Hat{f}_{\ell}(r;r_{\ell}^{\prime}) = \mathbb{E}_{g\vert y_{\ell}, r_{\ell}^{\prime}} \left[ -\log p\left(y_{\ell},g\middle\vert r\right)\right],
\end{equation}
where the expectation is taken over the conditional distribution of the latent complex reflectance $g$ given $y_{\ell}$ and $r_{\ell}^{\prime}$, an approximation of $r$. Assuming~\eqref{eq:forward-model} and~\eqref{eq:g-given-r}, and using $y_\ell|g,r = y_\ell|g$,~\cite{pellizzariPhaseErrorEstimation2017} shows that this conditional distribution of $g$ given $y_{\ell}, r_{\ell}^{\prime}$ is Gaussian with the form
\begin{align} 
\label{eq:ConditionalDisGGivenYR}
p(g \vert y_{\ell}, r_{\ell}^{\prime}) = \frac{1}{Z} 
\exp \left\{ - ( g  - \mu_{\ell} )^H \tilde{C}_{\ell}^{-1} ( g  - \mu_{\ell} ) \right\},
\end{align}
where $Z$ is a normalizing constant and $\mu_{\ell}$ and $\tilde{C}_{\ell}$ are the mean and covariance given by
\begin{align}
\label{eq:mu2}
\mu_{\ell} &= \underset{g\in \mathbb{C}^n}{\operatorname{argmin}}\left\{  \frac{1}{2\sigma_w^{2}}\lVert y_{\ell}-Ag \rVert^2 + \frac{1}{2} g^H \mathcal{D}\left(\frac{1}{r_{\ell}^{\prime}}\right) g \right\}, \\
\label{eq:Bi}
\tilde{C}_{\ell} &= \left[ \frac{1}{\sigma_{w}^{2}} A^{H}A + \mathcal{D}\left(\frac{1}{r_{\ell}^{\prime}}\right) \right]^{-1}.
\end{align}
Computing the expectation in~\eqref{eq:em-surrogate} then gives the surrogate function as
\begin{equation}\label{eq:em-surrogate-final}
    \hat{f}(r; \mu_{\ell}, \tilde{C}_{\ell}) = \sum_{j=1}^{n} \left\{ \log\left( r_j \right) + \frac{\left\lvert \mu_{\ell,j} \right\rvert^2 + \tilde{c}_{\ell,j}}{r_j}\right\},
\end{equation}
where $j$ is the voxel index, $\mu_{\ell,j}$ is the $j$th element of $\mu_{\ell}$, and $\tilde{c}_{\ell,j}$ is the $j$th diagonal element of $\tilde{C}_{\ell}$.

However, using this surrogate function is also intractable for large images since computing the covariance, $\tilde{C}_{\ell}$, requires a large matrix inversion, while computing the mean, $\mu_{\ell}$, of~\eqref{eq:mu2} directly requires the solution of a large quadratic optimization problem.
We address each of these issues separately.

First, to improve the tractability of computing $\tilde{C}_{\ell}$, we make the approximation that $A^HA = \alpha I$, where $\alpha = \lVert a \rVert_1 / n$ has the interpretation of the fraction of light passing through the aperture.
In this case, $\tilde{C}_{\ell}$ is replaced by diagonal matrix $C_{\ell}$ with diagonal entries $c_{\ell, j}$, given by
\begin{equation}
    c_{\ell, j} = \frac{\sigma_w^2 r_{\ell, j}'}{\alpha r_{\ell, j} + \sigma_w^2}.
\end{equation}
Alternatively, if a more accurate approximation is needed, one could use the method used in~\cite{tuckerSpeckleSuppressionMultiChannel2020} at the cost of more computation.
Importantly, we only make this orthogonality assumption in approximating $\tilde{C}_{\ell}$, not throughout the entire forward model.

Next, to evaluate~\eqref{eq:mu2} efficiently, we solve the minimization problem iteratively.
Rather than fully solving it at each CLAMP iteration, we compute only a single gradient descent step.
As CLAMP progresses, these incremental steps gradually minimize~\eqref{eq:mu2} over multiple iterations.

To increase stability and avoid division by zero in the objective in~\eqref{eq:mu2}, we add the noise floor $\sigma_w^2 / \alpha$ to $r_\ell^\prime$.
This yields the regularized objective 
\begin{equation}\label{eq:mu-objective}
    h (g; y_{\ell}, r_{\ell}^\prime ) = \frac{1}{2\sigma_w^{2}}\lVert y_{\ell}-Ag \rVert^2 + \frac{1}{2} g^H \mathcal{D}\left(\frac{1}{r_{\ell}^{\prime} + \sigma_w^2/\alpha}\right) g.
\end{equation}
The gradient step is then given by
\begin{align}
d &\gets - \nabla h(\mu_{\ell}; y_{\ell}, r_{\ell}^\prime ) \label{eq:h-gradient} \\
\mu_{\ell} &\gets \mu_{\ell} + \gamma_* d \label{eq:mu-update},
\end{align}
where $\gamma_*$ is the optimal step-size and is given by
\begin{equation}
    \gamma_* = \frac{(Ad)^H(y_{\ell}-A\mu_{\ell}) - \sigma_w^2 d^H\mathcal{D}(1/r_{\ell}')\mu_{\ell}}{d^H\left(A^HAd + \sigma_w^2 \mathcal{D}(1/r_{\ell}')d\right)}.
\end{equation}

Finally, using these forms for $\mu_{\ell,j}$ and $c_{\ell,j}$, the EM surrogate of~\eqref{eq:em-surrogate} can be calculated~\cite{pellizzariPhaseErrorEstimation2017} to be
\begin{align}
\label{eq:EM-surrogate}
\Hat{f}_{\ell}(r; \mu_{\ell}, c_{\ell} ) &=
\sum_{j=1}^n \left\{ \log r_{j} + \frac{ \vert \mu_{\ell,j} \vert^2 + c_{\ell,j}}{ r_{j} } \right\} .
\end{align}
From this, the CLAMP forward agent is the proximal map 
\begin{equation}
\label{eq:surrogate-prox}
\Hat{F_{\ell}}(v; \mu_{\ell}, c_{\ell} ) = \underset{r}{\operatorname{argmin}}\left\{ \Hat{f_{\ell}}(r; \mu_{\ell}, c_{\ell}) + \frac{1}{2\sigma^2} \left\lVert r - v \right\rVert^2 \right\},
\end{equation}
which can be computed in closed form by solving for a root of a cubic equation for each voxel in the image.
However, in practice, to further simplify computation and speed up implementation, we employ an additional quadratic surrogate of $\hat{f}_{\ell}$.
The details of which are given in Appendix~\ref{ap:quadratic-surrogate}.

\subsection{3D Deep Prior Agent}\label{sec:cnn}
\noindent 

\begin{figure}
    \centering
    \includegraphics[width=0.99\columnwidth]{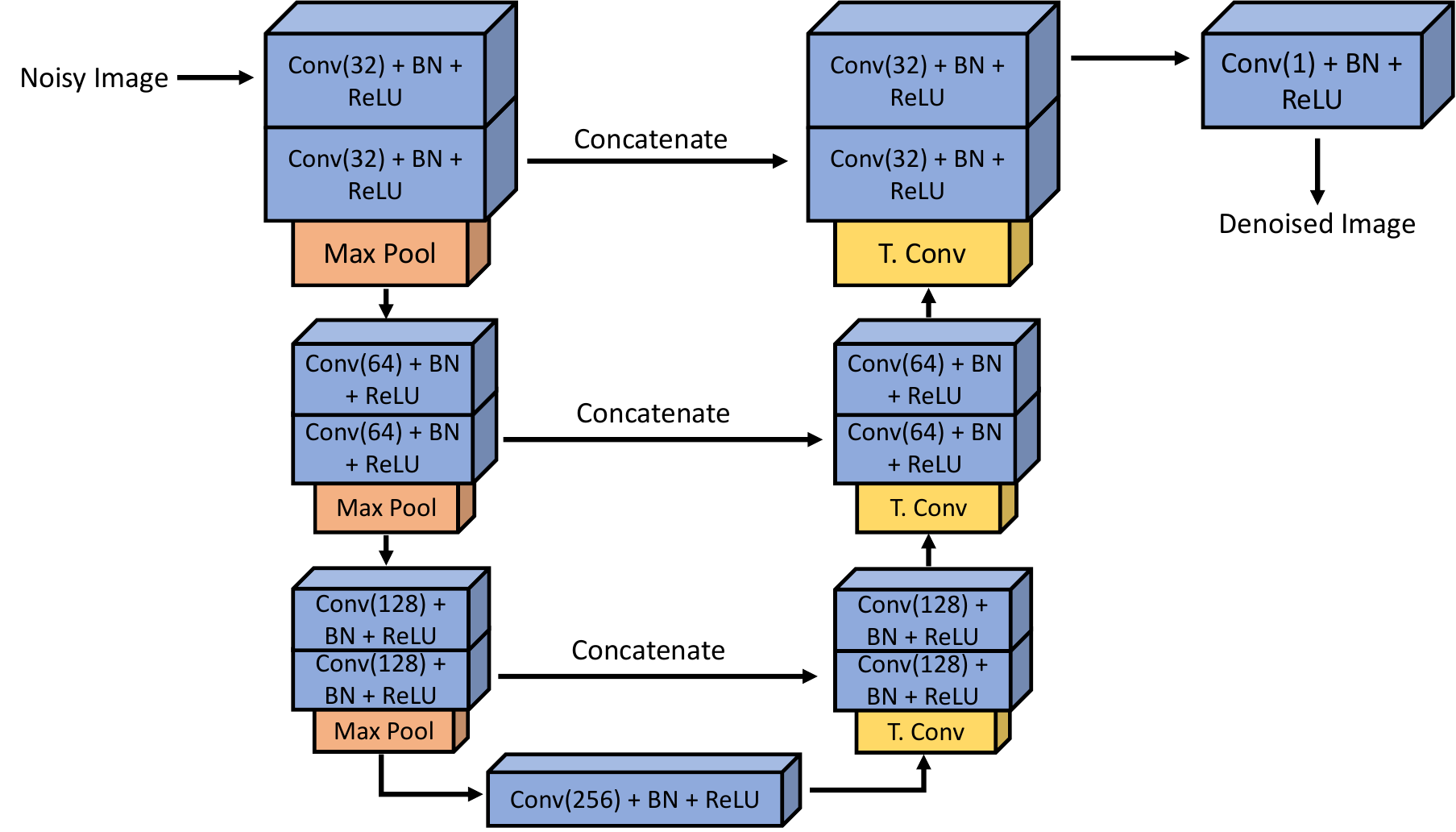}
    \caption{The architecture of the 3D UNet used as a prior agent. The number of output channels for the convolutional layers is shown in parentheses.}\label{fig:unet3d}
\end{figure}
In order to regularize the image, we use a 3D U-Net image denoiser similar to the one used in~\cite{mimo3d} for 3D radar imaging.
This model was chosen because it was shown to be effective in jointly exploiting correlations in cross-range and range dimensions.

The architecture of the 3D U-Net is shown in Figure~\ref{fig:unet3d}.
It consists of three encoding blocks, each with two convolutional layers, followed by three decoding blocks, each with two convolutional layers.
Each convolutional block is followed by a batch normalization (BN) layer and a ReLU activation function and a max pooling to reduce the dimensions of feature tensors by a factor of two, whereas the decoding stages employ a transposed convolution to upsample the feature tensor by a factor of two.
A final 3D convolutional block is used to produce the denoised image.
 
The network was trained to denoise 3D images from ShapeNet~\cite{changShapeNetInformationRich3D2015} with 10\% additive white Gaussian noise. The training set consisted of  800 models rendered on a $256 \times 256 \times 256$ grid, and an additional 200 models were used for validation.
The images were patched into $32 \times 32 \times 32$ patches in order to fit into the GPU memory.
We used the Adam optimizer with an exponentially decaying learning rate and the MSE loss function.
The initial learning rate was set to $10^{-2}$ and decayed by a factor of 0.975 every 300 training batches.
The batch size was set to 24 and the network was trained for 100 epochs.
The network weights were saved when the validation loss was lowest during training.

\begin{algorithm}[th]
    \caption{CLAMP}\label{alg:EM-MACE}\label{alg:M-MACE-ap}\label{alg:clamp}
    \begin{algorithmic}[1]
    \State\textbf{Input}: $y_{\ell}\in \mathbb{C}^n$ for $\ell =1, \ldots, L$
    \State{Initialize: $\mu_{\ell} = A^H y_{\ell}$, $r_{\ell} = \frac{1}{L} \sum_{\ell=1}^{L}|A^H y_{\ell}|^2$, $w_{\ell} = r_{\ell}$, for $\ell = 1, \ldots, L$}
    \While{not converged}
    \For{$\ell = 1,\dots,L$} c
    \State{$\forall j$, $c_{\ell,j} \gets \frac{\sigma_w^2 r_{\ell,j}}{\alpha r_{\ell,j} + \sigma_w^2}$}  \algorithmiccomment{Update $C_{\ell}$}
    \State{$d \gets - \nabla h(\mu_{\ell}; y_{\ell}, r_{\ell} )$} \algorithmiccomment{Update $\mu_{\ell}$}
    \State{$\gamma \gets \argmin_{\gamma} \left\{ h(\mu_{\ell} - \gamma d ; y_{\ell}, r_{\ell} ) \right\}$}
    \State{$\mu_{\ell} \gets \mu_{\ell} + \gamma d$}
    \State{$r_{\ell} \gets \hat{F}_{\ell}\left(w_{\ell}; \mu_{\ell}, c_{\ell} \right)$}  \algorithmiccomment{Update $r_{\ell}$ by~\eqref{eq:surrogate-prox}}
    \EndFor
    \State{$r_{L+1} \gets \textbf{H}(w_{L+1})$} \algorithmiccomment{Apply prior agent}
    \State{$\mathbf{r} \gets \left[ r_1, \ldots, r_{L+3} \right]$} \algorithmiccomment{Do MACE iteration}
    \State$\mathbf{x} \gets 2\mathbf{r} - \mathbf{w}$
    \State$\mathbf{w} \gets \mathbf{w}+ 2\rho \left( \mathbf{G}(\mathbf{x}) - \mathbf{r} \right)$
\EndWhile
    \State\textbf{Output}: $r^* = \overline{\textbf{w}}$
    \end{algorithmic}
\end{algorithm}

\subsection{Summary of CLAMP}\label{sec:clamp-summary}
\noindent 
A summary of CLAMP is given in Algorithm~\ref{alg:clamp}.
CLAMP begins with an initialization of the image 
In steps 4--9, the each forward agent is approximated by the EM surrogate, and the image is updated using the proximal map of the surrogate function to enforce fidelity to the measurements, $y_{\ell}$.
In step 11, the deep prior agent are applied to the image to regularize the image in all three dimensions.
Finally, in steps 12--14, a consensus among the agents is enforced, resulting in an image that is both consistent with the measurements and regularized in all three dimensions.

A key benefit of CLAMP is its computational efficiency compared to a direct application of MACE.
CLAMP achieves this efficiency by using EM surrogates, given in~\eqref{eq:em-surrogate-final}, in replacement of the exact negative log likelihood functions in the forward agents.
However, these surrogate approximations alone are not enough to make the algorithm computationally feasible.
The computational complexity of MACE with EM surrogates is dominated by the computation of the diagonal of the covariance matrix, $\tilde{C}_{\ell}$, in~\eqref{eq:Bi}. 
Solving for this using direct methods has a complexity of $O(n^3)$~\cite{golub}, which is generally infeasible.
Similarly, solving the optimization problem in~\eqref{eq:mu2} requires the same $O(n^3)$ operation to compute $\mu_{\ell}$ exactly.

In contrast, CLAMP uses the diagonal approximation of the covariance matrix, $C_{\ell}$, and an iterative gradient step to compute the mean, $\mu_{\ell}$.
These steps, steps 5--8 in Algorithm~\ref{alg:clamp}, are dominated by the computation of the gradient in~\eqref{eq:h-gradient} for each of the $L$ looks.
Since this gradient can be computed efficiently using the FFT, the total complexity of each iteration of the algorithm is $O(n\log(n))$~\cite{fft}.
This makes CLAMP computationally feasible for large images.
To give insight into the amount of computation required, the experiments in Section~\ref{sec:results} were run on a 3D image of size $128^3$ on a single NVIDIA Quadro RTX 8000 GPU, and each iteration took less than 0.5 seconds.

\section{Majorized MACE Theory}\label{sec:M-MACE}
\noindent A crucial aspect of CLAMP is the use of a surrogate function through the EM algorithm.
In this section, we present a generalized theory of majorization-minimization within the MACE framework with guaranteed convergence to an exact MACE solution as defined in~\eqref{eq:mace}.

Although our theory is limited to the case where all agents are proximal maps, it applies in practice to many widely used denoisers.
Many PnP algorithms (including CLAMP) use denoising agents that are not necessarily proximal maps.
However, typical neural network denoisers are trained to be minimum mean squared error (MMSE) denoisers, which were shown to be proximal maps in~\cite{gribonvalShouldPenalizedLeast2011}.
This equivalence was also demonstrated experimentally in~\cite{xuProvableConvergencePlugandPlay2020}, in which the authors show that the convergence of a PnP algorithm with an exact MMSE denoiser agrees remarkably well with the convergence of the same algorithm with a neural network denoiser.
This highlights the practical relevance of our theory.

In our setting, we modify the operator $F_i$ by approximating its objective function, $f_i$, with a surrogate function, $\Hat{f_i}$, whose proximal map, $\Hat{F_i}$, can be computed more efficiently.
We make a few standard assumptions on the surrogate functions, given in Definition~\ref{def:surrogate}, which are adapted from those made in~\cite{mairalOptimizationFirstOrderSurrogate, Fessler-MM, sun-palomar-mm}.
Commonly used surrogates, such as the Jensen surrogate used in the EM algorithm or quadratic approximations of twice differentiable functions meet these requirements.

\begin{defn}[Surrogate]\label{def:surrogate}
    Let $\Theta \subset \mathbb{R}^n$ be convex, and $f: \mathbb{R}^{n} \to \mathbb{R}$ be a convex function.
    A function $\hat{f}: \mathbb{R}^{n} \to \mathbb{R}$ is a \textbf{surrogate} of $f$ near $\xi \in\Theta$ when the following conditions hold:
    \begin{itemize}
    \item 
    \textbf{Majorization}: we have $\hat{f}(x) \geq f(x)$ for all $x \in \mathbb{R}^n$.
    \item 
    \textbf{Smoothness}: the approximation error $e \triangleq \hat{f}-f$ is differentiable, and its gradient is $L$-Lipschitz continuous. 
    Moreover, we have that $e(\xi)=0$ and $\nabla e(\xi) = 0$,
    \item \textbf{Strong convexity}: $\hat{f}$ is $p$-strongly convex with $p \geq L$.
\end{itemize}
When the point $\xi$ is relevant, we write the surrogate as $\hat{f}(x; \xi)$, and we denote the set of such surrogates as $\mathcal{S}_{L,p}(f,\xi)$.
\end{defn}
We do not require the objective function, or its surrogate, to be differentiable, nor do we require the surrogate function to be continuous as a function of $\xi$.
These are common assumptions in the literature, but are not necessary for our theory.

Majorization-minimization schemes work by alternately minimizing the surrogate function and then updating the surrogate itself, which gives updates of the form 
\begin{equation}
    r^{(k+1)} \gets \underset{r}{\operatorname{argmin}} \left\{ \Hat{f}\left(r; r^{(k)} \right) \right\}.
\end{equation}
The conditions in Definition~\ref{def:surrogate} ensure that the original objective is decreasing, $f\left(r^{(k+1)}\right) \leq f\left(r^{(k)}\right)$, and that the iterates $r^{(k)}$ converge to a minimizer of $f$.

Algorithm~\ref{alg:M-MACE} applies these principles within the MACE framework. 
Assuming we have $N$ agents, on the $k$-th iteration, we compute a surrogate $\Hat{f_i}$ of the objective function $f_i$ at $r_i^{(k)}$ for $i=1,\dots,N$. 
We define the concatenated operators
\begin{equation}\label{eq:stacked-majorized-prox}
\Hat{\mathbf{F}}^{(k)}\left(\mathbf{w}; \mathbf{r}^{(k)}\right) = \left[ \Hat{F}_1^{(k)}(w_1; r_1^{(k)}),  \dots,  \Hat{F}_N^{(k)}(w_N; r_N^{(k)}) \right]
\end{equation}
and
\begin{equation}\label{eq:That}
    \Hat{\mathbf{T}}^{(k)} = \left( 2\mathbf{G}-I \right) \left( 2\Hat{\mathbf{F}}^{(k)}\left(\cdot\, ; \mathbf{r}^{(k)}\right ) -I \right),
\end{equation}
where $\hat{F}_i^{(k)}$ is the proximal map of $\hat{f}_i^{(k)}$.
Akin to the fixed point iterations in~\eqref{eq:mann}, we analyze the fixed point of the system 

\begin{align}\label{eq:aug-mann}
    \mathbf{w}^{(k+1)} &= \rho \Hat{\mathbf{T}}^{(k)} \mathbf{w}^{(k)} + (1-\rho) \mathbf{w}^{(k)}, \\
    \mathbf{r}^{(k+1)} &= \Hat{\mathbf{F}}^{(k)}(\mathbf{w}^{(k)}; \mathbf{r}^{(k)}). \nonumber
\end{align}

Our main result is Theorem~\ref{THM:MAJORIZATION}, which shows that Algorithm~\ref{alg:M-MACE} converges to an exact MACE solution.
Part~\ref{THM:PART1}, states that any fixed point $(\mathbf{w}^*, \mathbf{r}^*)$ of~\eqref{eq:aug-mann} is an exact MACE solution as in~\eqref{eq:mace}.
Part~\ref{THM:PART2} states that convergence to a fixed point is guaranteed if the surrogate is strongly convex.  Note that if $\mathbf{r}^{(k+1)}= \mathbf{r}^{(k)}$, then $\Hat{\mathbf{F}}^{(k+1)} = \Hat{\mathbf{F}}^{(k)}$ and $\Hat{\mathbf{T}}^{(k+1)} = \Hat{\mathbf{T}}^{(k)}$, so the functions are also fixed at such a fixed point.

\begin{theorem}[Majorized-MACE Solution]\label{THM:MAJORIZATION} 
    Let $f_i: \mathbb{R}^{n} \to \mathbb{R}, i = 1, \dots, N$, be convex, and $\Hat{f_i}^{(k)}\in \mathcal{S}_{L_i, p_i} \left( f_i, r_i^{(k)}  \right)$. Let $F_i$ and $\Hat{F_i}^{(k)}$ be the proximal maps of $f_i$ and $\Hat{f_i}^{(k)}$, respectively, and $\Hat{\mathbf{T}}^{(k)}$ be the operator defined in~\eqref{eq:That}. Then:
    \begin{enumerate}[label=(\roman*)]
        \item~\label{THM:PART1} Any fixed point of the augmented system~\eqref{eq:aug-mann} is a solution to the original MACE equation given in~\eqref{eq:mace}.
        \item~\label{THM:PART2} The iterates defined by~\eqref{eq:aug-mann} converge to a fixed point, and hence a MACE solution, if one exists.
    \end{enumerate}
\end{theorem}

\begin{proof}
    Proof is in Appendix~\ref{ap:majorization}.
\end{proof}

\begin{algorithm}[t]
    \caption{Majorized-MACE}\label{alg:M-MACE}
    \begin{algorithmic}[1]
        \State\textbf{Input}: Initialize 
        $\mathbf{w}^{(0)}$, $\mathbf{r}^{(0)} \in \mathbb{R}^{n(L+3)}, \rho \in (0,1)$
        \State{$k=0$}
        \While{not converged}
            \For{$i=1,\dots,N$}
            \State{Compute surrogate $\Hat{f_i}^{(k)} \in \mathcal{S}_{L_i,p_i} \left( f_i, r_i^{(k)}  \right)$}
            \EndFor
            \State$\mathbf{r}^{(k+1)} = \Hat{\mathbf{F}}^{(k)}(\mathbf{w}^{(k)}; \mathbf{r}^{(k)})$
            \State$\mathbf{x} = 2\mathbf{r}^{(k+1)} - \mathbf{w}^{(k)}$
            \State$\mathbf{w}^{(k+1)} = \mathbf{w}^{(k)} + 2\rho \left( \mathbf{G}(\mathbf{x}) - \mathbf{r}^{(k+1)} \right)$
        \State{$k = k+1$}
        \EndWhile
    \end{algorithmic}
\end{algorithm}

\section{Results}\label{sec:results}
\noindent We now demonstrate the effectiveness of CLAMP at reconstructing high-resolution images from multi-look coherent lidar data.
We perform a series of experiments on both synthetic and experimental data to compare the performance of CLAMP with the standard speckle average method, as well as using sparsity-based regularization in the form of $\ell_{2,1}$-regularization ($\ell_2$-regularization in cross-range and $\ell_1$-regularization in range), and isotropic total variation.
These methods are implemented as described in Algorithm~\ref{alg:clamp} with their corresponding regularizing agents in place of the proposed deep prior agent.

Further, we perform multiple experiments with different zero-padding factors, $q$, to investigate the effectiveness of CLAMP at different sampling rates. Specifically, we use $q=1, 1.5,$ and $2$. This results in images of various sizes, with the number of voxels in the image being $q^3$ times the number of measurements. 

We also perform an ablation study of the aperture model in CLAMP to demonstrate its importance in achieving high-resolution reconstructions.
In this case, we use the algorithm as described in Algorithm~\ref{alg:clamp}, but with $\mathbf{a} = \mathbf{1}$, effectively removing the aperture model from the reconstruction process.
The resulting algorithm differs only in the $\mu$-update (steps 6-8 of Algorithm~\ref{alg:clamp}), which can be simplified to be a closed-form solution of~\eqref{eq:mu2}.

\subsection{Methods}\label{sec:Methods}
For each reconstruction presented in this section, we ran the corresponding algorithm for 250 iterations with $\rho = 0.5$.
Each algorithm was initialized with $r = \frac{1}{L}\sum_{\ell=1}^L \lvert A^H y_{\ell}\rvert^2$ and setting $\mathbf{w}$, $\mathbf{r}$, to be stacked copies of $r$.
Each $\mu_{\ell}$ was initialized as $\frac{1}{\alpha}A^Hy_i$.
The forward model proximal parameter, $\sigma^2$, was empirically chosen in a range of $[0.001, 1.0]$ to produce the best image quality.

To measure convergence of CLAMP, we define the convergence error based on the MACE equation~\eqref{eq:mace} as
\begin{equation}\label{eq:mace-convergence-error}
  \text{Convergence Error} = \frac{\left\lVert \Hat{\mathbf{F}}(\mathbf{w}) - \mathbf{G}(\mathbf{w}) \right\rVert}{\left\lVert \mathbf{G}(\mathbf{w}) \right\rVert}.
\end{equation}
Additionally, to show convergence of the $\mu_\ell$ updates, we define the relative residual averaged across looks as
\begin{equation}\label{eq:residual}
  \mu\text{-Residual} = \frac{1}{L} \sum_{\ell=1}^L \frac{ \left\lVert \tilde{C}_\ell^{-1} \mu_\ell - \frac{1}{\sigma_w^2}A^Hy_\ell \right\rVert}{\left\lVert \frac{1}{\sigma_w^2} A^Hy_\ell \right\rVert},
\end{equation}
where the numerator comes from the exact solution of~\eqref{eq:mu2}. 
Plots of CLAMP's convergence behavior are shown in Section~\ref{sec:experimental}.

In the synthetic data experiments, we compare the reconstructed images to a ground truth image.
We use peak-signal-to-noise ratio (PSNR), computed as
\begin{equation}\label{eq:psnr}
  \text{PSNR}(\beta^* \hat{r}, r) = 10\log_{10}\left(\frac{n}{\lVert \beta^* \hat{r} - r \rVert^2}\right),
\end{equation}
to compare the reconstructions, $\hat{r}$, at $q=2$ to the ground truth image, $r$.
This was only computed for the $q=2$ reconstructions, for that is when $r$ and $\hat{r}$ are the same size ($128 \times 128 \times 128$).
The factor $\beta^*$ is a multiplicative factor that accounts for any scaling differences between the reconstructed images and the ground truth image, and is computed as
$\beta^* = {\operatorname{argmax}}_{\beta}\left\{\text{PSNR}(\beta\hat{r}, r) \right\} = \hat{r}^T r / ||\hat{r}||^2$.
This will allow for a fair comparison of the reconstructions --- any differences in PSNR will be due to the reconstruction method and not to a scaling factor.

In sparse 3D imaging, metrics on full 3D images can be misleading if they are dominated by small volumetric differences that otherwise have little impact on the image quality.
To address this, we employ two point-cloud metrics to quantify the accuracy of the reconstruction of the surface reflectivity.
We convert a 3D image, $r$, to a point cloud, $P\subseteq \mathbb{R}^{4}$, such that each point $p = (\textbf{s}_p, r_p) \in P$ corresponds to a voxel in $r$ with a value greater than a given threshold. In this work, we use the natural threshold, the noise floor, $\sigma_w^2/\alpha$.
The first component, $\textbf{s}_p\in\mathbb{R}^3$, of $p$ represents the spatial coordinates of the point, and $r_p$ represents the reflectivity of that point. 

We quantify the spatial accuracy of the reconstructions by computing the Euclidean distance between points in the reconstruction and points in the ground truth point cloud.
For each reconstruction, $\hat{r}$, we transform the image into a point cloud $P = \left\{(\textbf{s}_p, \hat{r}_p)\right\}$. 
For each point in $P$, we compute the Euclidean distance to its closest point in the ground truth point cloud, $Q = \left\{(\textbf{t}_q, r_q)\right\}$. Let $\tilde{n}(p) = \argmin_{q\in Q} ||\textbf{s}_p - \textbf{t}_q||$ be the point in $Q$ nearest to $p$, which, in practice, we determine by a $k$-d tree search~\cite{friedman-kdtree}.
The average Euclidean distance is then given by
\begin{equation}\label{eq:pc-distance}
\text{Euclidean Distance} = \frac{1}{|P|} \sum_{p \in P} ||\textbf{s}_p - \textbf{t}_{\tilde{n}(p)}||.
\end{equation}
To avoid corrupting this metric with outliers (points in the reconstruction with no nearby neighbor in the ground truth), we remove any point with a distance to the ground truth more then 3 times the Rayleigh criterion for resolution (approximately $\qty{1.5}{} \si{cm}$). Additionally, we report the false positive rate as the proportion of removed points relative to the total number of reconstructed points.

Using this method, we can also compute the normalized root mean squared error (NRMSE) of a reconstruction as
\begin{equation}\label{eq:nrmse}
  \text{NRMSE}(\beta^*\hat{r}, r) = \sqrt{\frac{\sum_{p\in P} (\beta^*\hat{r}_p - r_{\tilde{n}(p)})^2}{\sum_{p \in P} |r_{\tilde{n}(p)}|^2}},
\end{equation}
where $\beta^*$ is computed as $
  \beta^* = {\operatorname{argmin}}_{\beta}\left\{\text{NRMSE}(\beta\hat{r}, r) \right\}.$

\subsection{Synthetic Data Generation}
In the synthetic data experiment, we generated data by simulating the multi-look coherent lidar imaging process described in Section~\ref{sec:forward}.

\begin{figure*}[ht]
  \centering
  \begin{subfigure}{\textwidth}
      \centering
      \includegraphics[width=0.5\textwidth]{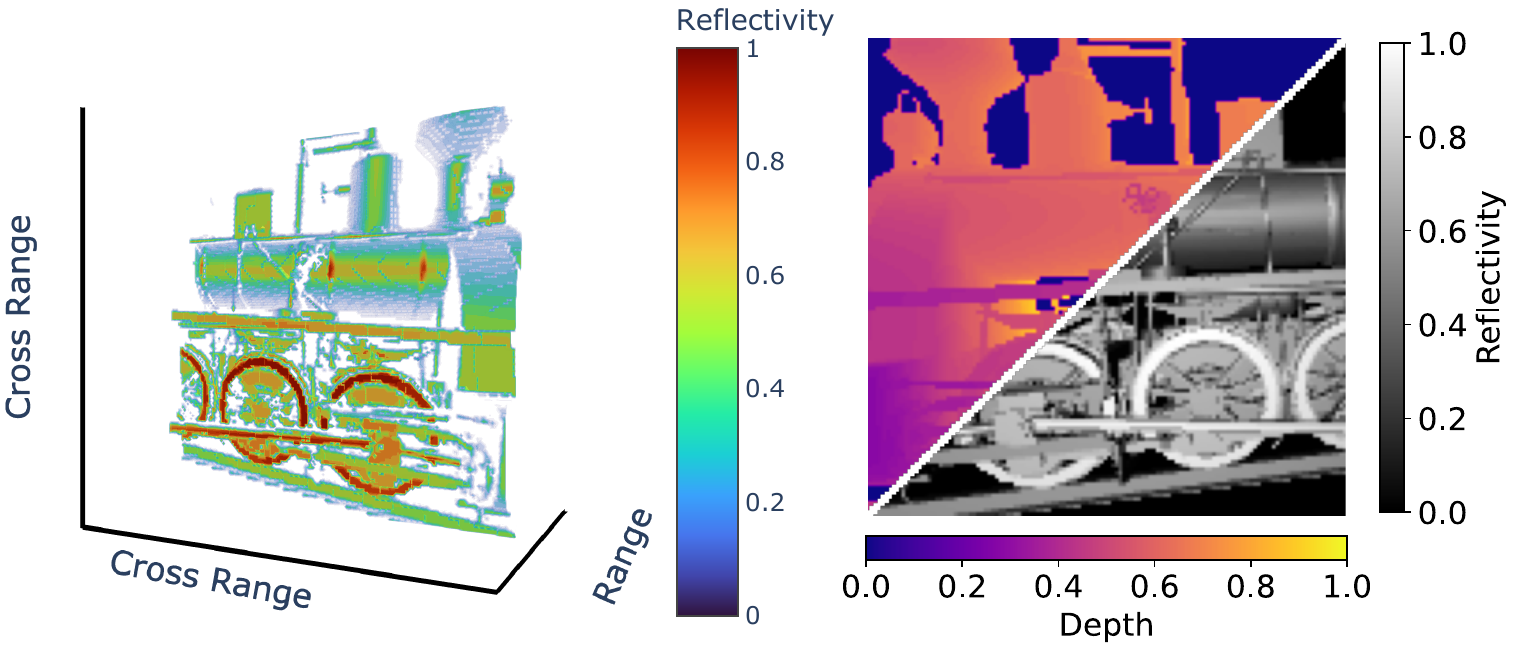} 
      \caption{Ground Truth}\label{fig:loco-truth}
  \end{subfigure}

  \begin{minipage}{0.05\textwidth}
    \vspace{-0.4cm}
    \hspace{0.1cm}
    \rotatebox{90}{\large{$q=2$} \hspace{0.7cm} \large{$q=1.5$} \hspace{0.7cm} \large{$q=1$}}
  \end{minipage}%
  \begin{minipage}{0.9\textwidth}
    \begin{subfigure}{0.22\textwidth}
      \centering
      \begin{minipage}{0.48\textwidth}
          \includegraphics[width=\textwidth]{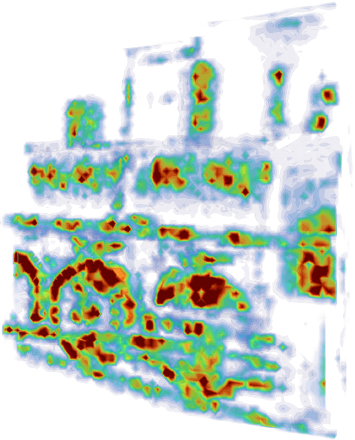}
      \end{minipage}
      \begin{minipage}{0.48\textwidth}
          \includegraphics[width=\textwidth]{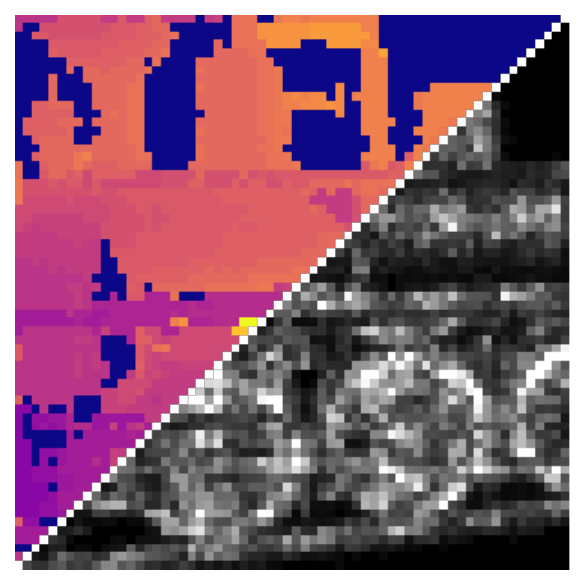}
      \end{minipage}
    \end{subfigure}
    \begin{subfigure}{0.22\textwidth}
      \centering
      \begin{minipage}{0.48\textwidth}
          \includegraphics[width=\textwidth]{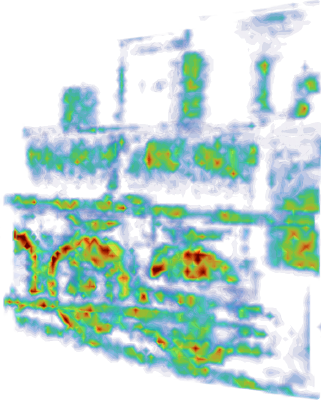}
      \end{minipage}
      \begin{minipage}{0.48\textwidth}
          \includegraphics[width=\textwidth]{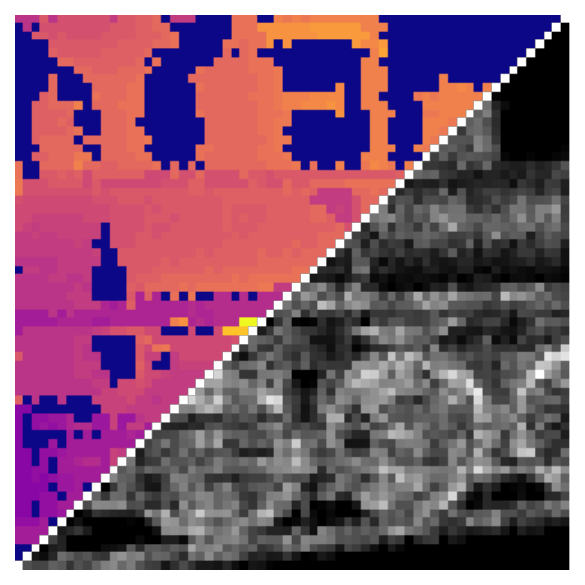}
      \end{minipage}
    \end{subfigure}
    \begin{subfigure}{0.22\textwidth}
      \centering
      \begin{minipage}{0.48\textwidth}
          \includegraphics[width=\textwidth]{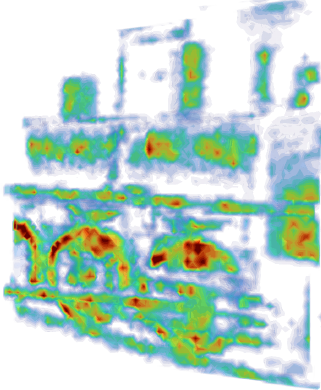}
      \end{minipage}
      \begin{minipage}{0.48\textwidth}
          \includegraphics[width=\textwidth]{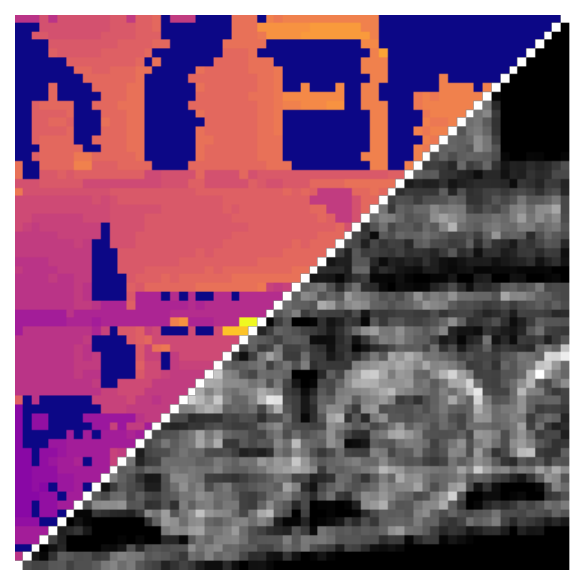}
      \end{minipage}
    \end{subfigure}
    \begin{subfigure}{0.22\textwidth}
      \centering
      \begin{minipage}{0.48\textwidth}
          \includegraphics[width=\textwidth]{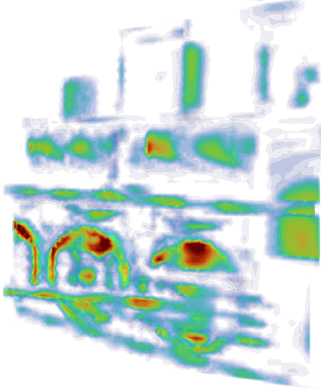}
      \end{minipage}
      \begin{minipage}{0.48\textwidth}
          \includegraphics[width=\textwidth]{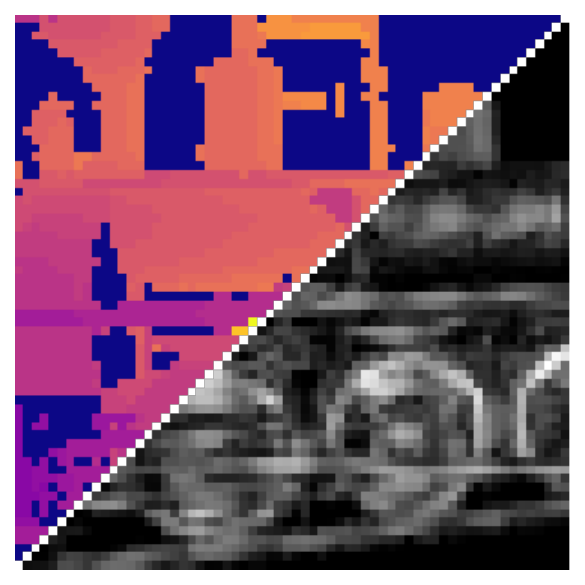}
      \end{minipage}
    \end{subfigure}
    \\
    \begin{subfigure}{0.22\textwidth}
      \centering
      \begin{minipage}{0.48\textwidth}
          \includegraphics[width=\textwidth]{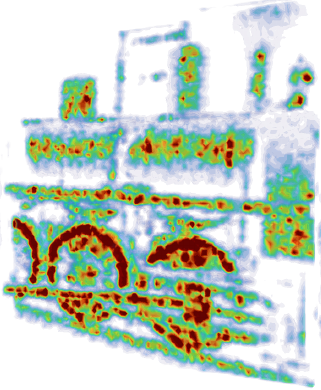}
      \end{minipage}
      \begin{minipage}{0.48\textwidth}
          \includegraphics[width=\textwidth]{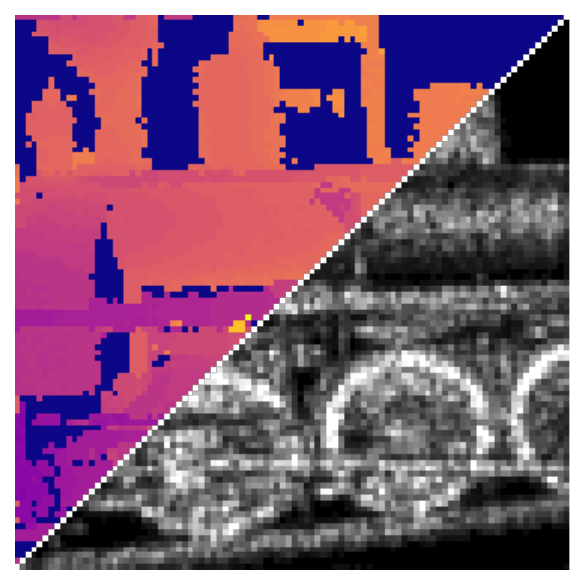}
      \end{minipage}
    \end{subfigure}
    \begin{subfigure}{0.22\textwidth}
      \centering
      \begin{minipage}{0.48\textwidth}
          \includegraphics[width=\textwidth]{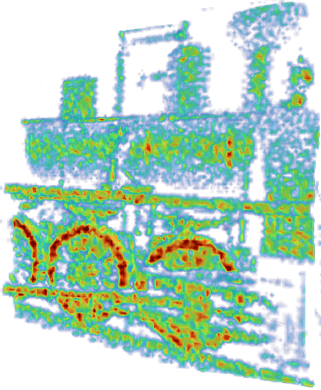}
      \end{minipage}
      \begin{minipage}{0.48\textwidth}
          \includegraphics[width=\textwidth]{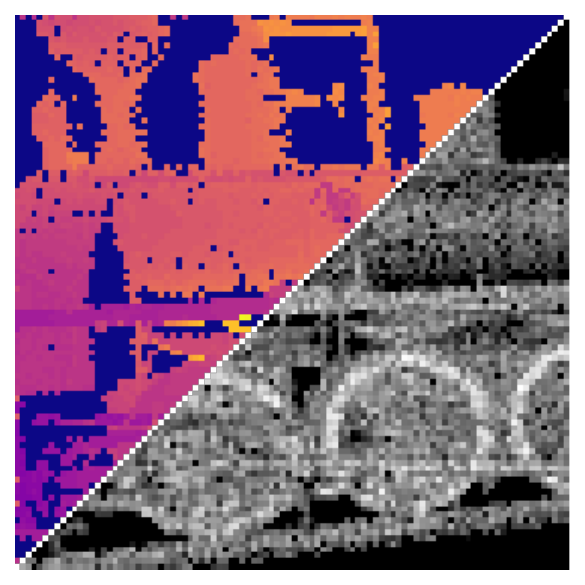}
      \end{minipage}
    \end{subfigure}
    \begin{subfigure}{0.22\textwidth}
      \centering
      \begin{minipage}{0.48\textwidth}
          \includegraphics[width=\textwidth]{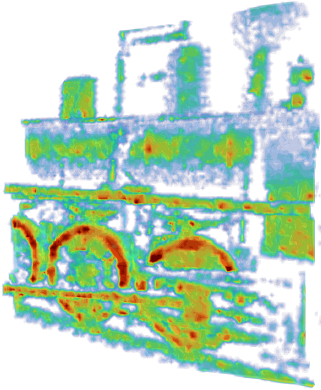}
      \end{minipage}
      \begin{minipage}{0.48\textwidth}
          \includegraphics[width=\textwidth]{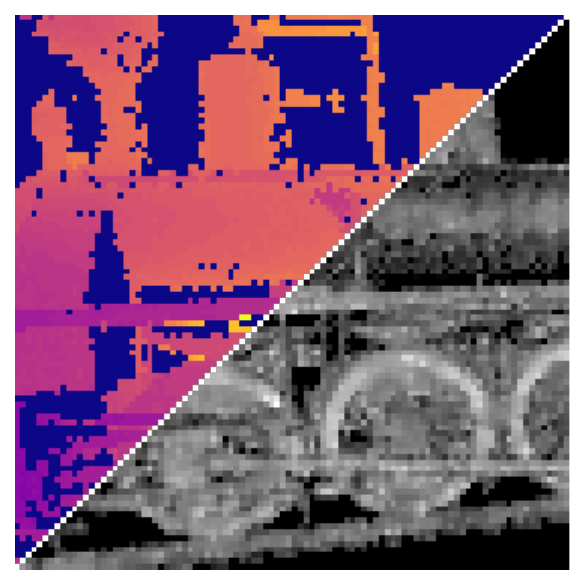}
      \end{minipage}
    \end{subfigure}
    \begin{subfigure}{0.22\textwidth}
      \centering
      \begin{minipage}{0.48\textwidth}
          \includegraphics[width=\textwidth]{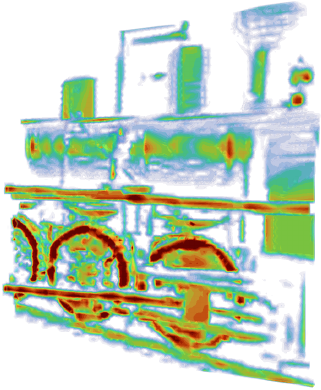}
      \end{minipage}
      \begin{minipage}{0.48\textwidth}
          \includegraphics[width=\textwidth]{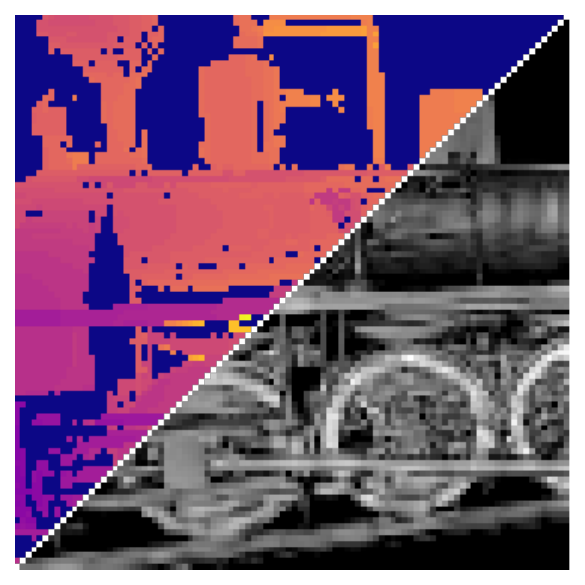}
      \end{minipage}
    \end{subfigure}
    \\
    \begin{subfigure}{0.22\textwidth}
      \centering
      \begin{minipage}{0.48\textwidth}
          \includegraphics[width=\textwidth]{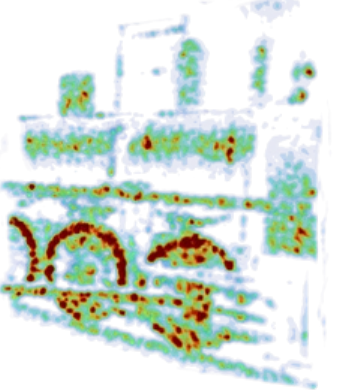}
      \end{minipage}
      \begin{minipage}{0.48\textwidth}
          \includegraphics[width=\textwidth]{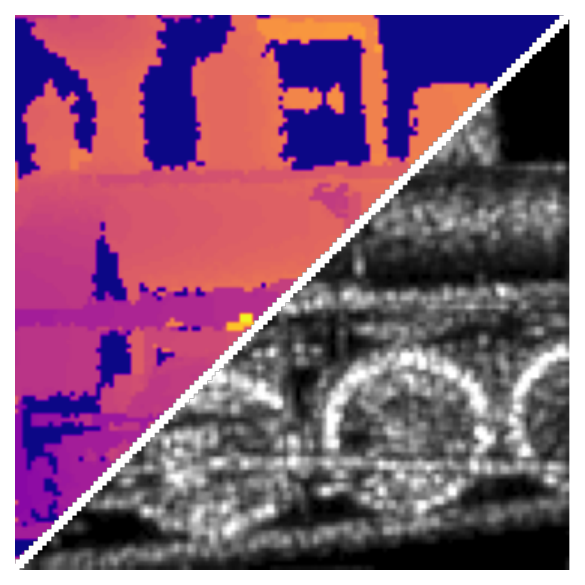}
      \end{minipage}
      \caption{Speckle Average}\label{fig:train-sa}
    \end{subfigure}
    \begin{subfigure}{0.22\textwidth}
      \centering
      \begin{minipage}{0.48\textwidth}
          \includegraphics[width=\textwidth]{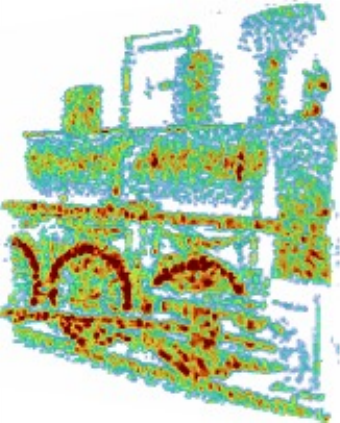}
      \end{minipage}
      \begin{minipage}{0.48\textwidth}
          \includegraphics[width=\textwidth]{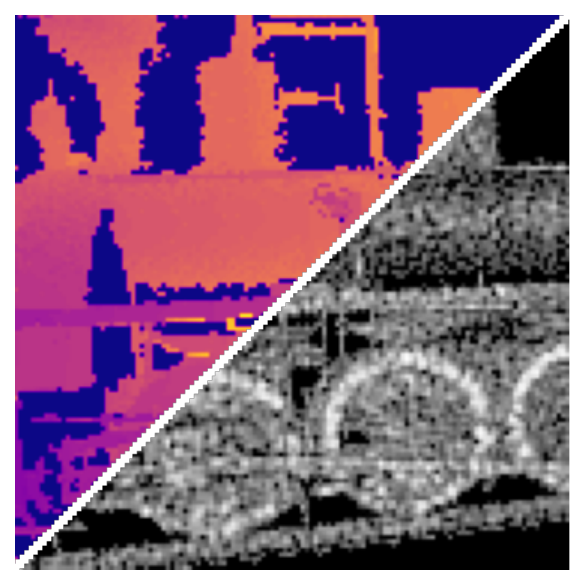}
      \end{minipage}
      \caption{$\ell_{2,1}$-regularization}\label{fig:train-l21}
    \end{subfigure}
    \begin{subfigure}{0.22\textwidth}
      \centering
      \begin{minipage}{0.48\textwidth}
          \includegraphics[width=\textwidth]{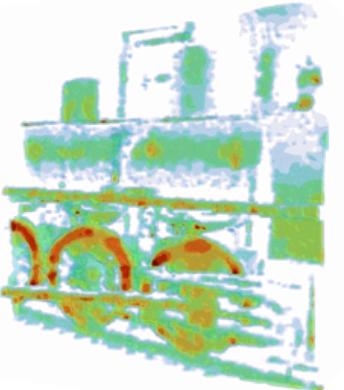}
      \end{minipage}
      \begin{minipage}{0.48\textwidth}
          \includegraphics[width=\textwidth]{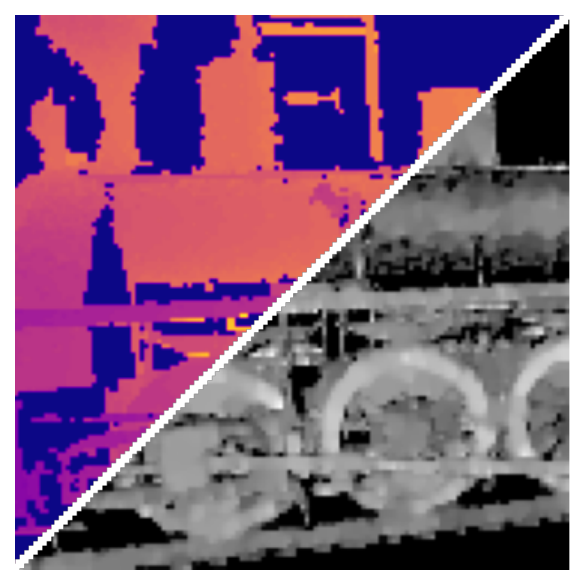}
      \end{minipage}
      \caption{TV-regularization}\label{fig:train-tv}
    \end{subfigure}
    \begin{subfigure}{0.22\textwidth}
      \centering
      \begin{minipage}{0.48\textwidth}
          \includegraphics[width=\textwidth]{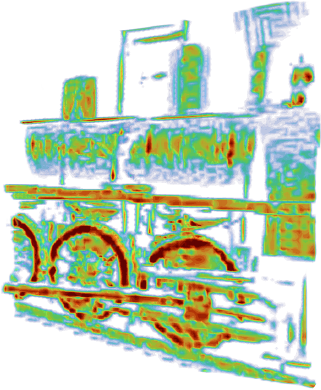}
      \end{minipage}
      \begin{minipage}{0.48\textwidth}
          \includegraphics[width=\textwidth]{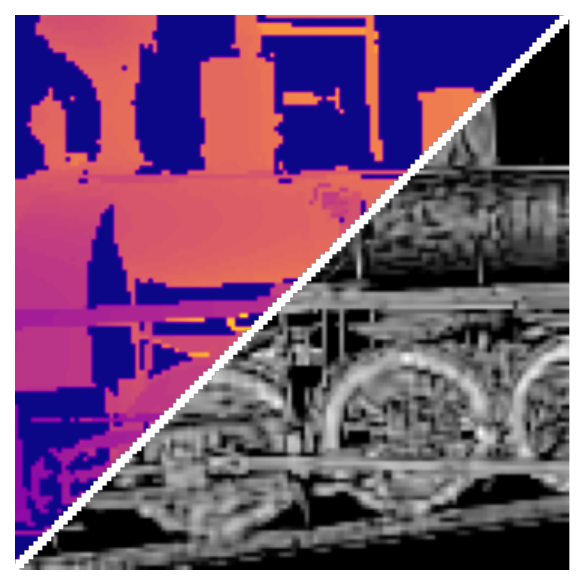}
      \end{minipage}
      \caption{\textbf{CLAMP}}\label{fig:train-clamp}
    \end{subfigure}
  \end{minipage}
  \caption{Sample reconstructions obtained from synthetic data using different zero-padding factors, $q=1, 1.5$ and 2; (a) the ground truth image of the train, (b)-(e) the reconstructions using the Speckle Average, $\ell_{2,1}$-regularization, TV-regularization, and CLAMP, respectively. Each 3D image is shown with its 2D depth and reflectivity image. The CLAMP reconstructions show the highest resolution and best agreement with the ground truth.}\label{fig:train-grid}
\end{figure*}

Each data simulation included nine looks at a 3D model of a train, shown in Figure~\ref{fig:loco-truth}.
Since 3D volumes can be difficult to visualize, we also show a pair of 2D depth and reflectivity images.
The color bars shown in this figure are representative and can be applied to all images in this paper.
The reflectivity image (shown in the bottom-right half) represents the maximum value along the depth dimension, while the depth (shown in the top-left half) is given by the location of the maximum.
These images often reveal more detail than the 3D volume image, which can be useful for visualizing the reconstruction quality.
However, due to the nature using the maximum value, undue noise can be introduced into the 2D images, particularly in the depth image.

The train was assumed to be 1 meter in size in each dimension.
Assuming the train to be a Lambertian surface, the reflectivity of each point on the target was computed as proportional to the cosine of the angle between the surface normal and the viewing direction.
The reflectivity was then normalized so that the brightest point in the image had unit reflectivity.
A speckle realization was generated by multiplying the reflectivity of each voxel in the image by a complex Gaussian random variable with unit variance.
The speckle realization was then propagated to the hologram plane by applying the forward model with parameters listed in the `Simulation' column of Table~\ref{table:params}.
The aperture was modeled as a centered, circular aperture with a diameter equal to $50\%$ of the hologram grid length.
Finally, the data is corrupted by adding complex white Gaussian noise with variance $\sigma_{\eta}^2 = 10^{-3}$.

\subsection{Experimental Data Measurement}

In addition to synthetic data, we evaluate CLAMP on data collected at the Air Force Research Laboratory.
The experiment consisted of two targets, a toy car and a hexagonal nut, which were painted with matte white paint so that the surfaces had a uniform Lambertian reflectance.
In order to measure multiple, statistically independent speckle realizations, the targets were placed on a high-precision rotation stage and rotated slightly after each measurement.

The targets were illuminated by a linear-frequency modulated waveform with central wavelength of \qty{1550}{\nm} (\qty{193.4}{THz}) and chirp rate of \qty{117.9}{\THz/s}.
The laser was split so that 95\% of the power was transmitted to the target and the remaining 5\% was used as a reference beam for holographic imaging.
The reference beam was placed off-axis in the plane of the exit pupil and pointed at the center of a focal plane array, where it interfered with a focused image of the target.
The interference pattern was recorded on a $640\times512$ pixel InGaAs photodetector with \qty{20}{\mu m} pixel pitch at 8-bit resolution. 
The camera's region of interest was narrowed to fit the image of the target, thereby allowing faster frame rates.
The toy car was measured by a $128 \times 128$ array of pixels at a frame rate of \qty{17.60}{\kHz}, and the hexagonal nut was measured by an $80 \times 80$ array at a frame rate of \qty{28.92}{\kHz}.
The noise variance was estimated to be $\sigma_{\eta}^2 = 0.0012$ for the toy car and $\sigma_{\eta}^2 = 0.0025$ for the hexagonal nut.
This estimation was done by computing the mean power within the pupil region of the hologram and dividing by the mean power outside of the pupil region.
This is similar to the noise variance used in the synthetic data experiment, which allows for a comparison of the reconstruction quality between the two experiments. 

Further details of the experimental setup, system hardware, and calibration methods can be found in~\cite[Chapter 4]{farrissIterativePhaseEstimation2021}. 
We summarize important experimental parameters in Table~\ref{table:params}.

\begin{table}
    \begin{center}
    \caption{Parameters used in our multi-look coherent lidar simulation and experimental measurements of a hexagonal nut and toy car.}\label{table:params}
    \begin{tabular}{ l c c c r }
        \toprule
        Item & Simulation & Hex. nut & Toy car & Units \\  
        \hline
        Distance & 52.9 & \qty{2.64}{} &\qty{2.64}{} &  \si{m} \\
        Central wavelength & \qty{1550}{} & \qty{1550}{} & \qty{1550}{} & \si{nm}\\
        Chirp rate & --- & \qty{117.9}{} & \qty{117.9}{} & \si{\THz}/\si{s} \\
        Chirp duration & --- & \qty{2.0}{} & \qty{2.0}{} & \si{\ms} \\
        Frame rate & --- & \qty{28.9}{} & \qty{17.6}{} & \si{\kHz} \\
        Frames, $N_t$ & $64$ & $46$ & $28$ & samples \\
        Frequency step size & \qty{0.15}{} & \qty{6.7}{} & \qty{4.2}{} & \si{\GHz} \\
        Aperture diameter & \qty{6.4}{} & \qty{6.4}{} & \qty{6.4}{} & \si{mm} \\
        Focal length & ---  & \qty{18.2}{} & \qty{18.2}{} & \si{cm} \\
        Pixel pitch & --- & \qty{20.0}{} & \qty{20.0}{} & $\mu$m \\
        Hologram grid size & $(128, 128)$ & $(80, 80)$  & $(128, 128)$ & pixels \\
        Noise variance, $\sigma_{\eta}^2$ & $0.001$ & $0.0025$ & $0.0012$ & --- \\
        SNR & 28.9 & 24.9 & 28.1 & dB \\
        \bottomrule 
    \end{tabular}
    \end{center}
\end{table}

\subsection{Synthetic Data Results}~\label{sec:synthetic}
In Figure~\ref{fig:train-grid}, we compare the reconstructions obtained from the synthetic data using the speckle average method, $\ell_{2,1}$-regularization, TV-regularization, and CLAMP using $q=1, 1.5$, and $2$.
As shown in the figure, incorporation of sparsity-based regularization, such as $\ell_{2,1}$ and TV, can reduce speckle noise and improve the resolution of the reconstruction.
However, in some cases, sparsity priors can also exacerbate speckle noise, as seen in the $\ell_{2,1}$-regularization reconstructions.
In contrast, CLAMP, with the use of a deep prior, is able to reconstruct a sharper and more detailed surface than the other methods. This is most evident in the undercarriage of the train, where the wheels and the tracks are more clearly defined in the CLAMP reconstruction with $q=2$.

\begin{figure}[!h]
  \centering
  \begin{minipage}{0.05\columnwidth}
    \vspace{-0.3cm} 
    \rotatebox{90}{\large{$q=2$} \hspace{0.5cm} \large{$q=1.5$} \hspace{0.5cm} \large{$q=1$}}
  \end{minipage}
  \hspace{0.01\columnwidth}
  \begin{minipage}{0.9\columnwidth}
    \begin{subfigure}{0.488\columnwidth}
      \centering
      \begin{minipage}{0.48\textwidth}
          \includegraphics[width=\textwidth]{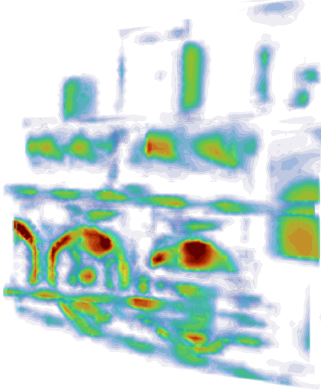}
      \end{minipage}
      \begin{minipage}{0.48\textwidth}
          \includegraphics[width=\textwidth]{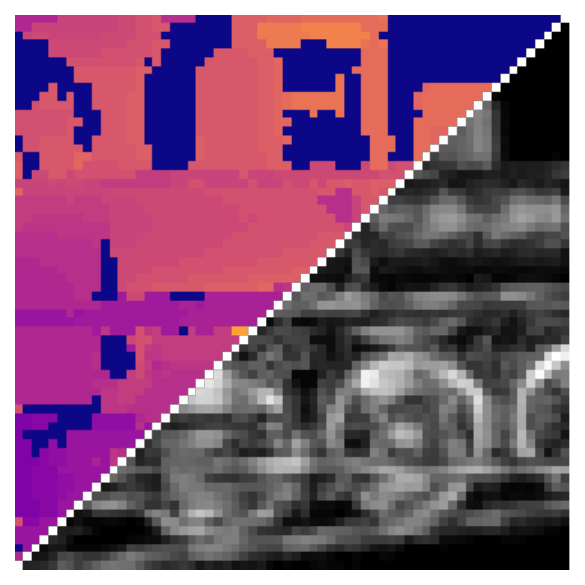}
      \end{minipage}
    \end{subfigure}
    \begin{subfigure}{0.488\columnwidth}
      \centering
      \begin{minipage}{0.48\textwidth}
          \includegraphics[width=\textwidth]{figures/train/loco3d-recon2.pdf}
      \end{minipage}
      \begin{minipage}{0.48\textwidth}
          \includegraphics[width=\textwidth]{figures/train/loco_recon_2d_q2.pdf}
      \end{minipage}
    \end{subfigure}
    \\

    \begin{subfigure}{0.488\columnwidth}
      \centering
      \begin{minipage}{0.48\textwidth}
          \includegraphics[width=\textwidth]{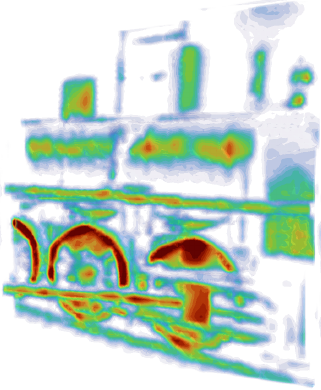}
      \end{minipage}
      \begin{minipage}{0.48\textwidth}
        \includegraphics[width=\textwidth]{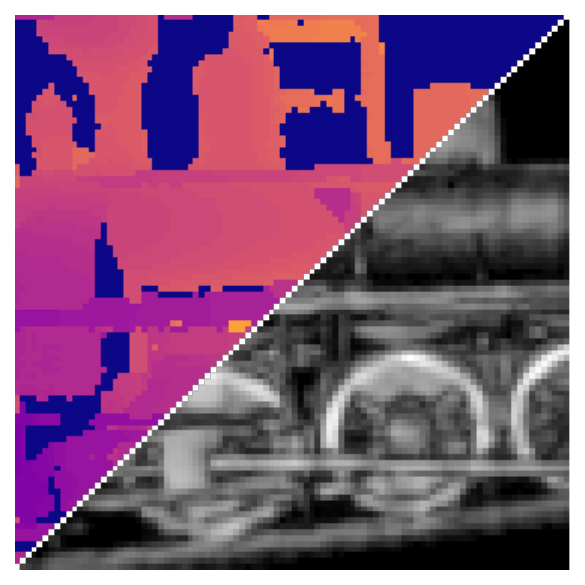}
      \end{minipage}
    \end{subfigure}
    \begin{subfigure}{0.488\columnwidth}
      \centering
      \begin{minipage}{0.48\textwidth}
          \includegraphics[width=\textwidth]{figures/train/loco3d-recon3.pdf}
      \end{minipage}
      \begin{minipage}{0.48\textwidth}
          \includegraphics[width=\textwidth]{figures/train/loco_recon_2d_q3.pdf}
      \end{minipage}
    \end{subfigure}
    \\

    \begin{subfigure}{0.488\columnwidth}
      \centering
      \begin{minipage}{0.48\textwidth}
          \includegraphics[width=\textwidth]{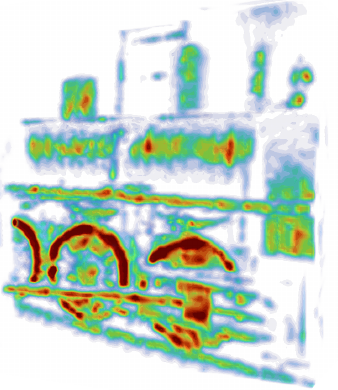}
      \end{minipage}
      \begin{minipage}{0.48\textwidth}
        \includegraphics[width=\textwidth]{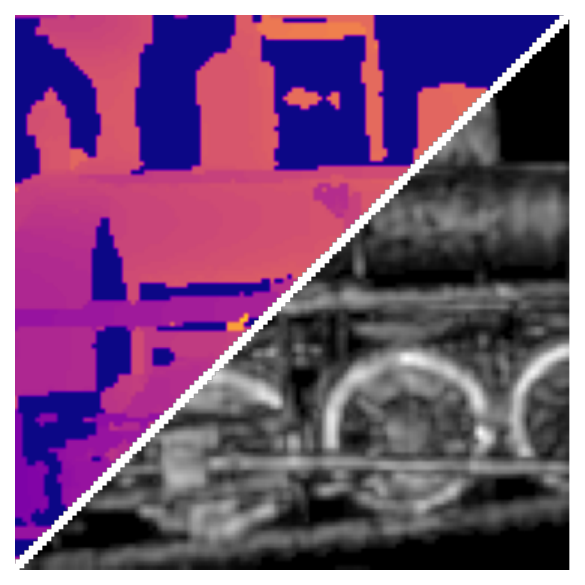}
      \end{minipage}
      \caption{No aperture model}
    \end{subfigure}
    \begin{subfigure}{0.488\columnwidth}
      \centering
      \begin{minipage}{0.48\textwidth}
          \includegraphics[width=\textwidth]{figures/train/loco3d-recon.pdf}
      \end{minipage}
      \begin{minipage}{0.48\textwidth}
          \includegraphics[width=\textwidth]{figures/train/loco_recon_2d.pdf}
      \end{minipage}
      \caption{With aperture model}
    \end{subfigure}
  \end{minipage}
  \caption{Ablation study of the aperture model in CLAMP. Each 3D image is shown with its 2D depth and reflectivity image. When the zero-padding factor $q=1$, the reconstructions are nearly identical. However, when $q=1.5$ or $q=2$, the aperture model in CLAMP improves the resolution of the reconstruction.}\label{fig:aperture-ablation}
\end{figure}

In Figure~\ref{fig:aperture-ablation}, we show the results of an ablation study of the aperture model in CLAMP.
We show reconstructions of the synthetic train data from two nearly identical CLAMP algorithms with $q=1, 1.5,$ and $2$.
The first column of reconstructions, labeled ``No aperture model,'' is the result of running CLAMP using $\mathbf{a} = \mathbf{1}$, effectively ignoring the aperture model.
The second column, labeled ``With aperture model,'' is the result of running the CLAMP algorithm as described in Algorithm~\ref{alg:clamp}.
When $q=1$, the reconstructions are nearly identical, as one should expect.
However, at $q=2$, the aperture model in CLAMP results in a sharper image and improves the ability to resolve high frequency components of the image.

In Table~\ref{table:results}, we show the PSNR, surface reflectivity NRMSE, Euclidean distance error and false positive rate for each reconstruction.
As seen from the table, the CLAMP reconstruction is more similar to the ground truth than the Speckle Average, $\ell_{2,1}$-regularization, and TV-regularization reconstructions, and with PSNR exceeding 30dB. 
The full CLAMP reconstruction also tend to have the lowest surface reflectivity NRMSE and Euclidean distance error.
This is likely since the other reconstructions are blurred and exhibit points further from the ground truth with lower reflectivity.
While the sparsity-based regularization methods also improve this metric, not to the extent of the proposed method.
The full CLAMP reconstruction shows improvement over the CLAMP reconstruction without the aperture model, particularly at higher zero-padding factors, further suggesting that the aperture model is crucial for quality, high-resolution reconstructions.
Finally, the advantage of regularization is further supported by the false positive rate, as Speckle Average reconstruction exhibits a significantly higher rate compared to the regularized methods.

\begin{table*}
  \begin{center}
    \caption{PSNRs, surface reflecitiy NRMSEs, and Euclidean distances (meters) between reconstructions and ground truth in our synthetic data experiment.}\label{table:results}
    \addtolength{\tabcolsep}{-0.4em}
    \begin{tabular}{c|c c c c c c}
    \toprule
    & \begin{tabular}{c} Zero-padding \\ Factor, $q$  \end{tabular} & \begin{tabular}{c} Speckle \\ Average  \end{tabular} &  \begin{tabular}{c} $\ell_{2,1}$ \\ Regularization  \end{tabular} & \begin{tabular}{c} TV \\ Regularization  \end{tabular} & \begin{tabular}{c} CLAMP \\ (no aperture \\  model) \end{tabular} & \begin{tabular}{c} \textbf{CLAMP} \\ \textbf{(with aperture} \\ \textbf{model)} \end{tabular} \\
    \hline
    \multirow{1}{*}{\parbox{1.5cm}{PSNR (dB)}} & 2 & 26.97 & 27.29 & 28.40 & 28.85 & \textbf{30.65} \\
    \hline
    \multirow{3}{*}{\parbox{1.5cm}{Surface Reflectivity NRMSE}} & 1 & 0.867 & 0.752 & 0.733 & \textbf{0.699} & 0.701 \\
    & 1.5 & 0.847 & 0.805 & 0.770 & 0.621 & \textbf{0.427} \\
    & 2 & 0.796& 0.742 & 0.699 & 0.662 & \textbf{0.410} \\
    \hline
    \multirow{3}{*}{\parbox{1.5cm}{Euclidean Distance (m)}}  & 1 & 0.017 & \textbf{0.011} & 0.012 & \textbf{0.011} & \textbf{0.011} \\
    & 1.5 & 0.018 & 0.031 & 0.010 & 0.031 & \textbf{0.008} \\
   & 2 & 0.018 & 0.010 & 0.013 & 0.013 & \textbf{0.009} \\
   \hline
   \multirow{3}{*}{\parbox{1.5cm}{False Positive Rate}}  
   & 1 & 0.33 & 0.054 & 0.039 & \textbf{0.029} & 0.036 \\
    & 1.5 & 0.48 & 0.031 & \textbf{0.006} & 0.021 & 0.008 \\
   & 2 & 0.56 & 0.018 & \textbf{0.004} & 0.016 & \textbf{0.004} \\
    \bottomrule
    \end{tabular}
  \end{center}
\end{table*}

\subsection{Experimental Data Results}\label{sec:experimental}
\begin{figure*}[!ht]
  \centering
  \begin{minipage}{0.05\textwidth}
    \vspace{-1.0cm}
    \hspace{0.2cm}
    \rotatebox{90}{\large{$q=2$} \hspace{0.25cm} \large{$q=1.5$} \hspace{0.25cm} \large{$q=1$}}
  \end{minipage}%
  \begin{minipage}{0.9\textwidth}
    \begin{subfigure}{0.19\textwidth}
        \centering
        \begin{minipage}{0.480\textwidth}
            \includegraphics[width=\textwidth]{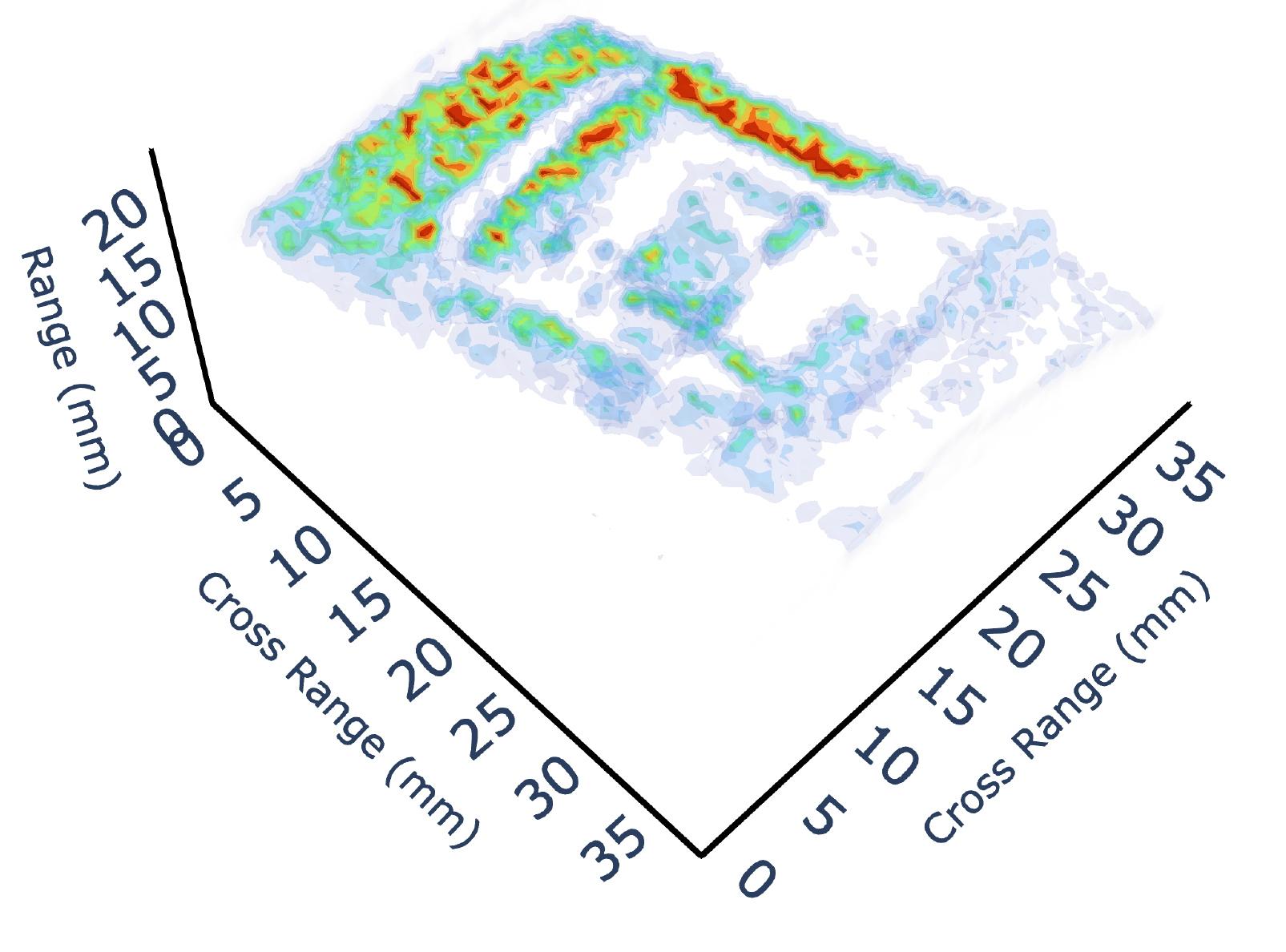}
        \end{minipage}
        \begin{minipage}{0.480\textwidth}
            \includegraphics[width=\textwidth]{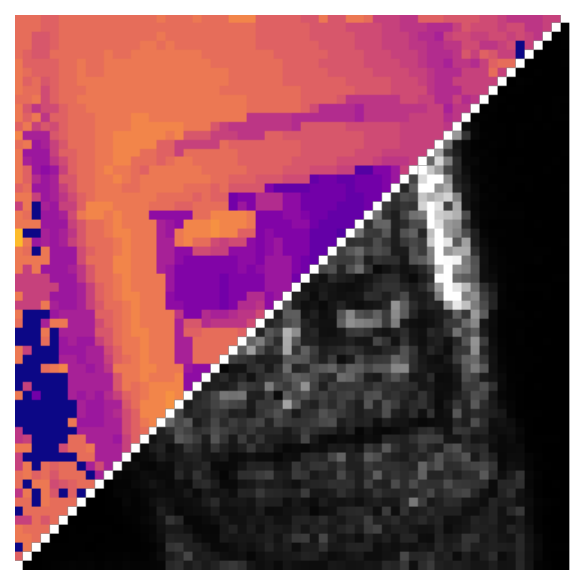}
        \end{minipage}
    \end{subfigure}
    \begin{subfigure}{0.19\textwidth}
      \centering
      \begin{minipage}{0.480\textwidth}
          \includegraphics[width=\textwidth]{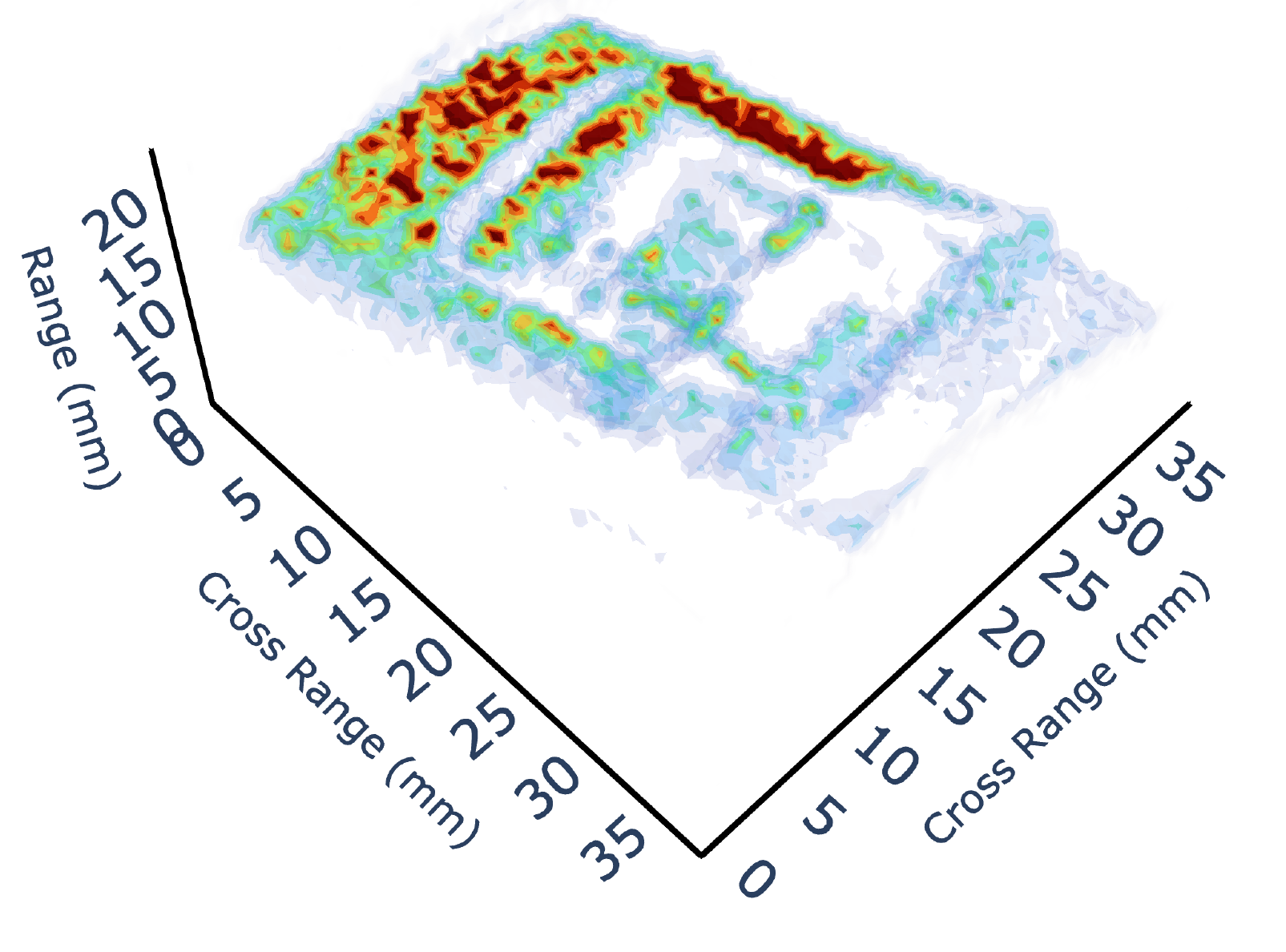}
      \end{minipage}
      \begin{minipage}{0.480\textwidth}
          \includegraphics[width=\textwidth]{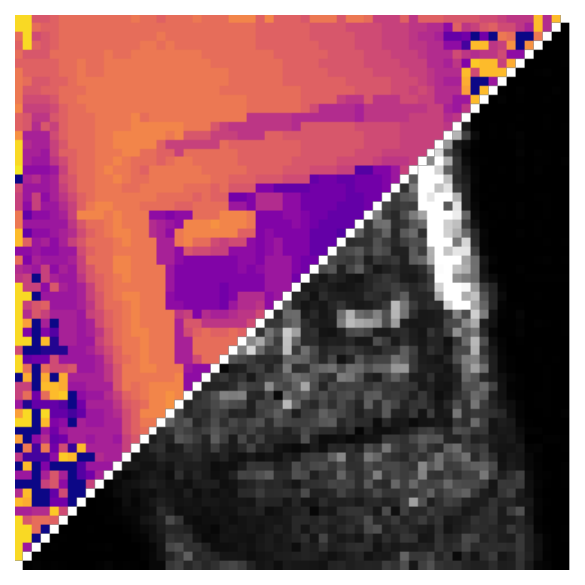}
      \end{minipage}
    \end{subfigure}
    \begin{subfigure}{0.19\textwidth}
      \centering
      \begin{minipage}{0.480\textwidth}
          \includegraphics[width=\textwidth]{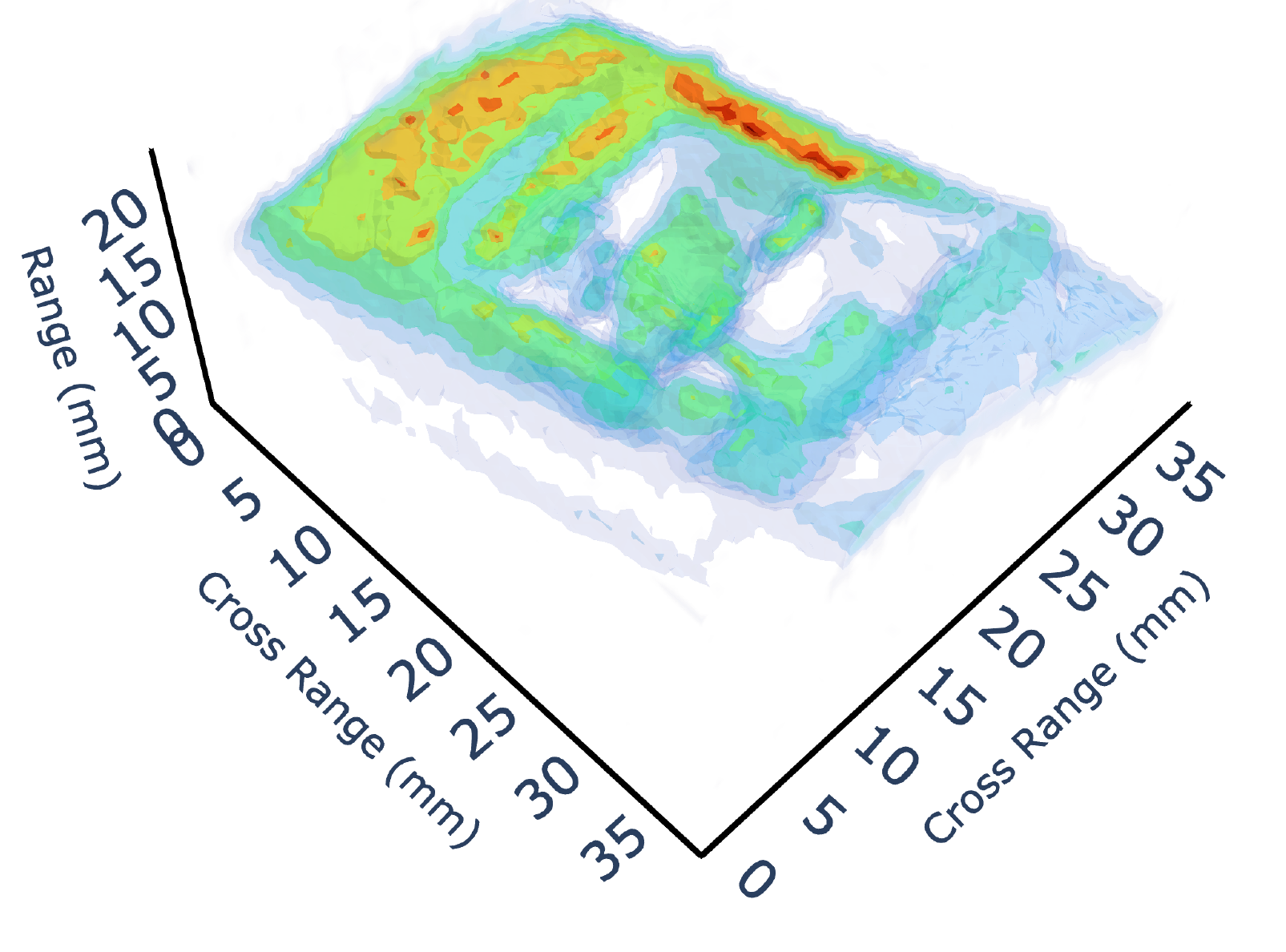}
      \end{minipage}
      \begin{minipage}{0.480\textwidth}
          \includegraphics[width=\textwidth]{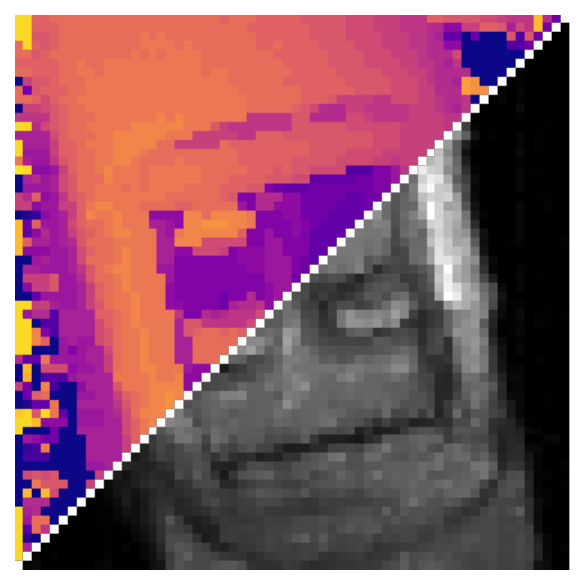}
      \end{minipage}
    \end{subfigure}
    \begin{subfigure}{0.19\textwidth}
      \centering
      \begin{minipage}{0.480\textwidth}
          \includegraphics[width=\textwidth]{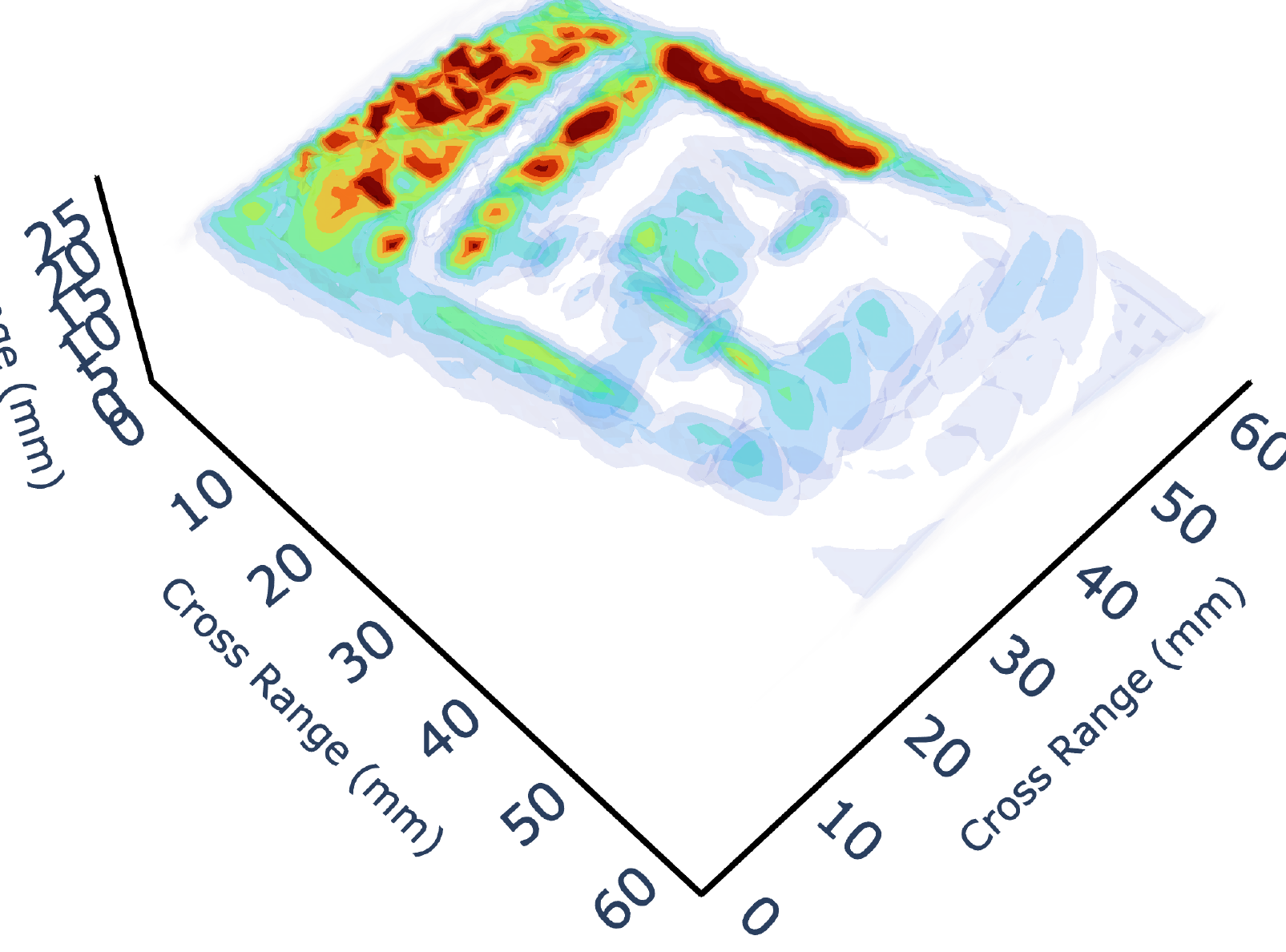}
      \end{minipage}
      \begin{minipage}{0.480\textwidth}
          \includegraphics[width=\textwidth]{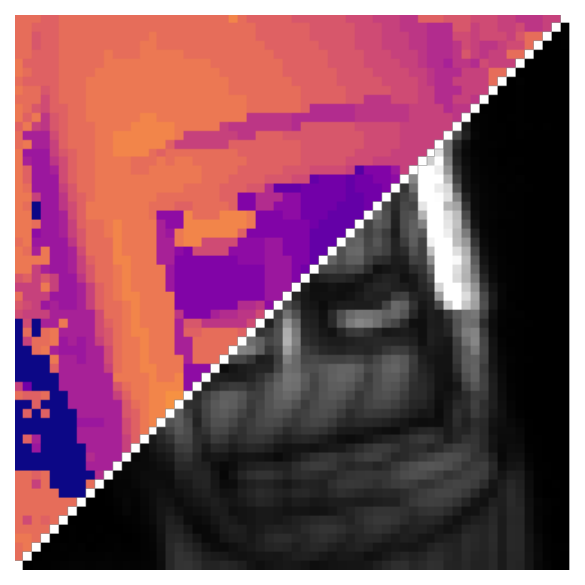}
      \end{minipage}
    \end{subfigure}
    \begin{subfigure}{0.19\textwidth}
      \centering
      \begin{minipage}{0.480\textwidth}
          \includegraphics[width=\textwidth]{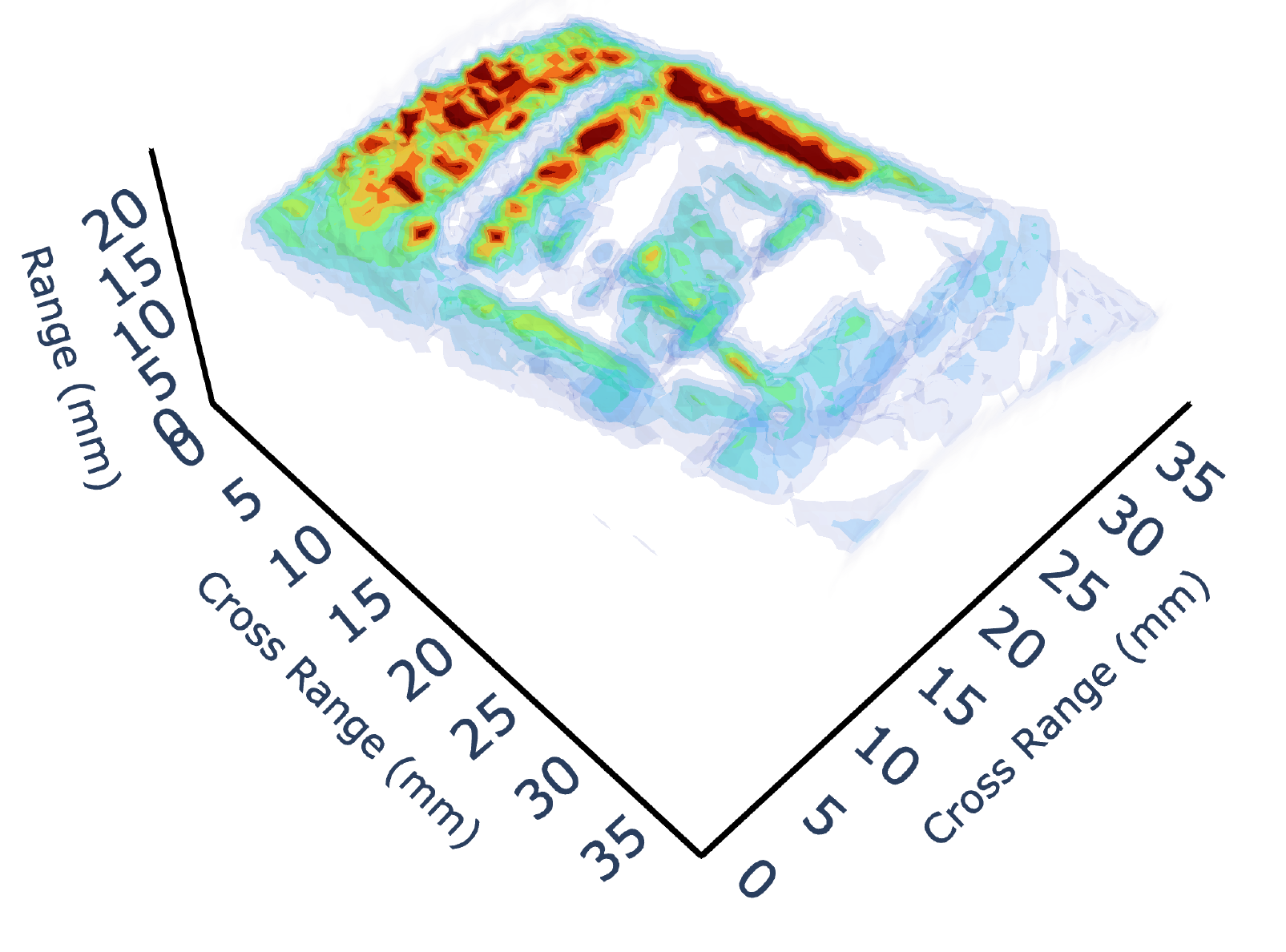}
      \end{minipage}
      \begin{minipage}{0.480\textwidth}
          \includegraphics[width=\textwidth]{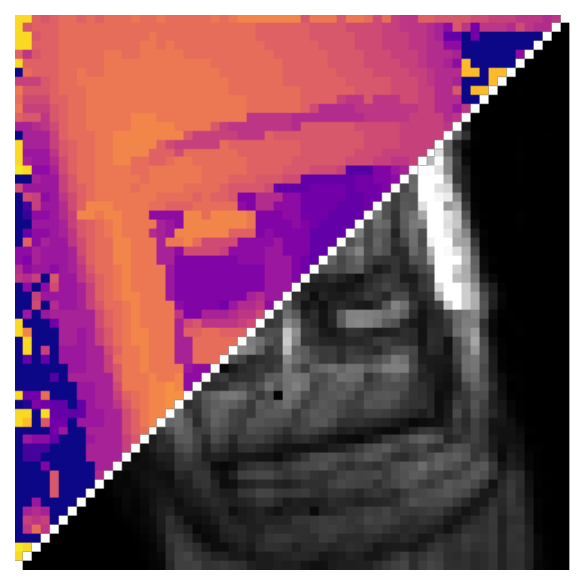}
      \end{minipage}
    \end{subfigure}
    \\
    \begin{subfigure}{0.19\textwidth}
      \centering
      \begin{minipage}{0.480\textwidth}
          \includegraphics[width=\textwidth]{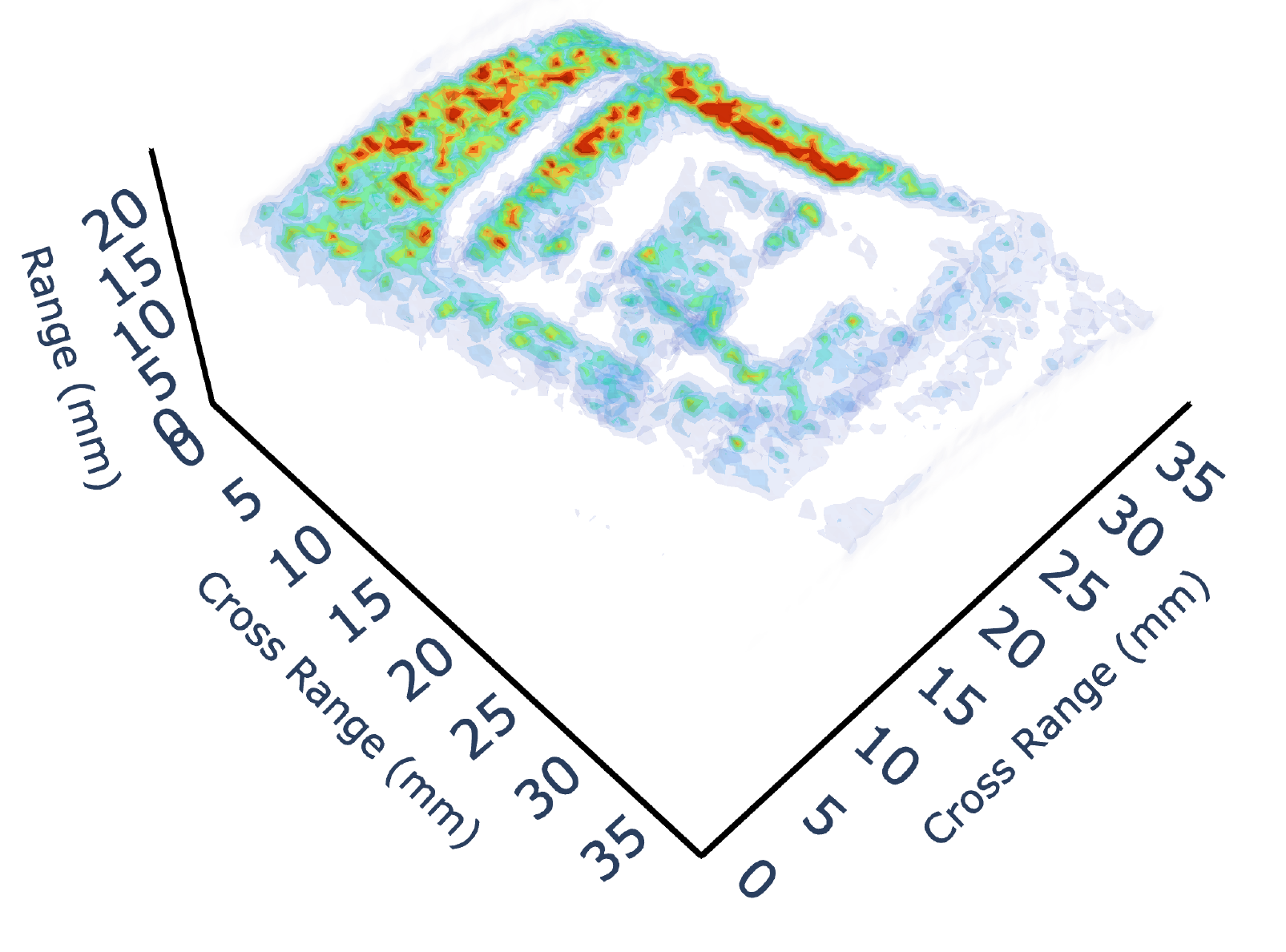}
      \end{minipage}
      \begin{minipage}{0.480\textwidth}
          \includegraphics[width=\textwidth]{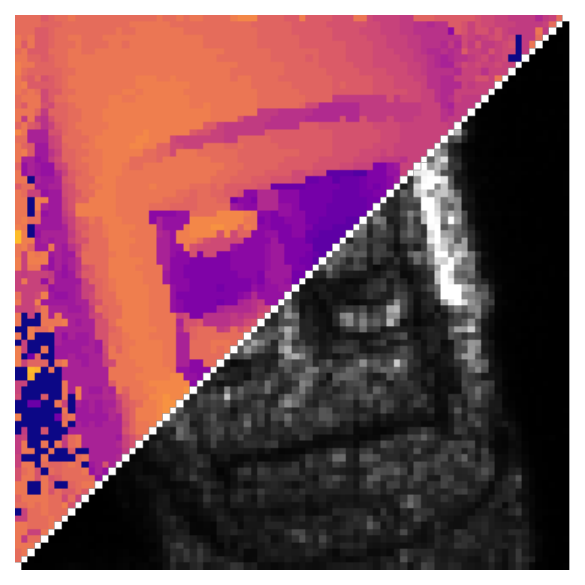}
      \end{minipage}
    \end{subfigure}
    \begin{subfigure}{0.19\textwidth}
      \centering
      \begin{minipage}{0.480\textwidth}
          \includegraphics[width=\textwidth]{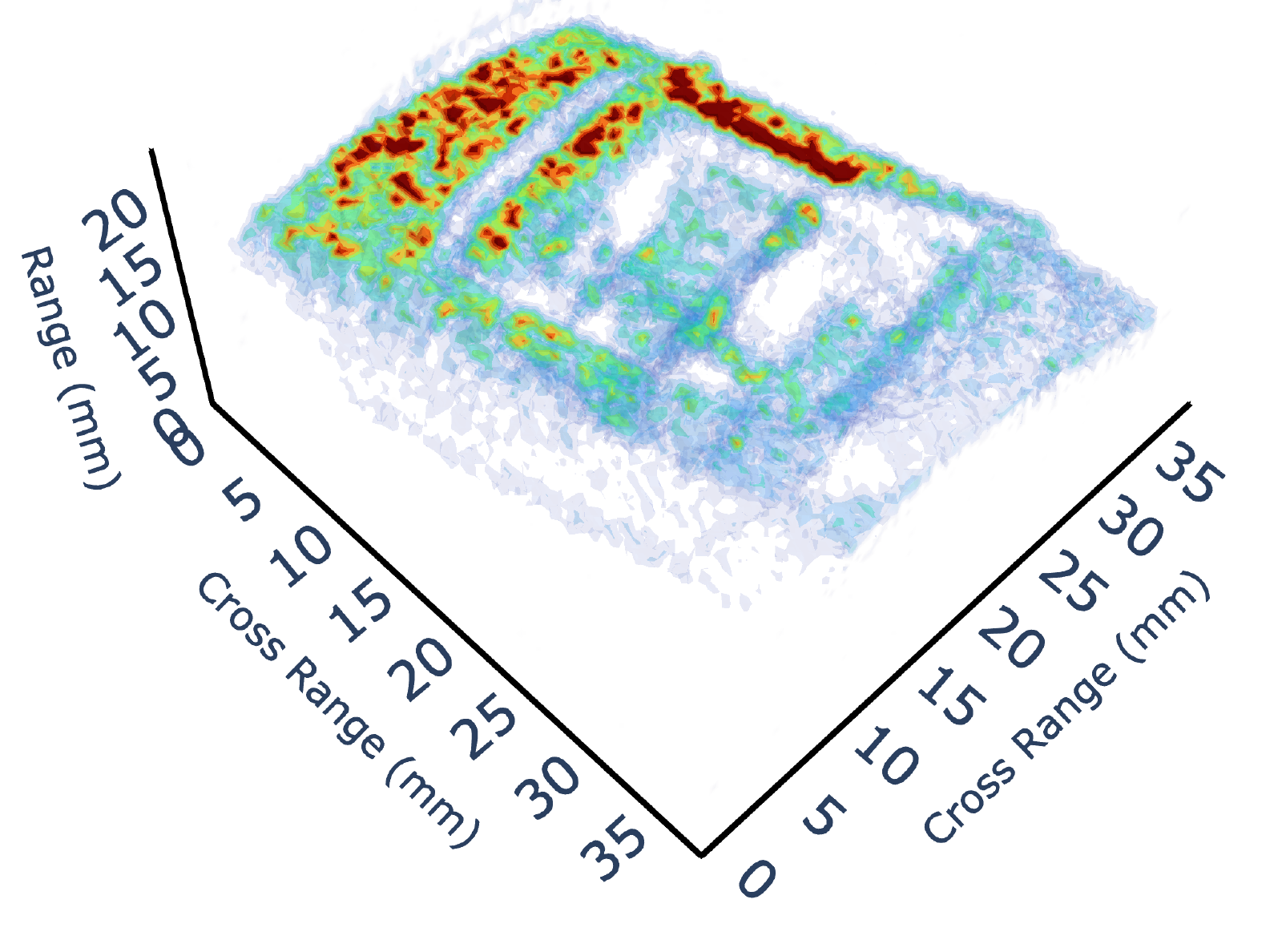}
      \end{minipage}
      \begin{minipage}{0.480\textwidth}
          \includegraphics[width=\textwidth]{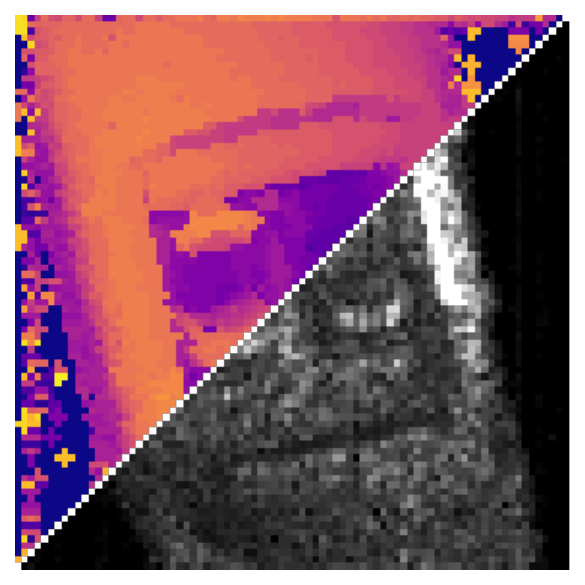}
      \end{minipage}
    \end{subfigure}
    \begin{subfigure}{0.19\textwidth}
      \centering
      \begin{minipage}{0.480\textwidth}
          \includegraphics[width=\textwidth]{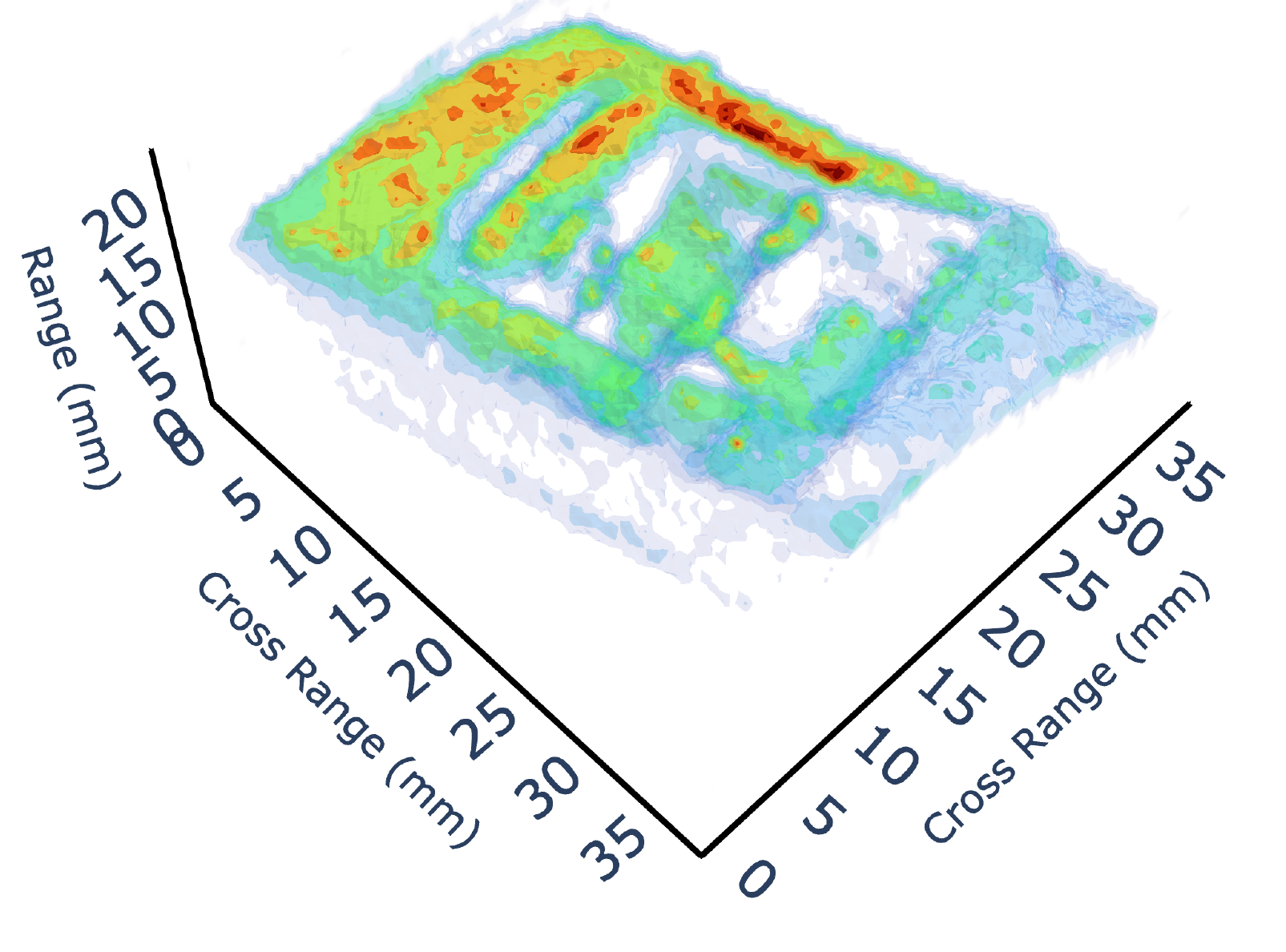}
      \end{minipage}
      \begin{minipage}{0.480\textwidth}
          \includegraphics[width=\textwidth]{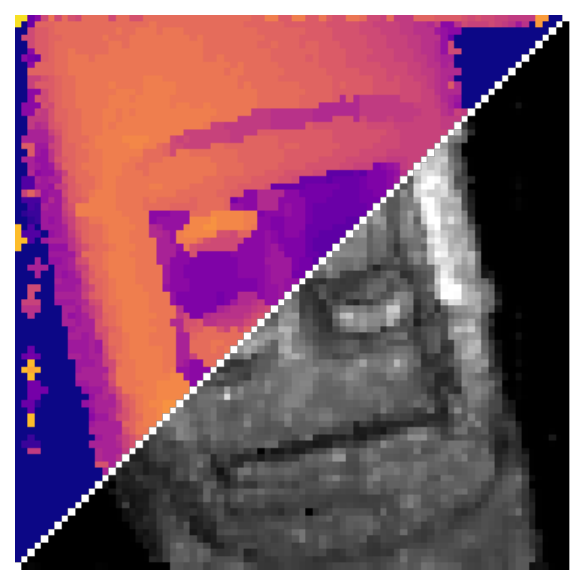}
      \end{minipage}
    \end{subfigure}
    \begin{subfigure}{0.19\textwidth}
      \centering
      \begin{minipage}{0.480\textwidth}
          \includegraphics[width=\textwidth]{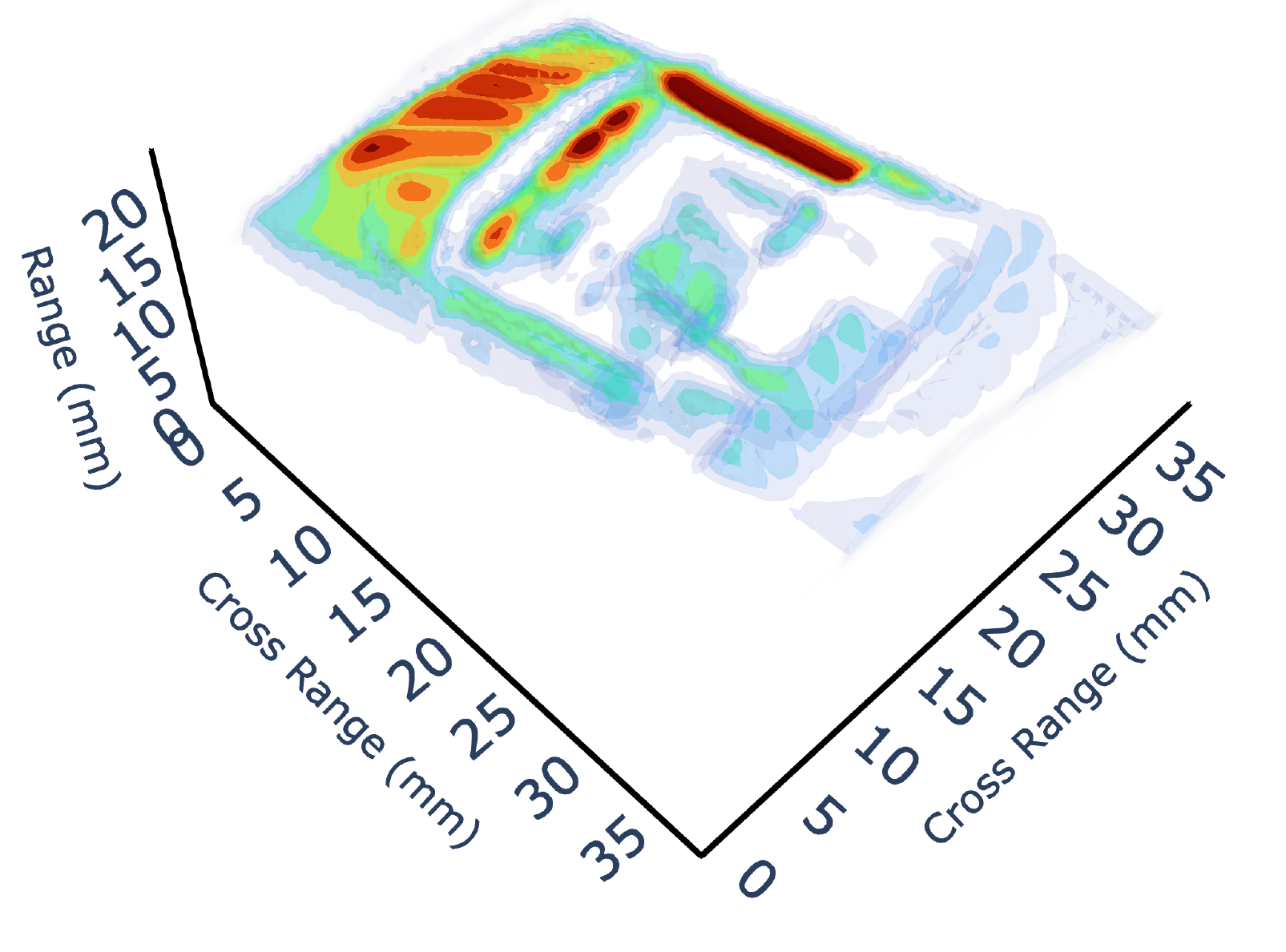}
      \end{minipage}
      \begin{minipage}{0.480\textwidth}
          \includegraphics[width=\textwidth]{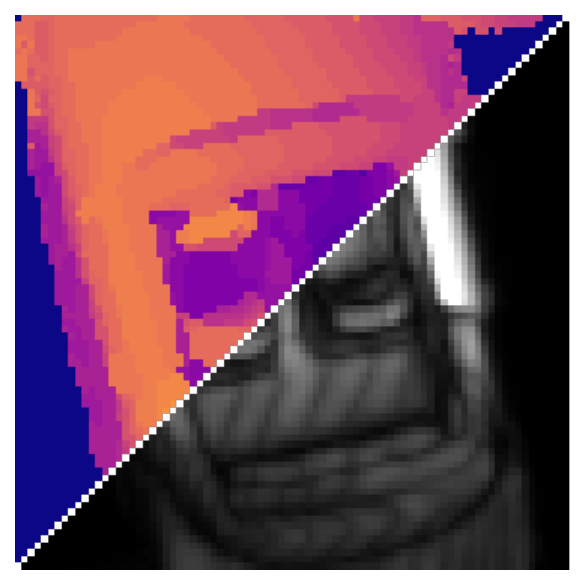}
      \end{minipage}
    \end{subfigure}
    \begin{subfigure}{0.19\textwidth}
      \centering
      \begin{minipage}{0.480\textwidth}
          \includegraphics[width=\textwidth]{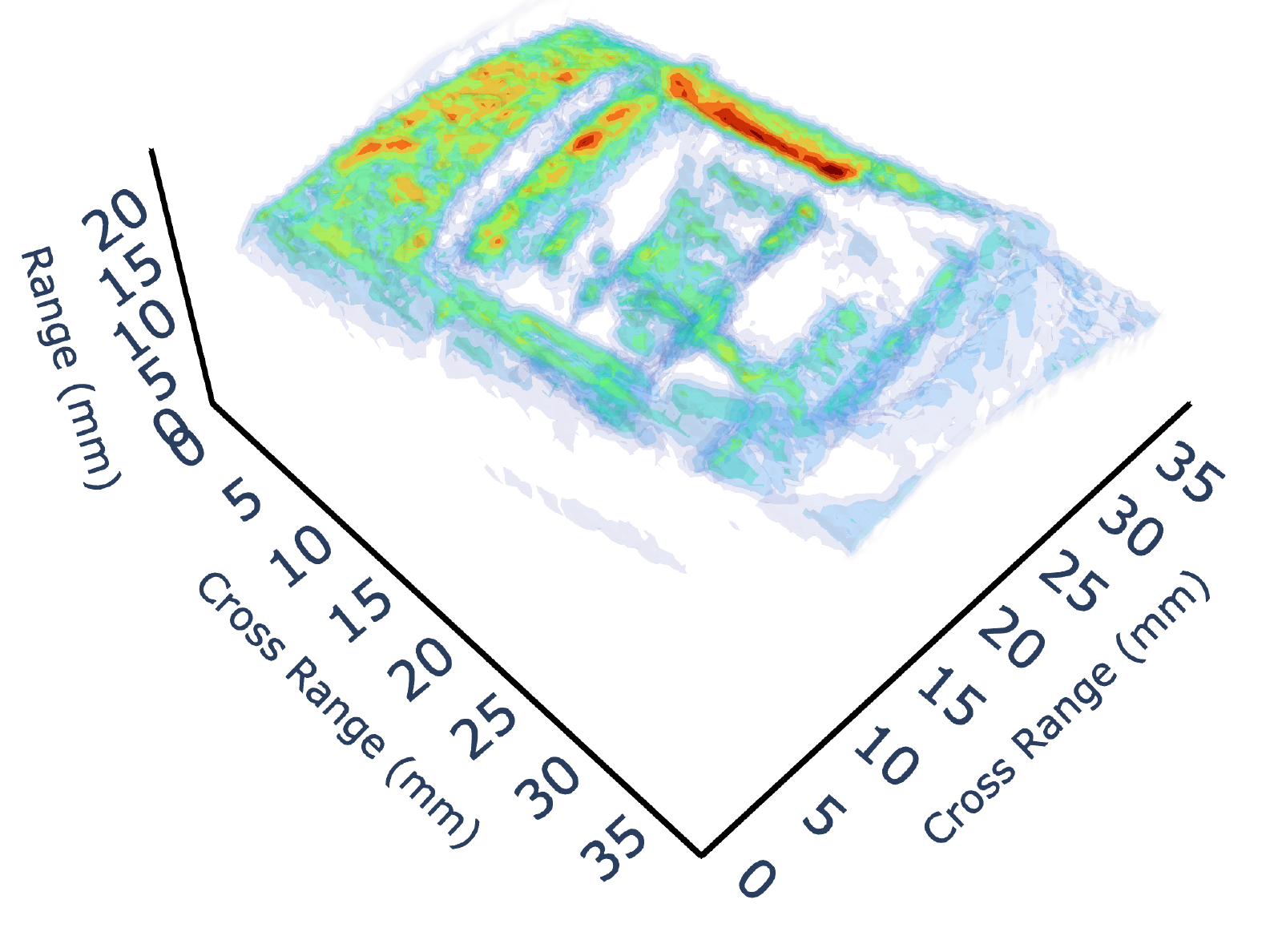}
      \end{minipage}
      \begin{minipage}{0.480\textwidth}
          \includegraphics[width=\textwidth]{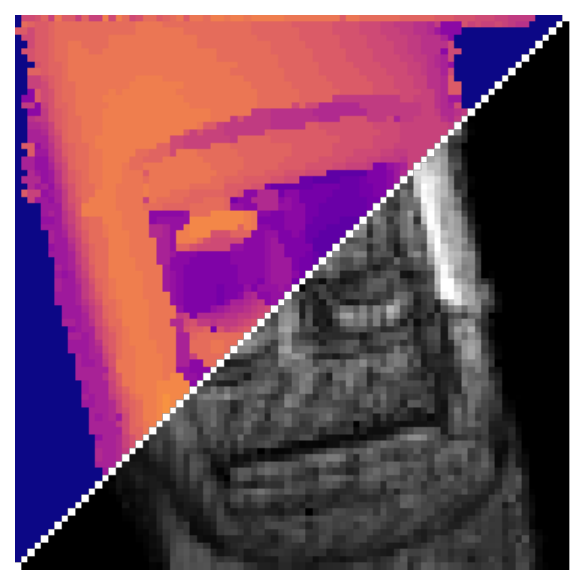}
      \end{minipage}
    \end{subfigure}
    \\
    \begin{subfigure}{0.19\textwidth}
      \centering
      \begin{minipage}{0.480\textwidth}
          \includegraphics[width=\textwidth]{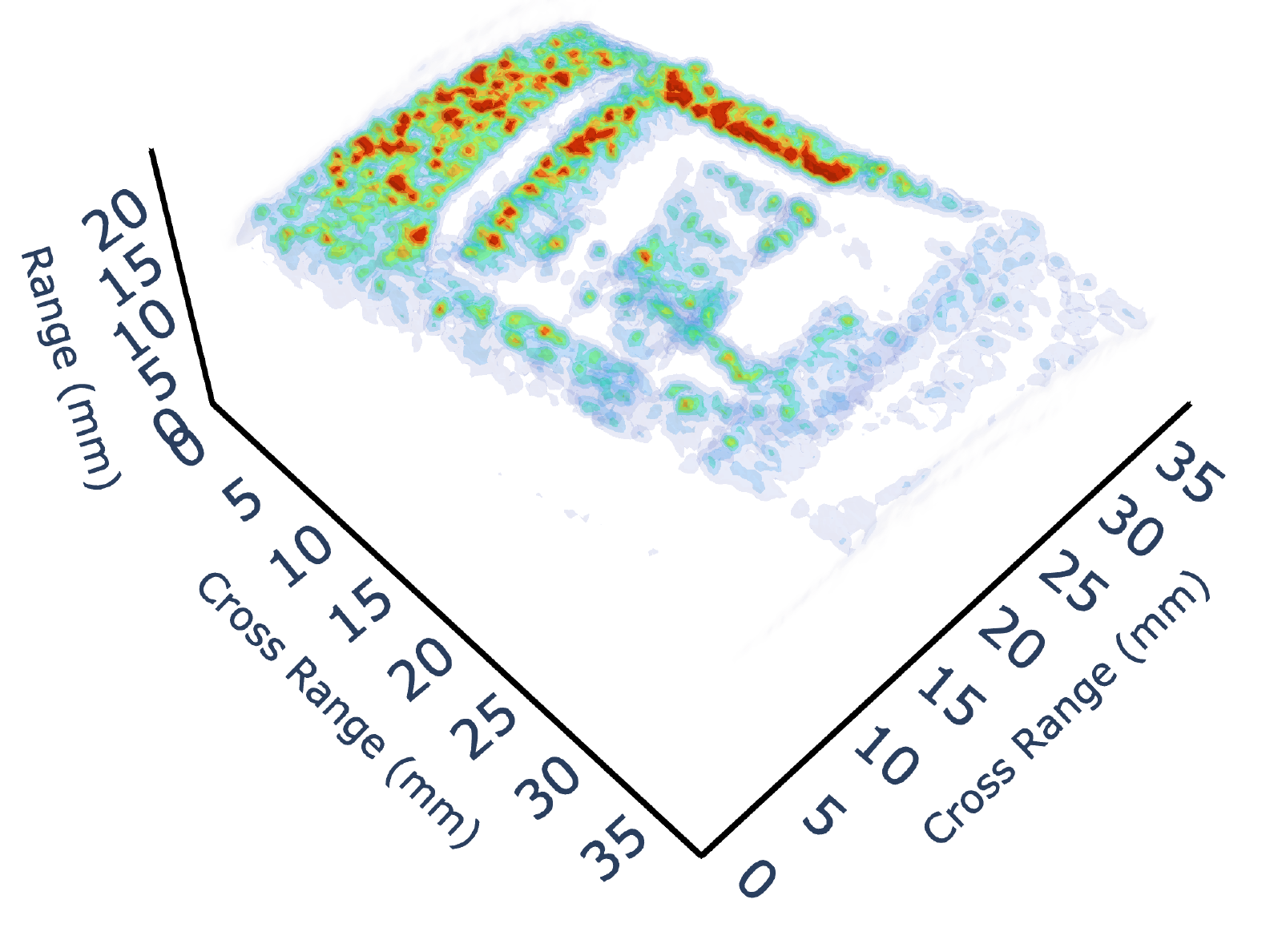}
      \end{minipage}
      \begin{minipage}{0.480\textwidth}
          \includegraphics[width=\textwidth]{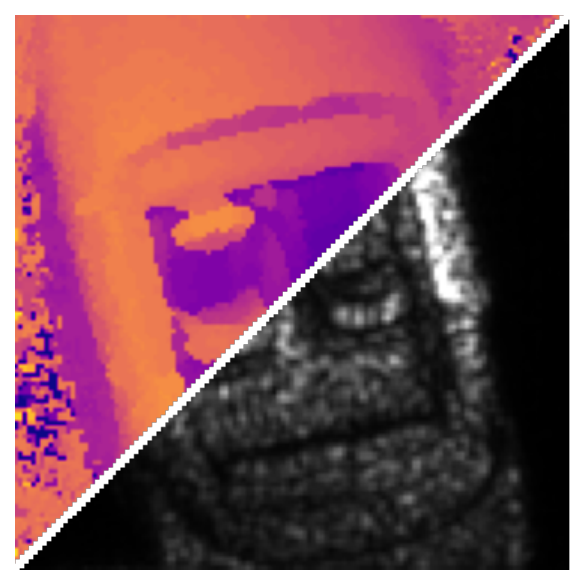}
      \end{minipage}
      \caption{Speckle Average}
    \end{subfigure}
    \begin{subfigure}{0.19\textwidth}
      \centering
      \begin{minipage}{0.480\textwidth}
          \includegraphics[width=\textwidth]{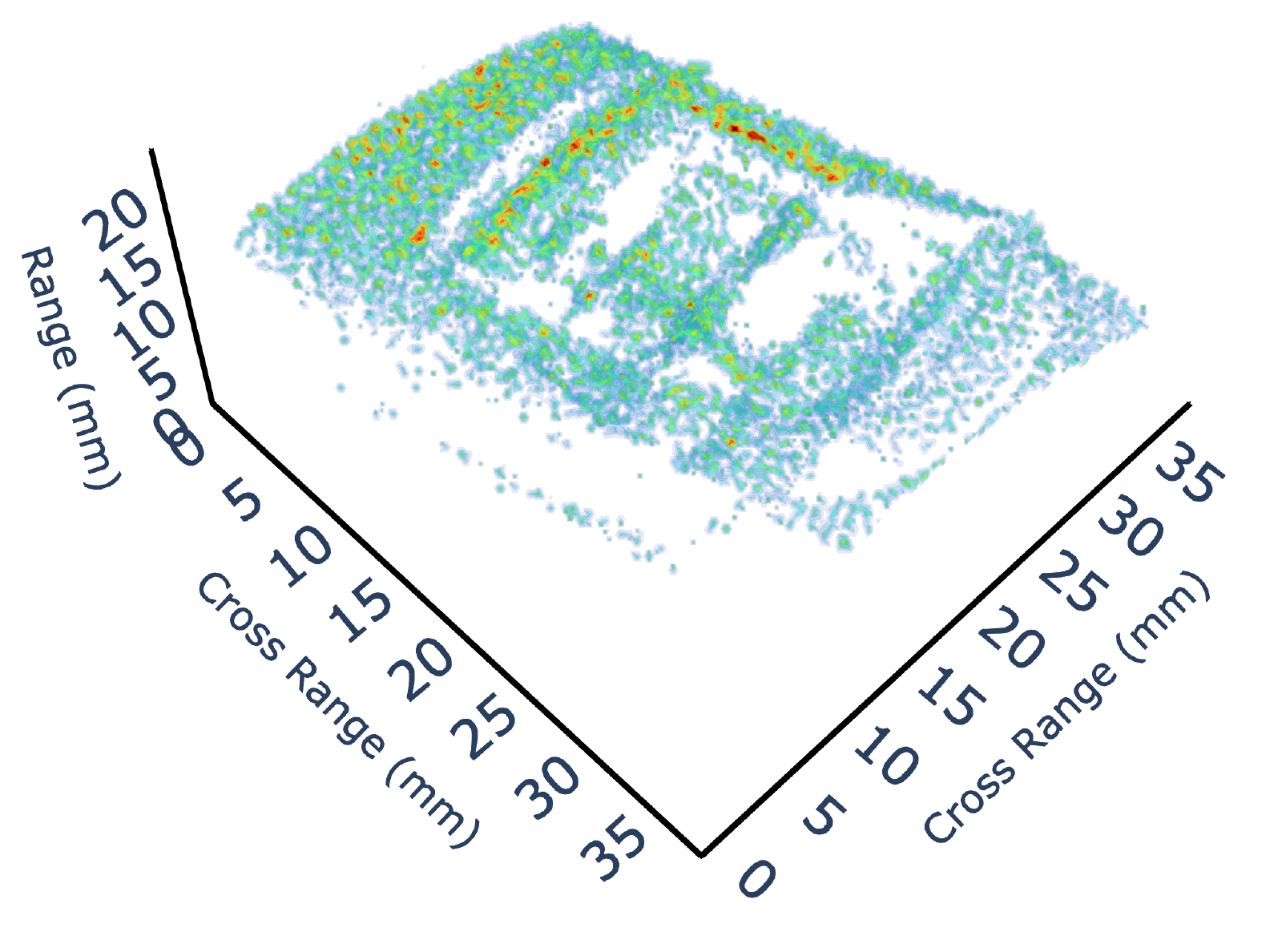}
      \end{minipage}
      \begin{minipage}{0.480\textwidth}
          \includegraphics[width=\textwidth]{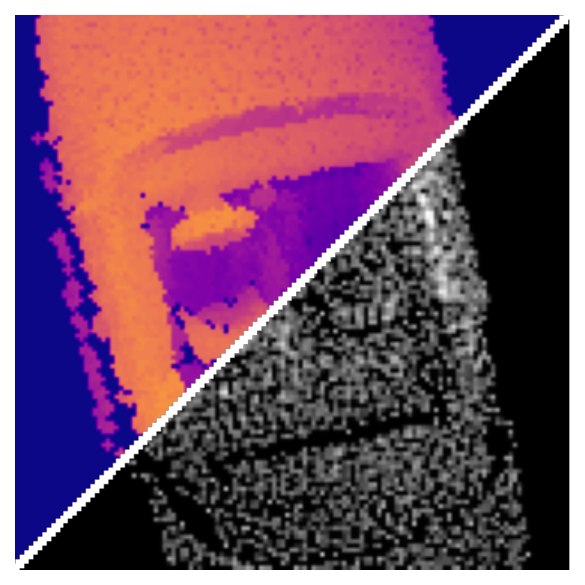}
      \end{minipage}
      \caption{$\ell_{2,1}$-regularization}
    \end{subfigure}
    \begin{subfigure}{0.19\textwidth}
      \centering
      \begin{minipage}{0.480\textwidth}
          \includegraphics[width=\textwidth]{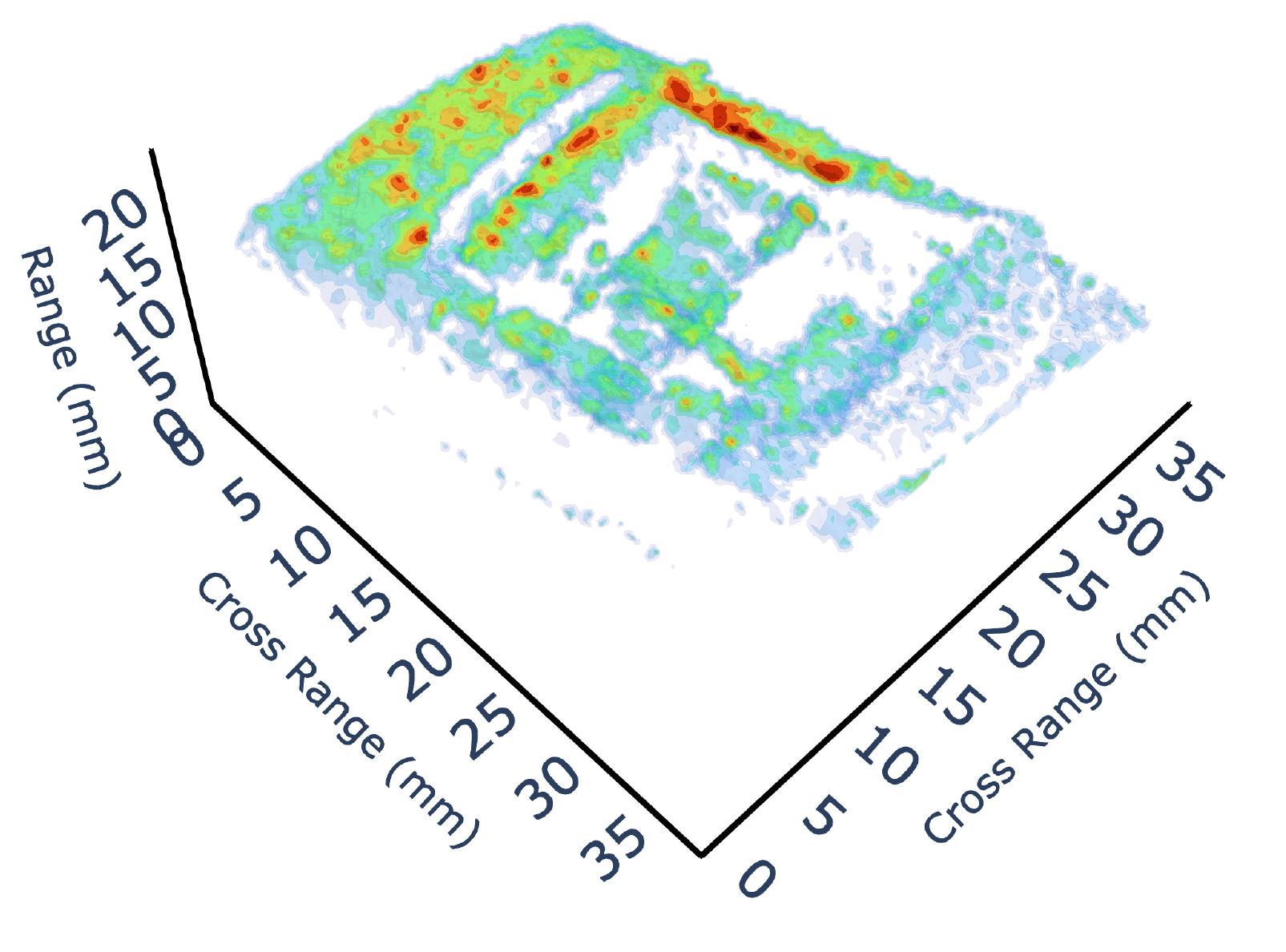}
      \end{minipage}
      \begin{minipage}{0.480\textwidth}
          \includegraphics[width=\textwidth]{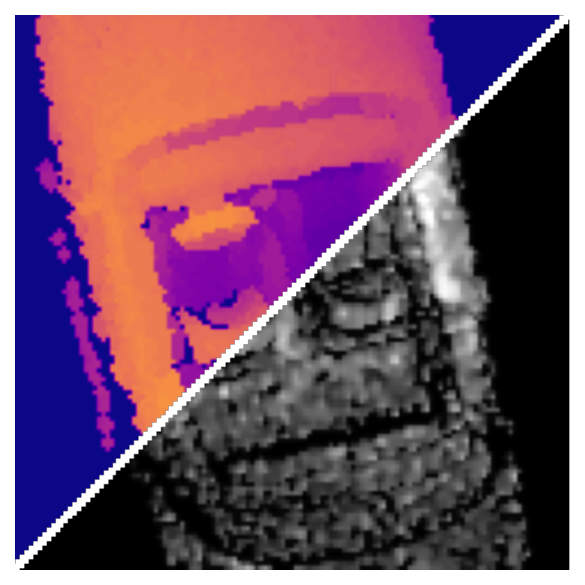}
      \end{minipage}
      \caption{TV-regularization}
    \end{subfigure}
    \begin{subfigure}{0.19\textwidth}
      \centering
      \begin{minipage}{0.480\textwidth}
          \includegraphics[width=\textwidth]{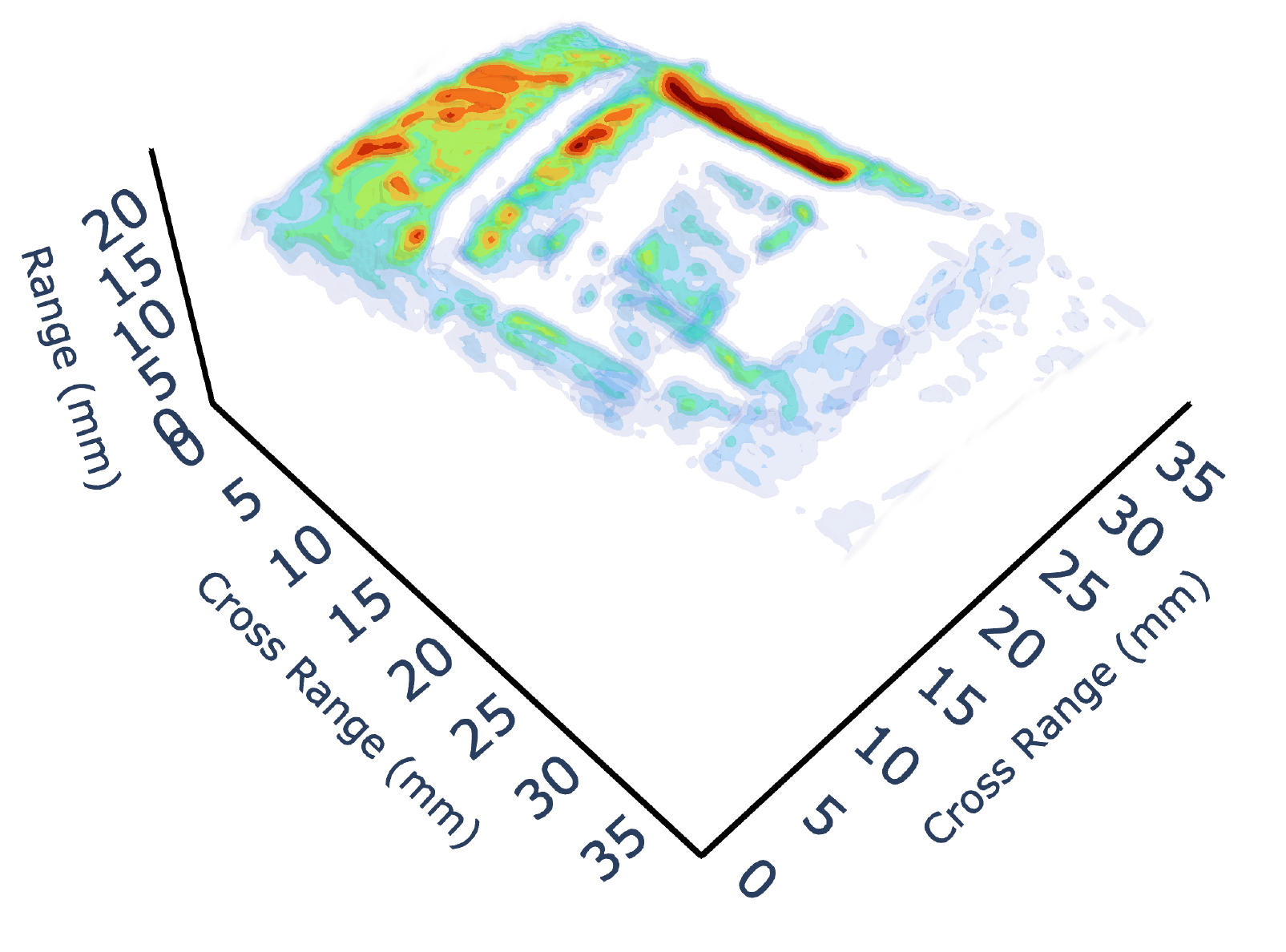}
      \end{minipage}
      \begin{minipage}{0.480\textwidth}
          \includegraphics[width=\textwidth]{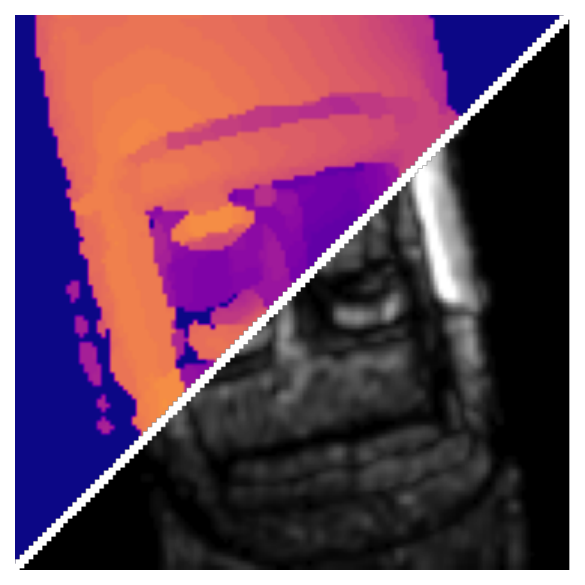}
      \end{minipage}
      \caption{CLAMP (No Aperture)}
    \end{subfigure}
    \begin{subfigure}{0.19\textwidth}
      \centering
      \begin{minipage}{0.480\textwidth}
          \includegraphics[width=\textwidth]{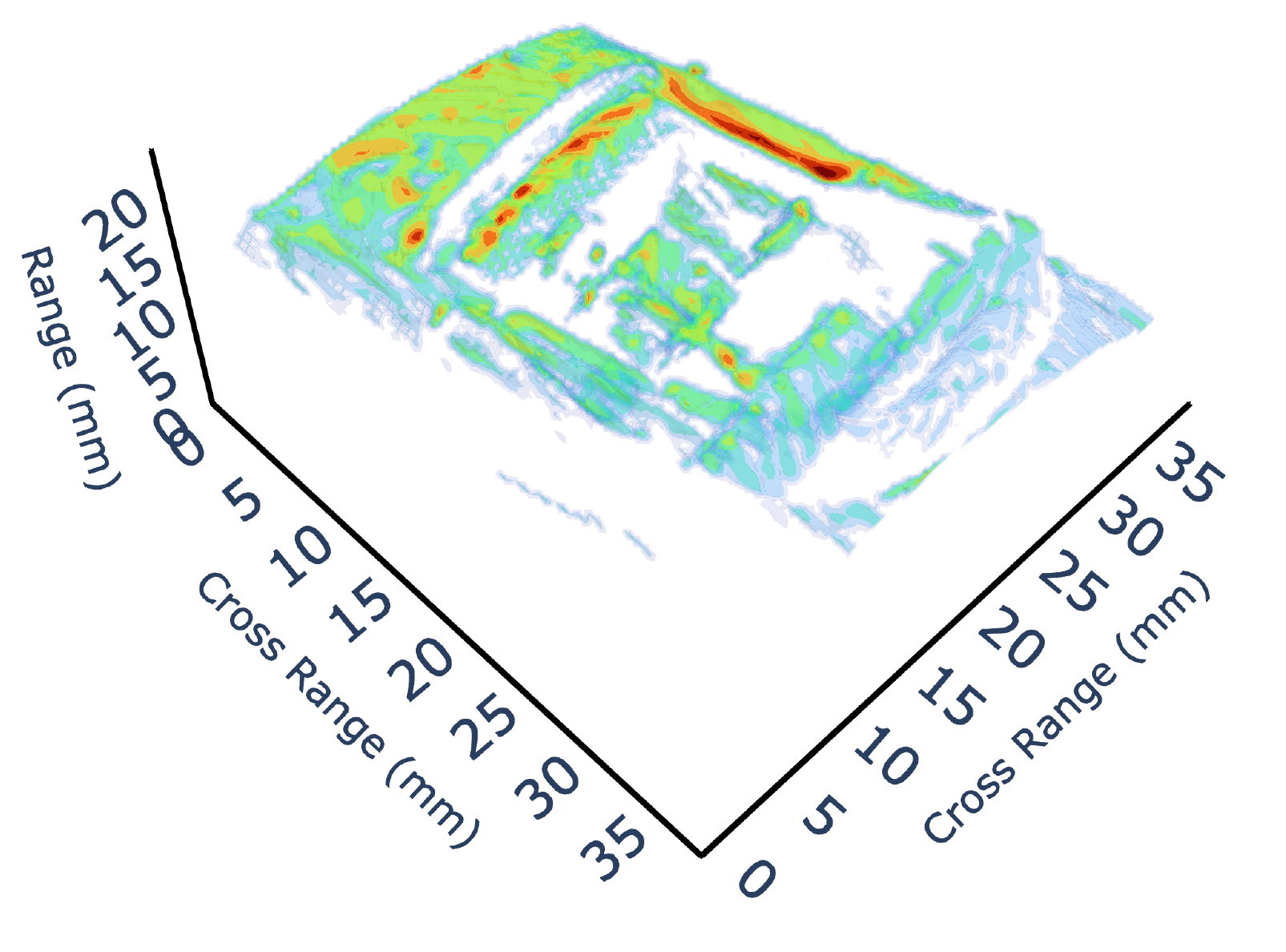}
      \end{minipage}
      \begin{minipage}{0.480\textwidth}
          \includegraphics[width=\textwidth]{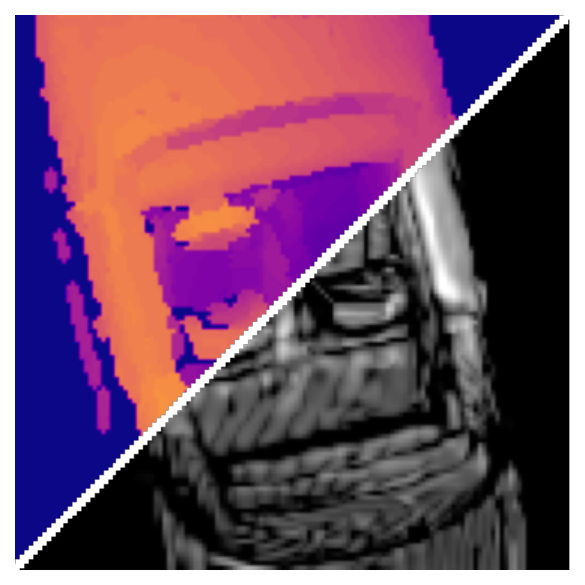}
      \end{minipage}
      \caption{\textbf{CLAMP}}
    \end{subfigure}
  \end{minipage}
  \caption{Reconstructions of the toy car from experimental data with zero-padding factors $q=1, 1.5$ and 2. Each 3D image is shown with its 2D depth and reflectivity image.
   The CLAMP reconstructions show improved speckle-reduction over the Speckle Average, $\ell_{2,1}$-, and TV-regularized reconstructions.}\label{fig:car-grid}
\end{figure*}

Figure~\ref{fig:car-grid} shows reconstructions from nine looks at a toy car at zero-padding factors, $q=1, 1.5$ and 2. 
Due to the advanced image prior, the CLAMP reconstructions shows improved speckle reduction and sharpness over the Speckle Average, $\ell_{2,1}$-, and TV-regularized reconstructions.
The $\ell_{2,1}$-regularized and TV-regularized reconstructions show some improvement over the Speckle Average, but the $\ell_{2,1}$-regularized suffers from speckle noise, and the TV-regularized reconstruction is overly smoothed, with some of the high reflectivity surfaces spreading across the surface.
The full CLAMP reconstruction again shows sharper and more detailed features over the CLAMP reconstructions without the aperture model.

\begin{figure*}[t]
  \centering
  \begin{minipage}{0.05\textwidth}
    \hspace{0.2cm}
    \rotatebox{90}{\large{$q=2$} \hspace{1.5cm} \large{$q=1.5$} \hspace{1.5cm} \large{$q=1$}}
  \end{minipage}%
  \begin{minipage}{0.9\textwidth}
    \begin{subfigure}{0.19\textwidth}
        \centering
        \begin{minipage}{0.99\textwidth}
            \includegraphics[width=\textwidth]{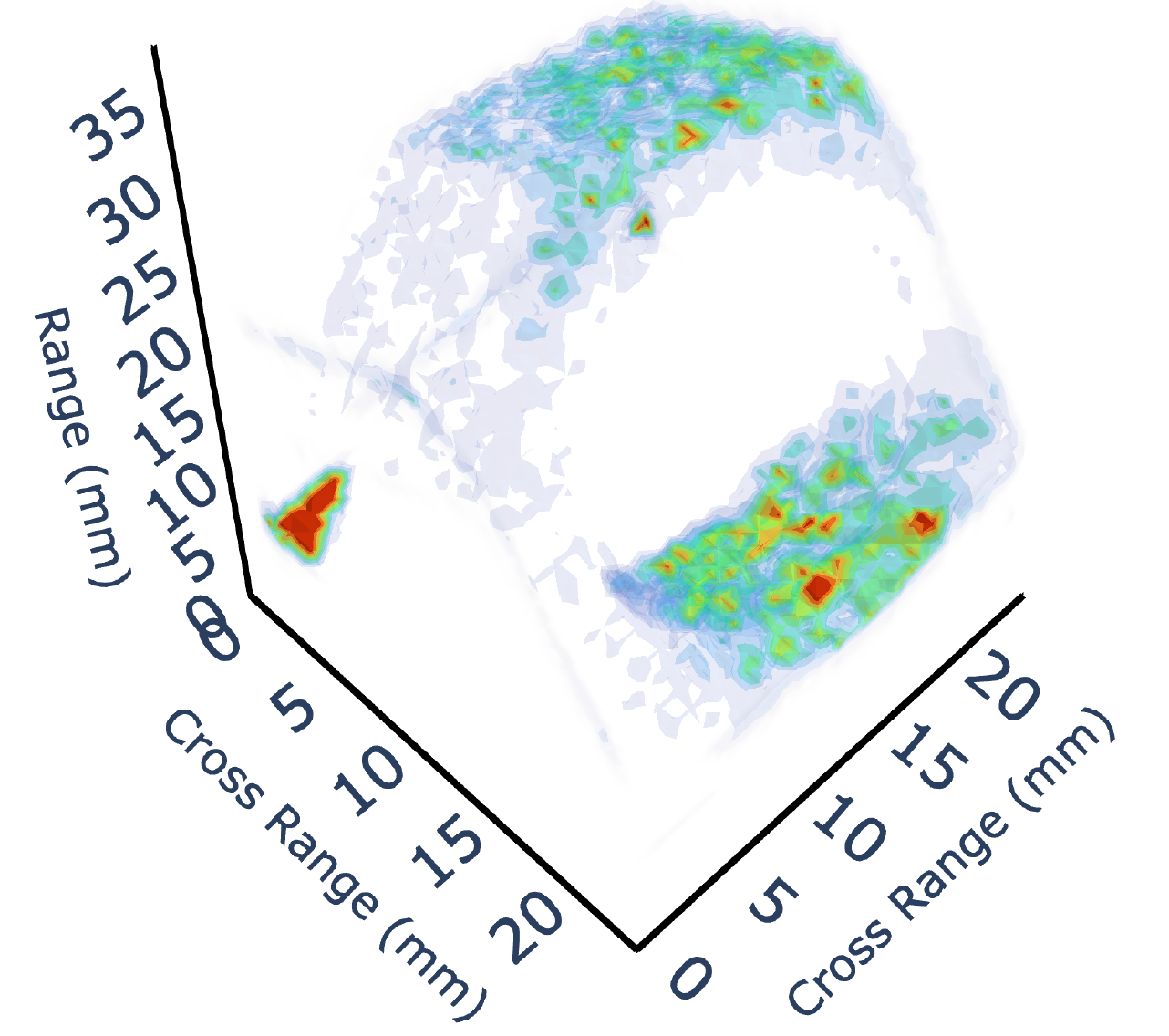}
        \end{minipage}
    \end{subfigure}
    \begin{subfigure}{0.19\textwidth}
      \centering
      \begin{minipage}{0.99\textwidth}
          \includegraphics[width=\textwidth]{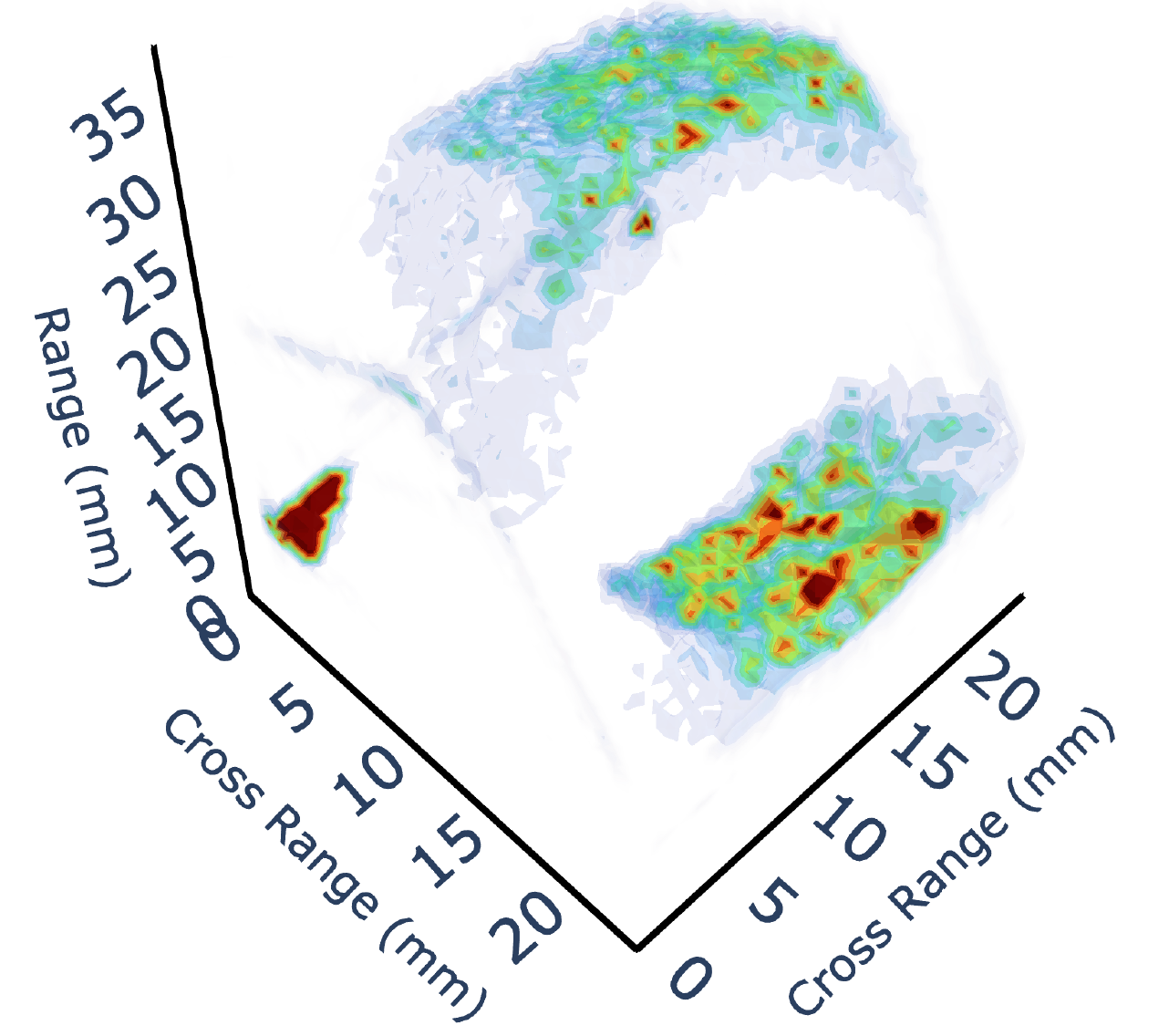}
      \end{minipage}
    \end{subfigure}
    \begin{subfigure}{0.19\textwidth}
      \centering
      \begin{minipage}{0.99\textwidth}
          \includegraphics[width=\textwidth]{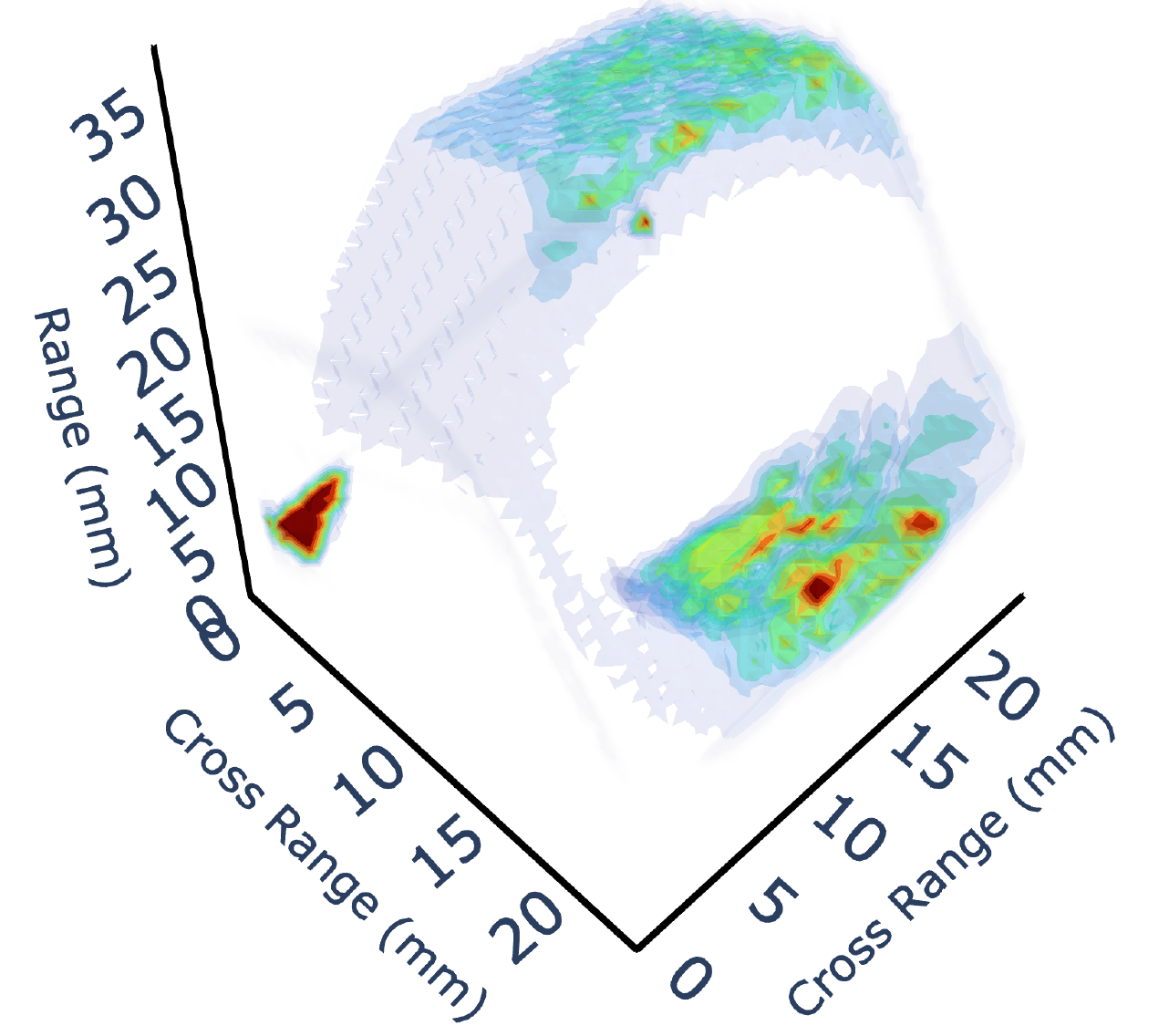}
      \end{minipage}
    \end{subfigure}
    \begin{subfigure}{0.19\textwidth}
      \centering
      \begin{minipage}{0.99\textwidth}
          \includegraphics[width=\textwidth]{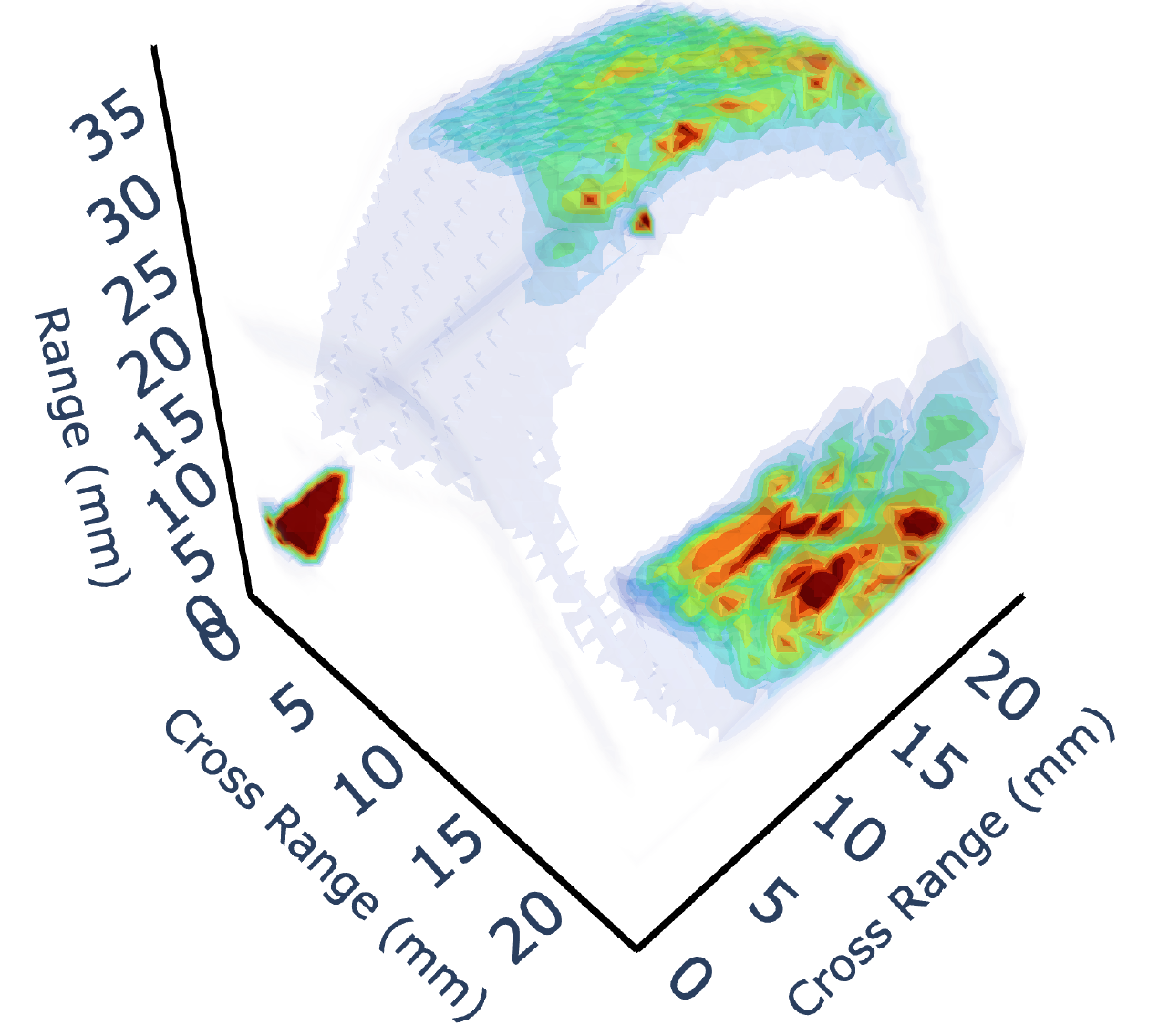}
      \end{minipage}
    \end{subfigure}
    \begin{subfigure}{0.19\textwidth}
      \centering
      \begin{minipage}{0.99\textwidth}
          \includegraphics[width=\textwidth]{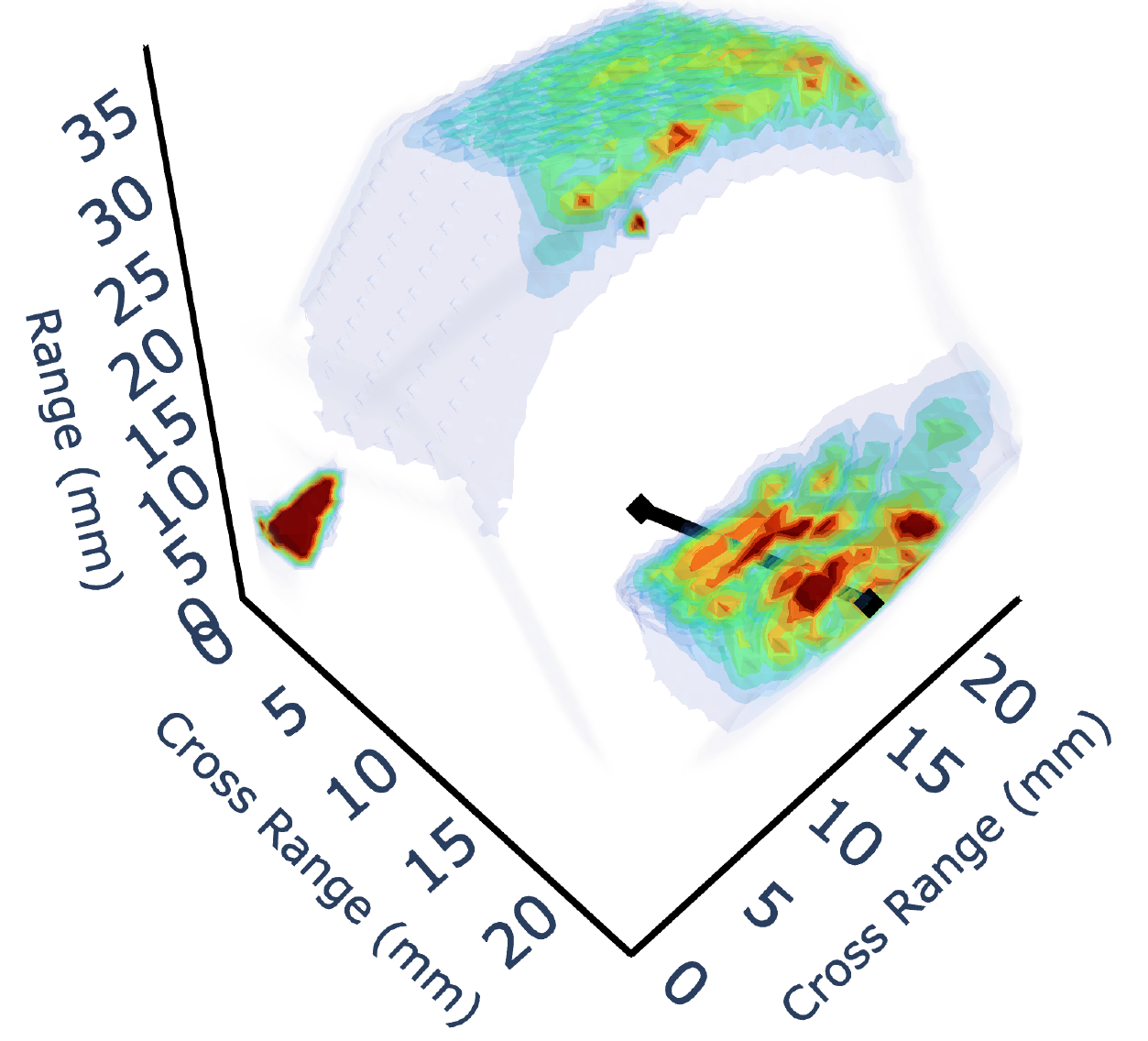}
      \end{minipage}
    \end{subfigure}
    \\
    \begin{subfigure}{0.19\textwidth}
      \centering
      \begin{minipage}{0.99\textwidth}
          \includegraphics[width=\textwidth]{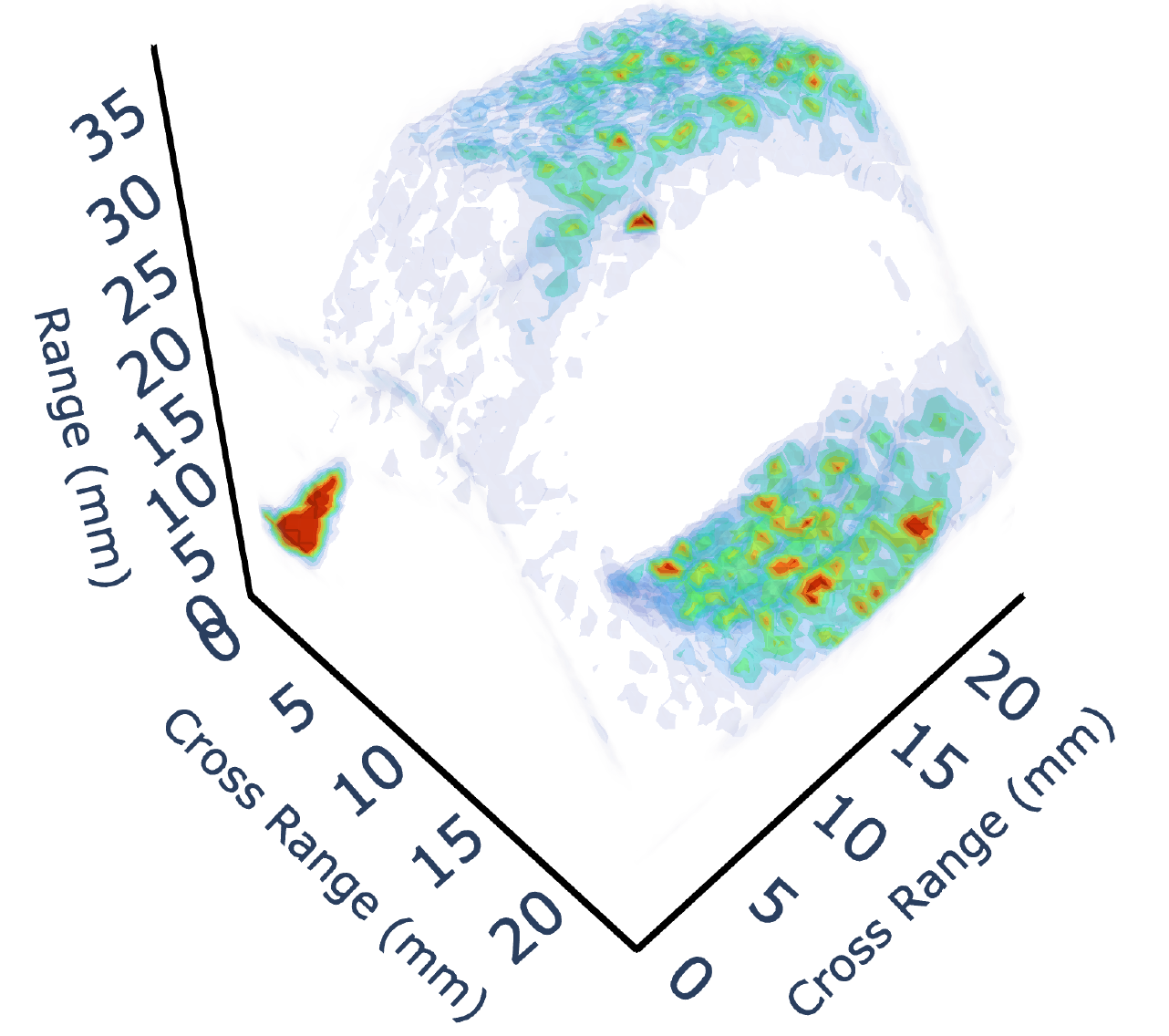}
      \end{minipage}
    \end{subfigure}
    \begin{subfigure}{0.19\textwidth}
      \centering
      \begin{minipage}{0.99\textwidth}
          \includegraphics[width=\textwidth]{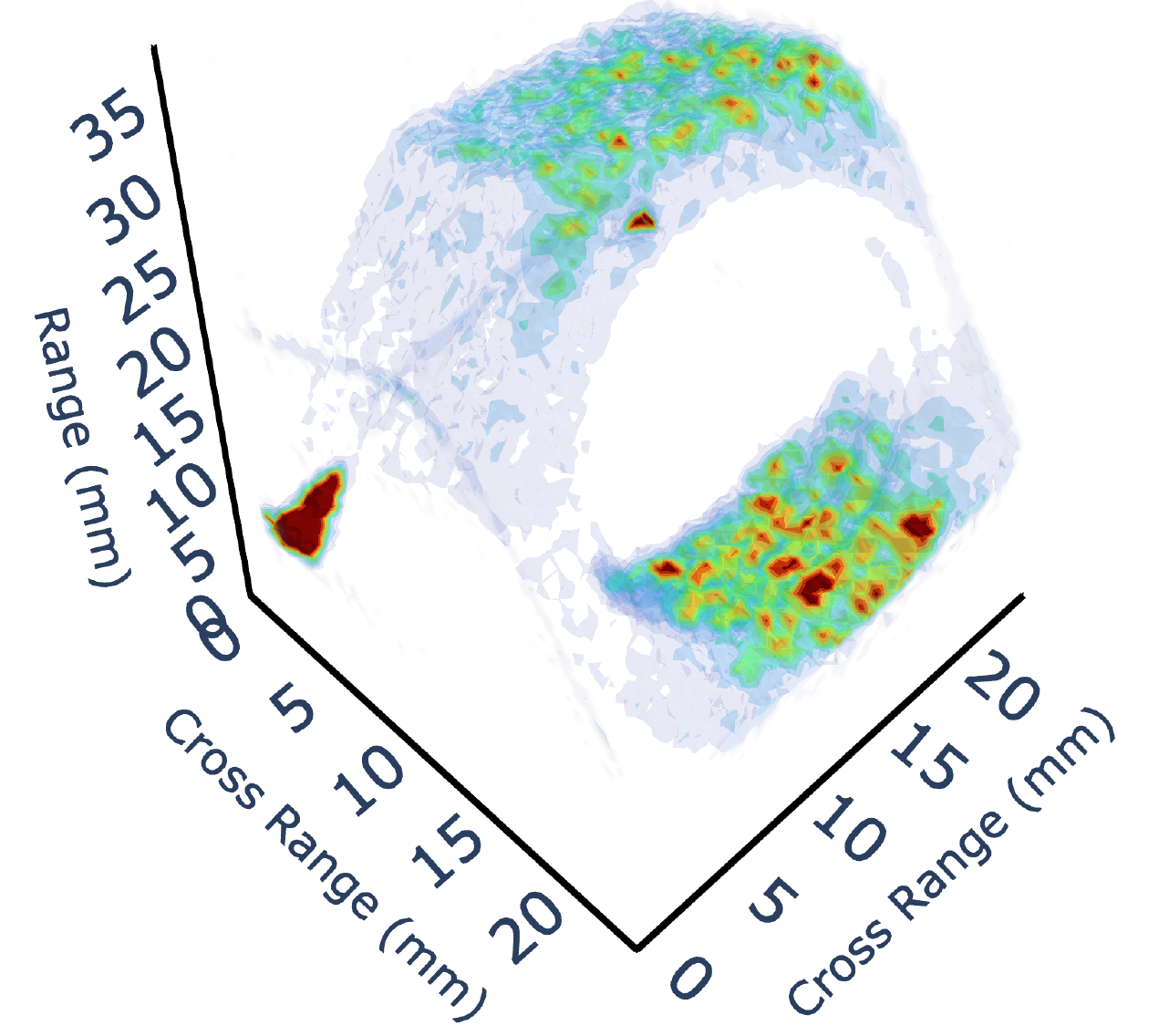}
      \end{minipage}
    \end{subfigure}
    \begin{subfigure}{0.19\textwidth}
      \centering
      \begin{minipage}{0.99\textwidth}
          \includegraphics[width=\textwidth]{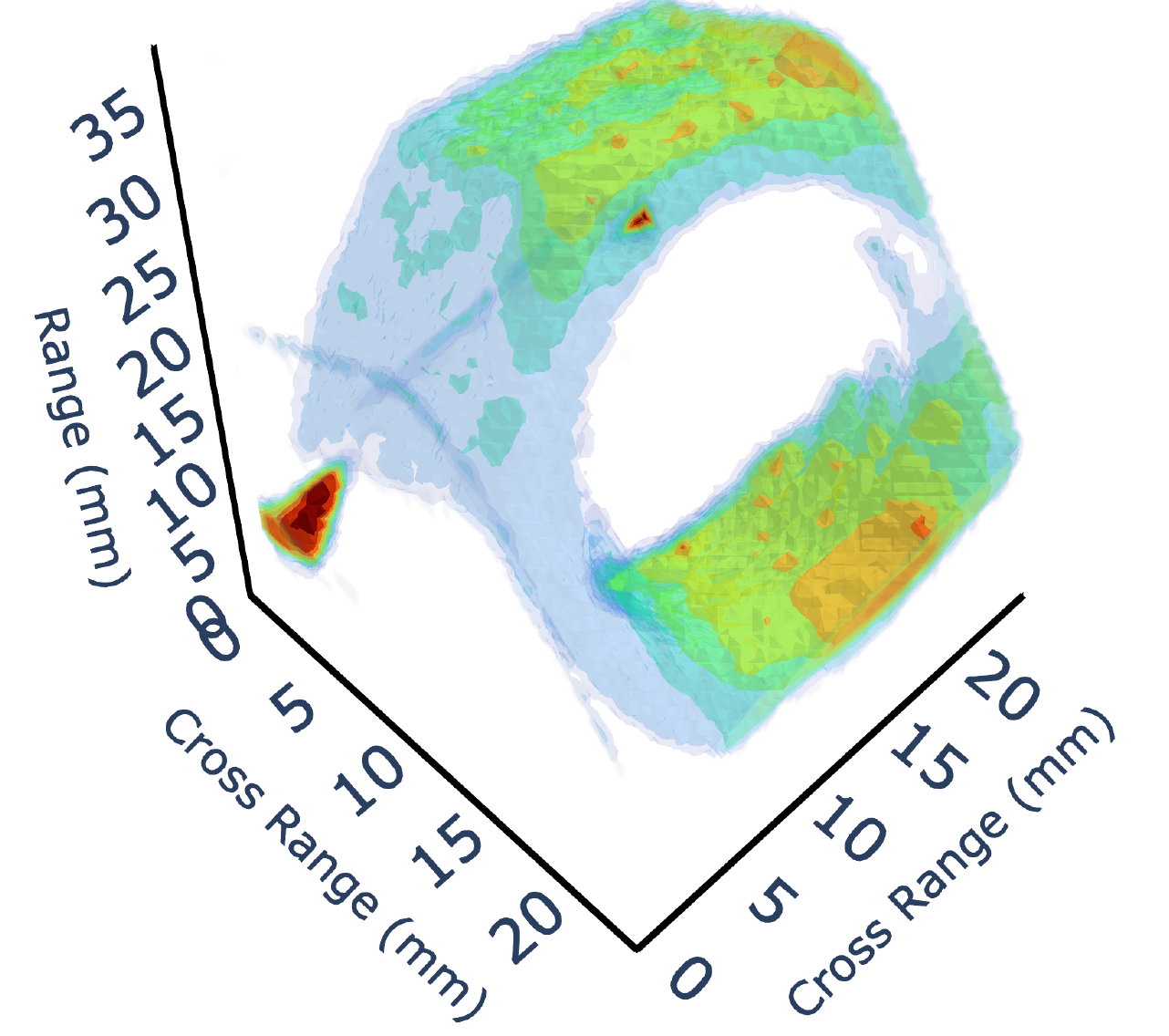}
      \end{minipage}
    \end{subfigure}
    \begin{subfigure}{0.19\textwidth}
      \centering
      \begin{minipage}{0.99\textwidth}
        \includegraphics[width=\textwidth]{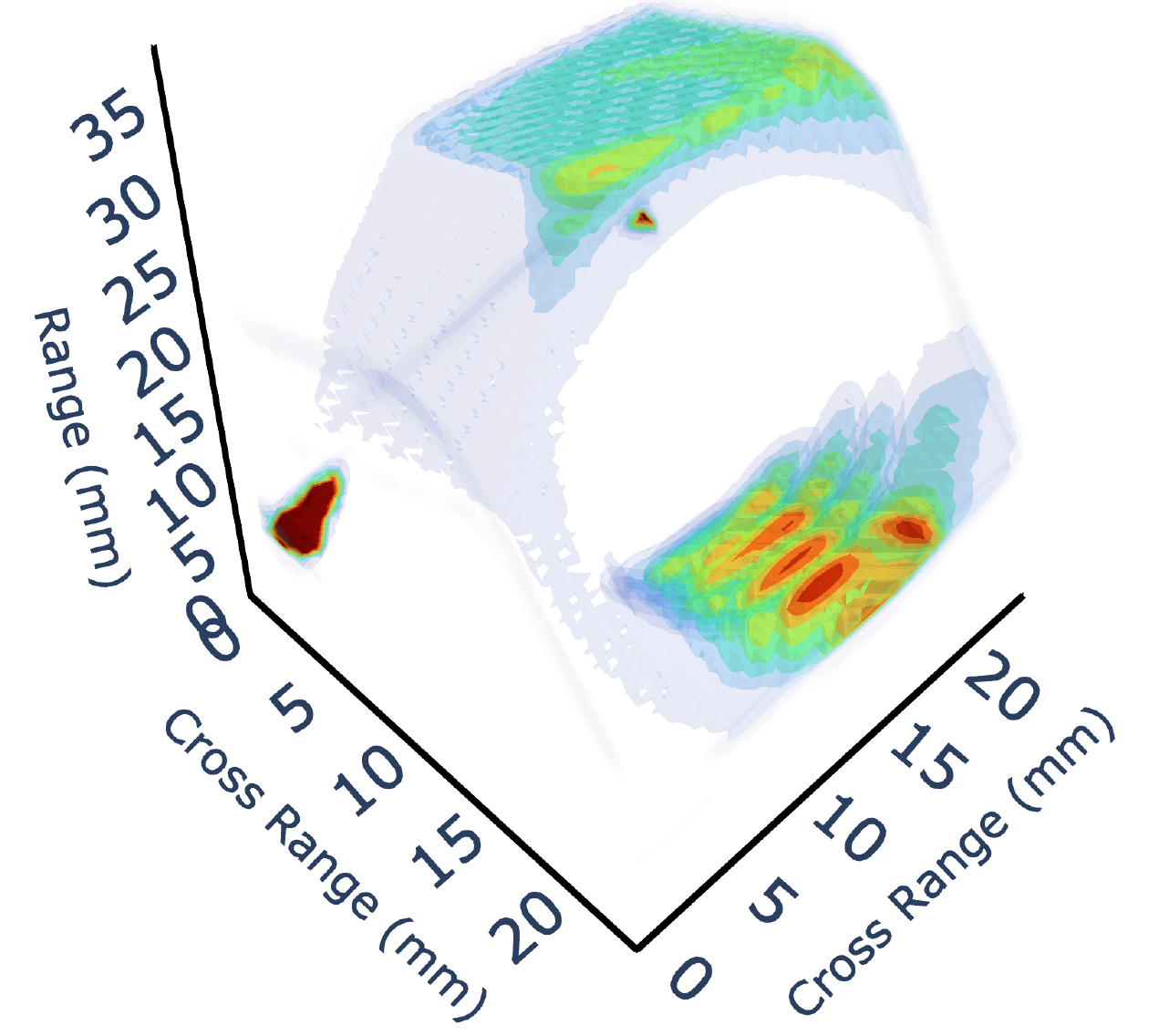}
      \end{minipage}
    \end{subfigure}
    \begin{subfigure}{0.19\textwidth}
      \centering
      \begin{minipage}{0.99\textwidth}
          \includegraphics[width=\textwidth]{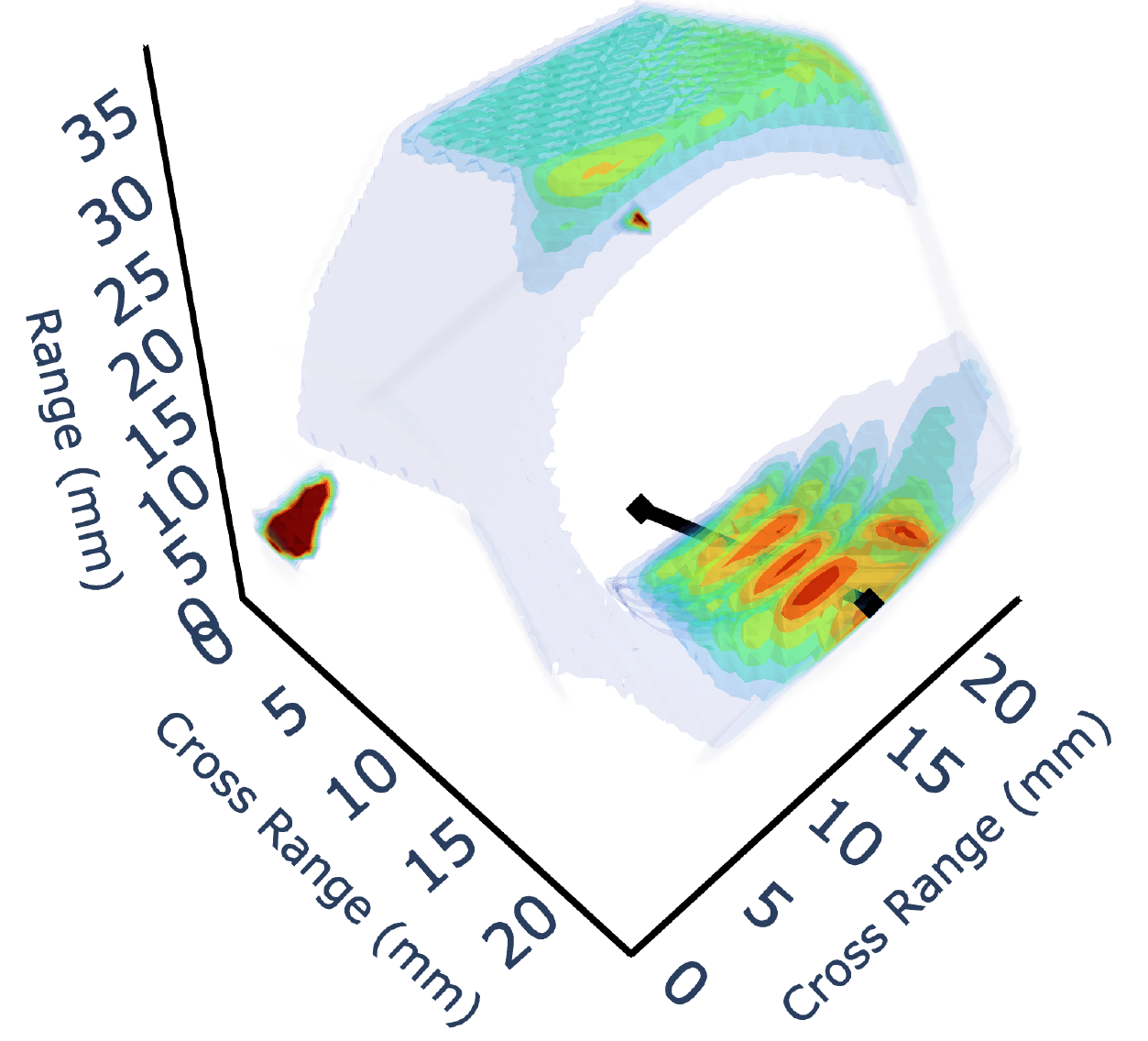}
      \end{minipage}
    \end{subfigure}
    \\
    \begin{subfigure}{0.19\textwidth}
      \centering
      \begin{minipage}{0.99\textwidth}
          \includegraphics[width=\textwidth]{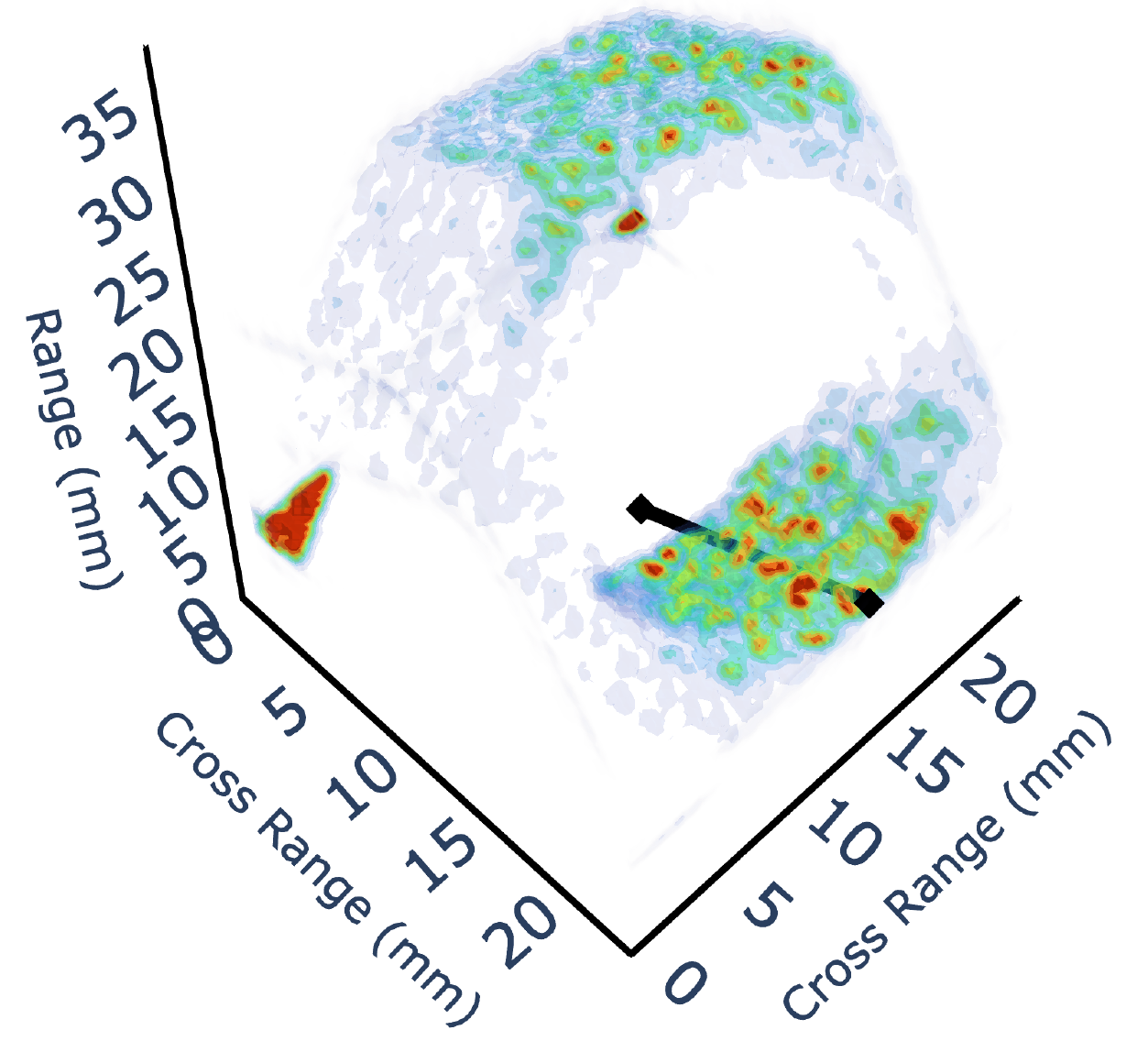}
      \end{minipage}
      \caption{Speckle Average}
    \end{subfigure}
    \begin{subfigure}{0.19\textwidth}
      \centering
      \begin{minipage}{0.99\textwidth}
          \includegraphics[width=\textwidth]{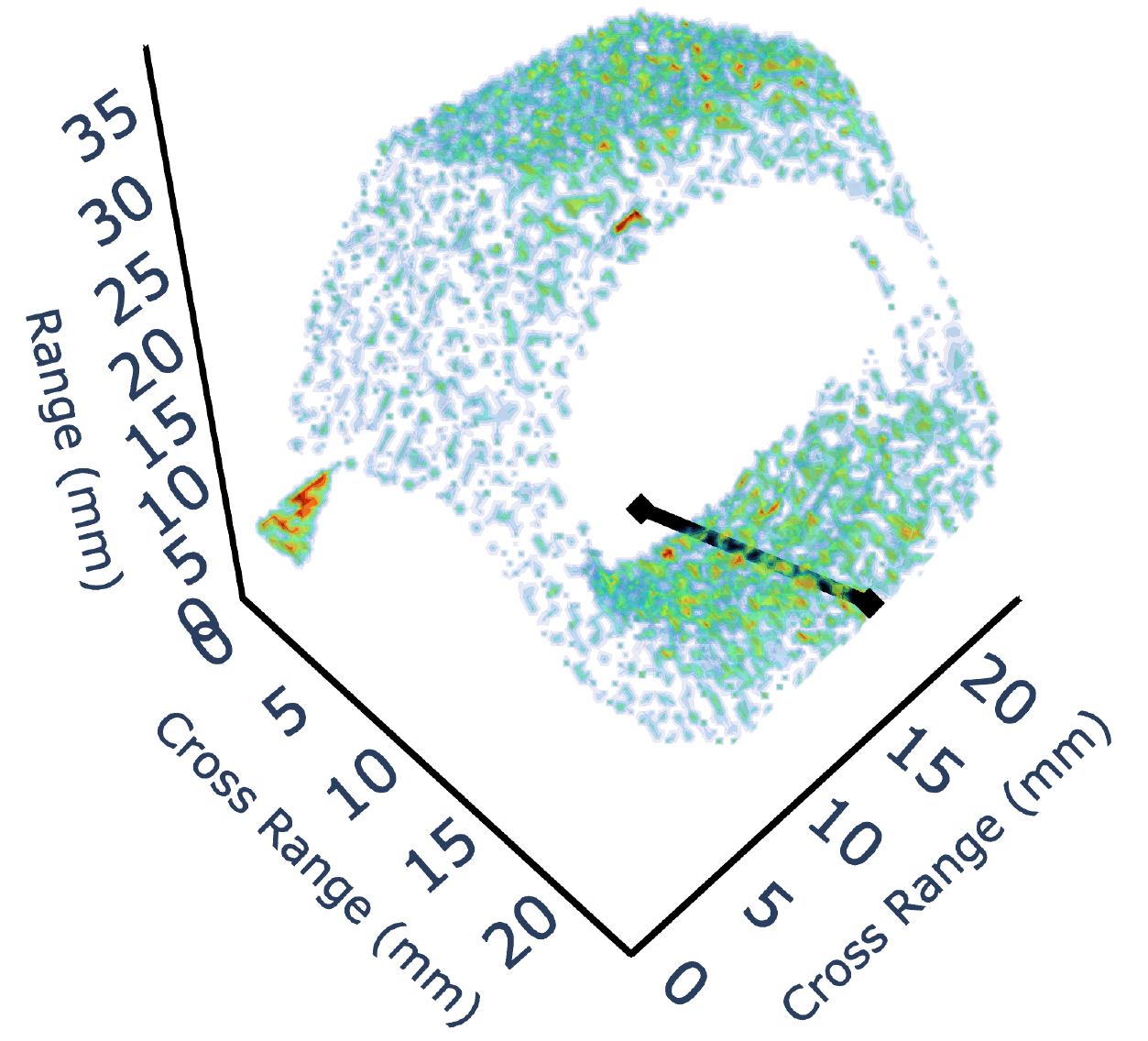}
      \end{minipage}
      \caption{$\ell_{2,1}$-regularization}
    \end{subfigure}
    \begin{subfigure}{0.19\textwidth}
      \centering
      \begin{minipage}{0.99\textwidth}
          \includegraphics[width=\textwidth]{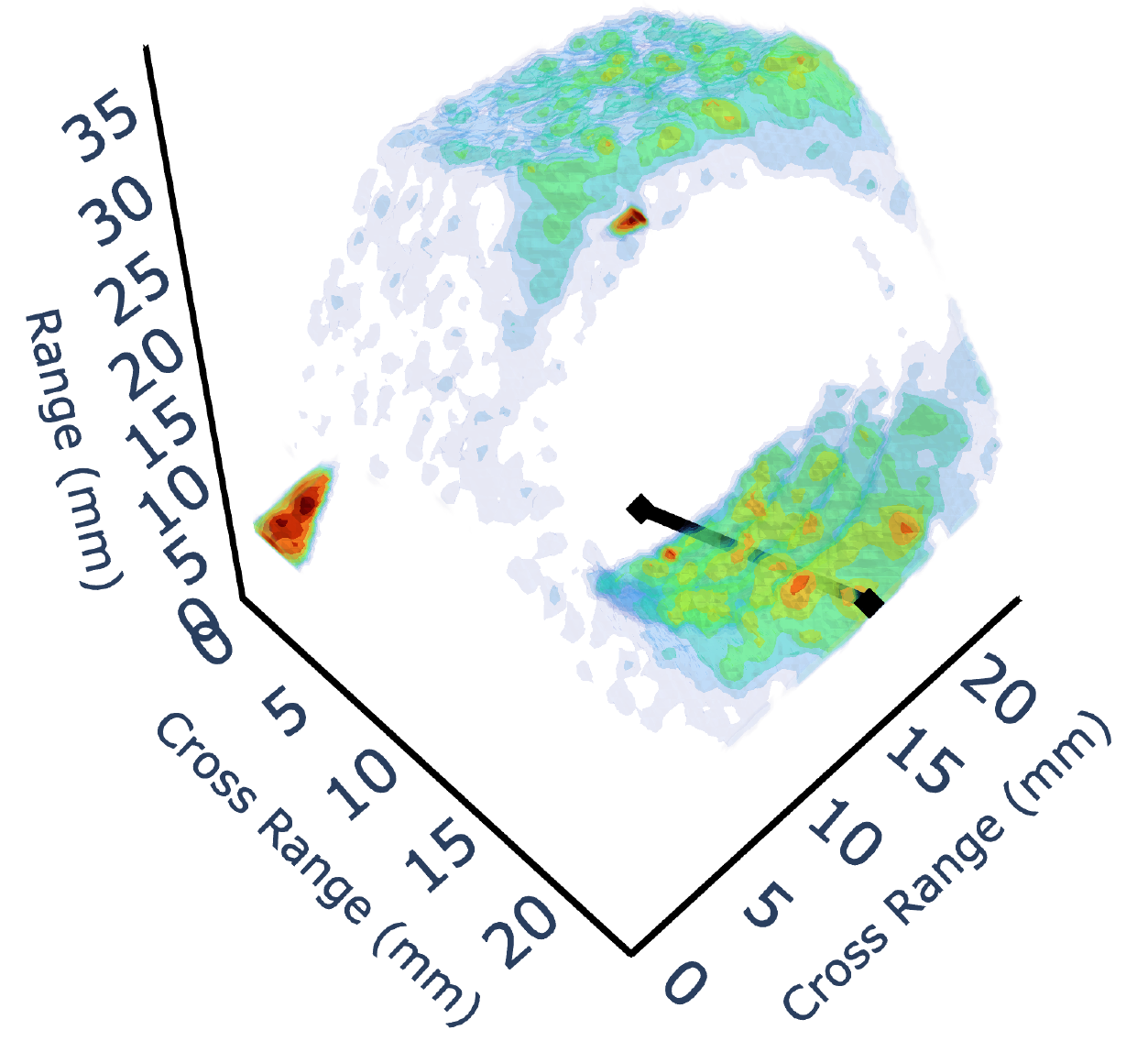}
      \end{minipage}
      \caption{TV-regularization}
    \end{subfigure}
    \begin{subfigure}{0.19\textwidth}
      \centering
      \begin{minipage}{0.99\textwidth}
          \includegraphics[width=\textwidth]{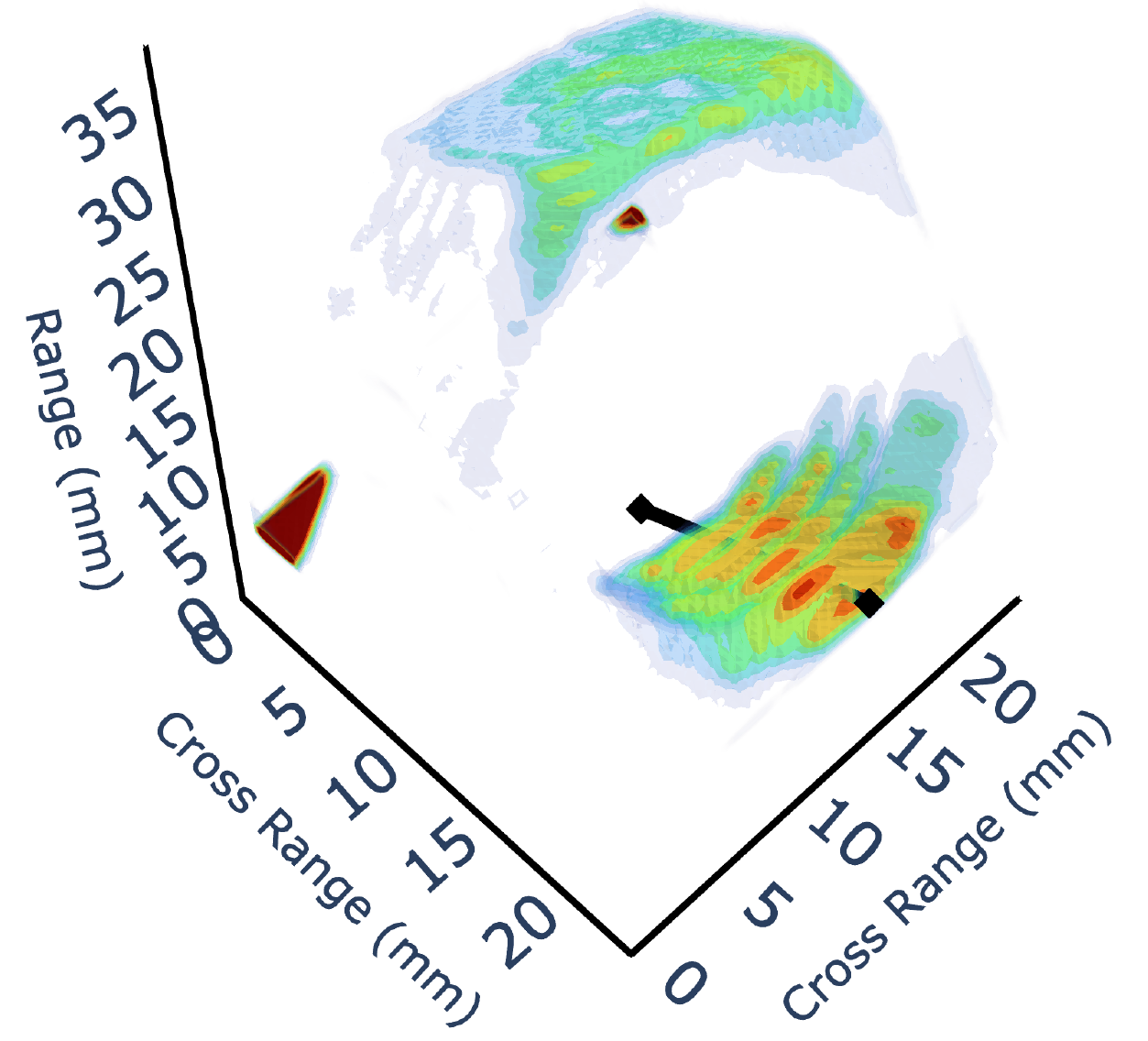}
      \end{minipage}
      \caption{CLAMP (No Aperture)}
    \end{subfigure}
    \begin{subfigure}{0.19\textwidth}
      \centering
      \begin{minipage}{0.99\textwidth}
          \includegraphics[width=\textwidth]{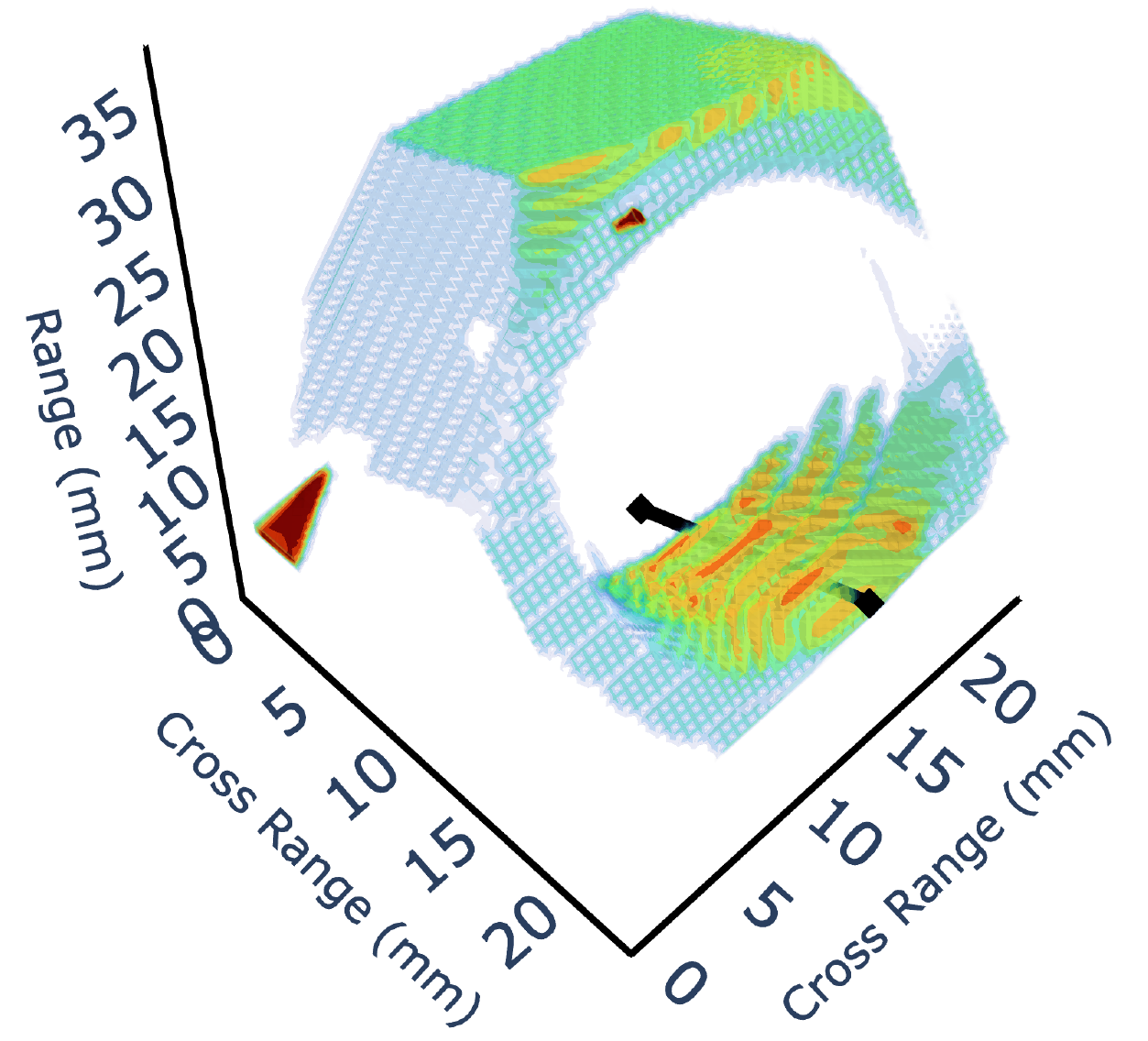}
      \end{minipage}
      \caption{\textbf{CLAMP}}
    \end{subfigure}
  \end{minipage}
  \caption{Reconstructions of the hexagonal nut from experimental data with zero-padding factors, $q=1, 1.5$ and 2. Several reconstructions show a line segment through the threads of the nut which are plotted in Figure~\ref{fig:nut-line}. The CLAMP reconstructions show the best speckle reduction, resolution, and sharpest edges.}\label{fig:nut-grid}
\end{figure*}

\begin{figure}
  \includegraphics[width=0.99\columnwidth]{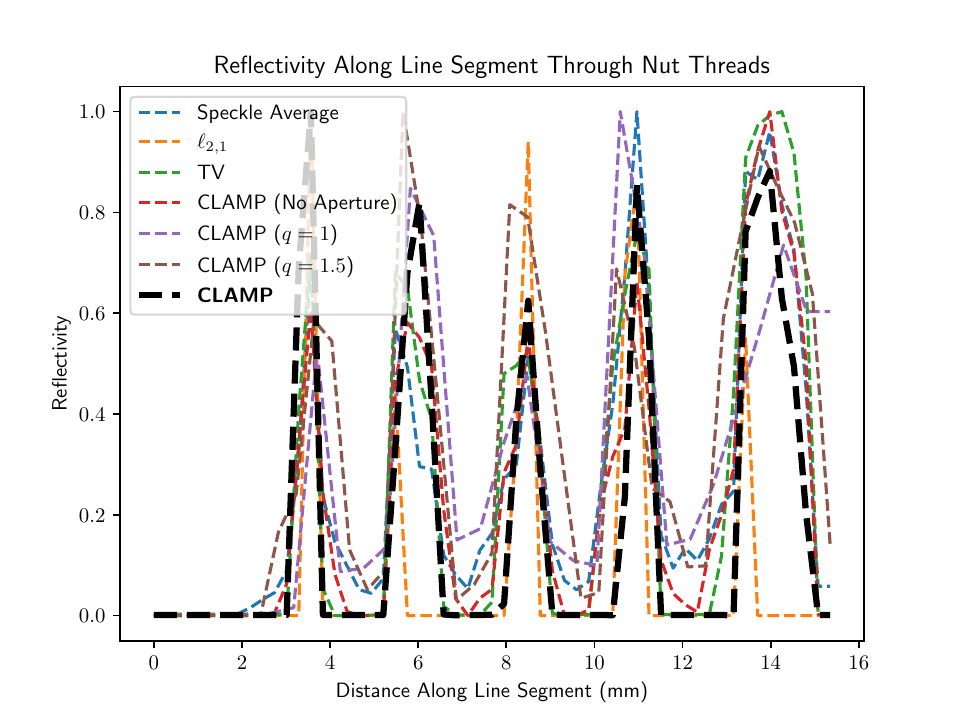}
  \caption{
    A line segment from the hexagonal nut reconstructions shown in Figure~\ref{fig:nut-line}. The CLAMP reconstruction (black) shows the sharpest and most well-defined peaks.
  }\label{fig:nut-line}
\end{figure}

In Figure~\ref{fig:nut-grid}, we show reconstructions of a hexagonal nut from nine looks for $q=1, 1.5,$ and $2$.
In each case, the CLAMP reconstructions show large improvements in reducing speckle noise.
Particularly at $q=2$, the full CLAMP method produces sharp image in which the threads of the nut are clearly visible.
This is further illustrated in Figure~\ref{fig:nut-line}, which shows the reflectivity values along the line segments shown in Figure~\ref{fig:nut-grid}.
The full CLAMP reconstruction shows the consistently sharp and narrow peaks that are expected from a machine-threaded nut, while the other methods show broader and less consistent, well-defined peaks.

\begin{figure}
    \begin{subfigure}{0.49\columnwidth}
    \centering
  \includegraphics[width=0.99\textwidth]{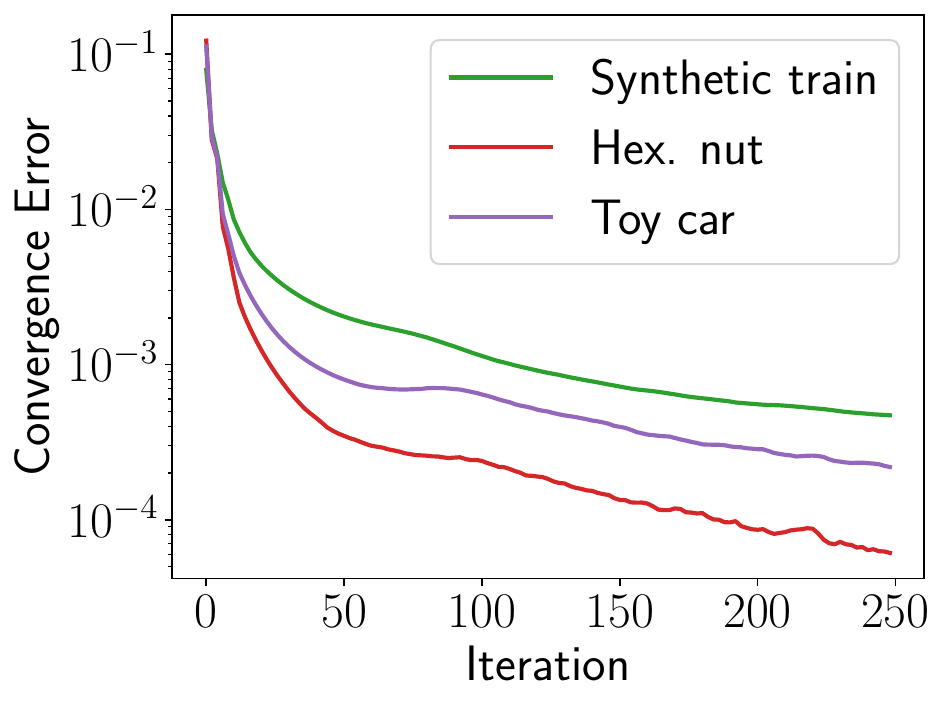}
    \end{subfigure}
        \begin{subfigure}{0.49\columnwidth}
    \centering
      \includegraphics[width=0.99\textwidth]{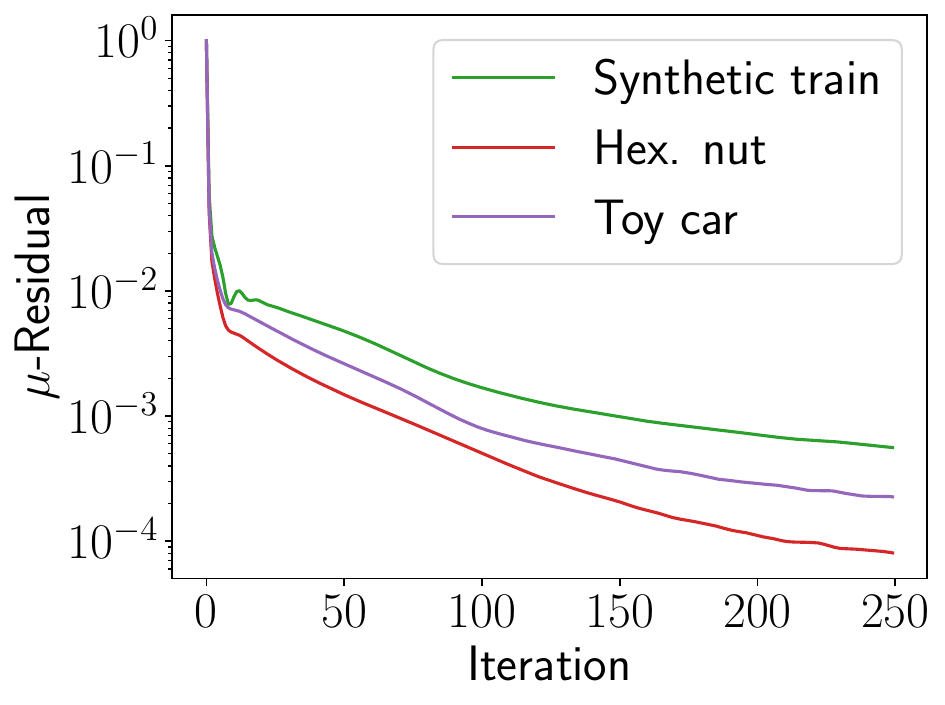}
    \end{subfigure}
      \caption{Convergence error and $\mu$-residual of CLAMP with aperture modeling for synthetic and experimental data.
      CLAMP converged to an error of less than $10^{-3}$ in all cases.
  }\label{fig:convergence-error}
\end{figure}

In Figure~\ref{fig:convergence-error}, we show the the convergence error of CLAMP, defined in~\eqref{eq:mace-convergence-error}, and a normalized $\mu$-residual averaged across looks, defined in~\eqref{eq:residual}, 
at each iteration of CLAMP for each dataset with $q=2$.
All instances of CLAMP converge to an error of less than $10^{-3}$.
Similarly, the $\mu$-residual converges, indicating our iterative approach to solving~\eqref{eq:mu2} is converging to the exact solution.

\section{Conclusion}\label{sec:conclusion}
    In this paper, we introduced CLAMP, a majorized PnP algorithm for multi-look coherent lidar image reconstruction.
    CLAMP is built on the MACE framework and consists of three main components:
    (1) an accurate surrogate physics-based model of coherent lidar, (2) an efficient method to invert the blurring effects of the aperture and (3) a deep 3D image prior model.
    Together, these components enable CLAMP to produce high-resolution images with reduced speckle.

    To demonstrate its effectiveness, we applied CLAMP to synthetic and measured coherent lidar data.
    Our results show that CLAMP produces high-quality 3D images with reduced speckle and noise compared to speckle-averaging, while maintaining fidelity to multi-look data.
    Furthermore, ablating the aperture model in CLAMP highlights the importance of our method in improving resolution.

    Finally, we formalized the use of surrogate optimization in the MACE framework and proved convergence to the exact consensus equilibrium solution for a general class of surrogate functions.
    Our formalization provides a theoretical foundation for the use of majorization-minimization in consensus optimization problems, is general and applicable to many other imaging problems, and includes many existing image reconstruction algorithms as special cases.
\begin{appendices}

\section{Quadratic Surrogate}\label{ap:quadratic-surrogate}

\noindent To further reduce the computational cost of the proximal operator in~\eqref{eq:surrogate-prox}, we follow~\cite{sridharFastAlgorithmsModelbased} and construct a quadratic first-order surrogate, $\widehat{Q}_{\ell}$, of $\hat{f}_{\ell}$, which is then a first-order surrogate of the original function $f_{\ell}$. 
For the remainder of this section, we drop the $\ell$ subscript for simplicity, but note that the following surrogate is formed for each look $\ell = 1,\ldots,L$.

To construct the surrogate, we use the approach in \cite[Section IV]{yuAcceleratedLineSearch2006}, which has two primary components: determining the interval of validity for the surrogate, which depends on $r^\prime$ and $v$, and determining the surrogate itself.  
We first obtain a convex surrogate, $Q$, of $\Hat{f}$ by a first-order approximation at $r'$ of the logarithm in \eqref{eq:EM-surrogate}:
\begin{equation}\label{eq:convex-surrogate}
Q\left(r; r' \right) = \sum\limits_{j=1}^{n} \frac{r_{j}}{r'_{j}} + \frac{1}{r_{j}}\left ( c_j + \left \vert \mu_{j} \right \vert^{2} \right ).
\end{equation}
Treating each $j$ separately, it suffices to form a surrogate the function
\begin{equation} \label{eq:g}
\rho_j(r_j) = \frac{r_j}{r_j^\prime} + \frac{c_j + \left \vert \mu_{j} \right \vert^{2}}{r_j} + \frac{1}{2 \sigma^2} (r_j - v_j)^2.
\end{equation}
We use (11) and (12) of \cite{yuAcceleratedLineSearch2006} to fit a quadratic surrogate of $\rho_j$ by a linear interpolation of its derivatives
given by
\begin{equation}\label{eq:quadratic-surrogate}
    \widehat{\rho}_j(r_j) = \frac{a_{j}}{2}r_j^{2}+ b_{j}r_j.
\end{equation}
Summing over $j$ yields the surrogate $\widehat{Q}$ for $Q$.

The surrogate $\widehat{\rho}_j$ is constructed by a linear interpolation of the derivative $\rho'_j$ at two points: $r_j'$ and $\xi_j$. 
The point $\xi_j$ is chosen so that it is in the direction of the minimum of $\rho_j$ from $r_j'$.
When $\rho_j'(r_j') > 0$, we 
will need $r_j < r_j'$ to minimize $\rho$, while if $\rho_j'(r_j') < 0$ we will need $r_j > r_j'$.  Hence we select a scaling factor $\beta > 1$ and choose $\xi_j = r_j' / \beta$ in the first case and $\xi_j = \beta r_j'$ in the second.  This defines an interval of validity of the surrogate  between $r_j'$ and $\xi_j$, which we denote $\mathcal{I}_j$.  

The reflectivity update in~\eqref{eq:surrogate-prox} is then replaced by the proximal map of $\widehat{Q}$ at $v$, which is
\begin{equation}\label{eq:r-update}
r' \gets \underset{r_j \in \mathcal{I}_j, \forall j}{\operatorname{argmin}}\left\{ \widehat{Q}(r;r') + \frac{1}{2\sigma^{2}}\ \lVert{r-v}\rVert^2 \right\}.
\end{equation} 
Since the objective is separable, this simplifies to
\begin{equation}\label{eq:r-update2}
    r'_j \gets \frac{v_j - \sigma^2b_j}{1+\sigma^2a_j}
\end{equation}
but clipped to lie in the interval $\mathcal{I}_j$.
Finally, we empirically choose $\beta=1+\frac{2}{\log(k)}$ where $k$ is the CLAMP iteration number, which ensures that $\beta$ decays to 1 more slowly than the convergence of $r'$.

\section{Proof of Theorem~\ref{THM:MAJORIZATION}}\label{ap:majorization}
In this section, we prove the convergence of majorized MACE to a fixed point of the original MACE equation when the surrogate functions meet the conditions of Definition~\ref{def:surrogate}.

We begin by stating a few lemmas before proving the main results.
First, a crucial, perhaps intuitive, fact is that the objective function and its surrogate have the same subgradients at the point of approximation. This is stated in Lemma~\ref{cor:subgrad} and is a result of the smoothness assumption in Definition~\ref{def:surrogate}.
\begin{lemma}\label{cor:subgrad}
Let $\hat{f}\in \mathcal{S}_{L,p} \left( f, \xi \right)$ (see Definition~\ref{def:surrogate}).
Then $\partial f(\xi) = \partial \hat{f}(\xi)$.
\end{lemma}
\begin{proof}
    Proof is given in the supplemental material. We note that this result holds even if $\hat{f}$ is not strongly convex. 
\end{proof}
Another critical fact is given in Lemma~\ref{lemma:surrogate}, which is adapted from~\cite[Lemma~2]{luUnifiedAlternatingDirection2016}.
\begin{lemma}\label{lemma:surrogate}
Let $\hat{f}\in \mathcal{S}_{L,p} \left( f, \xi \right)$ (see Definition~\ref{def:surrogate}), and let $u\in \partial \hat{f}(x)$ be any subgradient of $\hat{f}$ at $x$.
Then for any $y$,
\begin{align*}
f(x) - f(y) - &\left \langle u, x-y \right \rangle \\ 
& \leq \frac{1}{2}  \left( L\left\lVert {y-\xi} \right\rVert^{2} - p \left\lVert {x-y} \right\rVert ^{2} \right) .
\end{align*}
\end{lemma}
\begin{proof}
Proof is given in the supplemental material.
\end{proof}

We are now ready to prove part~\ref{THM:PART1} of Theorem~\ref{THM:MAJORIZATION}.

\begin{proof}[Proof of Theorem~\ref{THM:MAJORIZATION} part~\ref{THM:PART1}]
    Let $(\mathbf{w}^*, \mathbf{r}^*)$ be a fixed point of~\eqref{eq:aug-mann}. Since $\mathbf{r}^* = \Hat{\mathbf{F}}^* \mathbf{w}^*$, where $\Hat{\mathbf{F}}^*$ is the stacked proximal operator of surrogates $\Hat{f}_i(x; r_i^*)$, we have
    \begin{equation}
        -\frac{1}{\sigma^2} \left( r_i^* - w_i^* \right) \in \partial \Hat{f}_i(r^*; r^*),
    \end{equation}
    and by Lemma~\ref{cor:subgrad},
    \begin{equation}
        -\frac{1}{\sigma^2} \left( r_i^* - w_i^* \right) \in \partial f_i(r^*),
    \end{equation}
    for $i=1,\dots,N$.

    Thus, $\mathbf{r}^* = \mathbf{F}\mathbf{w}^*$.
    Since $\mathbf{w}^* = \left(2\mathbf{G}-I \right) \left(2 \Hat{\mathbf{F}}^* - I \right)  \mathbf{w}^*$, we get $\mathbf{w}^* = \left(2\mathbf{G}-I \right) \left(2 \mathbf{F} - I \right) \mathbf{w}^*$, proving $\mathbf{w}^*$ is a MACE solution.
\end{proof}

In order to prove Theorem~\ref{THM:MAJORIZATION} part~\ref{THM:PART2}, we prove convergence of an algorithm that is equivalent to Algorithm~\ref{alg:M-MACE}. 
First, it is noted in~\cite{buzzardPlugandPlayUnpluggedOptimizationFree2018} and elsewhere that using Mann iterations to solve the MACE equations, as written in~\eqref{eq:mann}, is equivalent up to a change of variables to a form of ADMM~\cite{boydDistributedOptimizationStatistical2010}\@.

We present the equivalency and prove convergence for $\rho=1/2$, though with extra bookkeeping, and using the appropriate variant of ADMM, one can generalize this result for $\rho \in (0,1)$.

To begin, consider one iteration of Algorithm~\ref{alg:M-MACE}, written as in \eqref{eq:aug-mann}, using $\rho = 1/2$,
\begin{align*}
    \mathbf{w}^{(k+1)} &= \frac{1}{2} (2 \mathbf{G} - \mathbf{I})(2\Hat{\mathbf{F}}^{(k)} -\mathbf{I}) \mathbf{w}^{(k)} + \frac{1}{2} \mathbf{w}^{(k)}, \\
    \mathbf{r}^{(k+1)} &= \Hat{\mathbf{F}}^{(k)} \mathbf{w}^{(k)}, \nonumber
\end{align*}
where we recall that $\Hat{\mathbf{F}}^{(k)} \mathbf{w}^{(k)} = \Hat{\mathbf{F}}^{(k)} (\mathbf{w}^{(k)}; \mathbf{r}^{(k)})$.  
Since $\mathbf{G}$ is linear and $\mathbf{G}\mathbf{G} = \mathbf{G}$, some algebra shows that applying $\mathbf{G}$ to both sides of the first equation yields
$\mathbf{G}\mathbf{w}^{(k+1)} = \mathbf{G}\mathbf{r}^{(k+1)}$, and hence $\mathbf{G}\mathbf{w}^{(k+1)} = \mathbf{G} \Hat{\mathbf{F}}^{(k)} \mathbf{w}^{(k)}$ from the second equation.

Using this to re-express the first equation, recalling that $\overline{\mathbf{r}}$ is one averaged image in $\mathbf{G} \mathbf{r}$, and noting that $\mathbf{w}$ and $\mathbf{r}$ may be updated in either order, the $i$-th component $(r_i, w_i)$ satisfies
\begin{align}
    \label{eq:dr-admm}
    r_i^{(k+1)} &= \Hat{F}_i^{(k)}w_i^{(k)}, \\
    \label{eq:dr-addm1}
    w_i^{(k+1)} &= w_{i}^{(k)} + (\overline{\mathbf{r}}^{(k+1)}- r_{i}^{(k+1)}) \\ 
    \label{eq:dr-admm2}
    & \quad + (\overline{\mathbf{r}}^{(k+1)}- \overline{\mathbf{r}}^{(k)}).
\end{align}

Now, we introduce the variables $u_i^{(k)} = \overline{\mathbf{r}}^{(k)} - w_i^{(k)}$ and $z^{(k+1)} = \overline{\mathbf{r}}^{(k+1)} + \overline{\mathbf{u}}^{(k)}$.
Since $\overline{\mathbf{w}}^{(k+1)} = \overline{\mathbf{r}}^{(k+1)}$, we have $\overline{\mathbf{u}}^{(k+1)} = 0$ for $k \geq 1$, and so for $k \geq 3$, the updates~\eqref{eq:dr-admm}--\eqref{eq:dr-admm2} can be written equivalently as
\begin{align}\label{eq:admm1}
    r_{i}^{(k+1)} &= \Hat{F}_i^{(k)} \left(z^{(k)} - u_i^{(k)} \right), \\
    \label{eq:admmz}
    z^{(k+1)} &= \frac{1}{N}\sum\limits_{i=1}^{N}\left(r_{i}^{(k+1)}+u_{i}^{(k)}\right), \\
    \label{eq:admm2}
    u_{i}^{(k+1)} &= u_{i}^{(k)} + (r_{i}^{(k+1)} - z ^{(k+1)}),
\end{align}
where \eqref{eq:admm2} follows from \eqref{eq:dr-admm2} using $w_i^{(k)} = \overline{\mathbf{r}}^{(k)} - u_i^{(k)}$ and simplifying.  
The updates \eqref{eq:admm1}--\eqref{eq:admm2} are a form of consensus ADMM~\cite[Chapter 7]{boydDistributedOptimizationStatistical2010} to solve the problem,
\begin{equation}\label{eq:consensus-problem}
    \begin{aligned}
    \textrm{minimize} \quad & f(\mathbf{r}) = \sum_{i=1}^{N} f_i(r_i) \\
    \textrm{subject to} \quad & r_i = z,\\
    \end{aligned}
\end{equation}
except the updates of $r_i$ in~\eqref{eq:admm1} are given by proximal maps of surrogate functions instead of proximal maps of $f_i$.

In Theorem~\ref{thm:admm-converge}, we prove that despite this difference, the iterations~\eqref{eq:admm1}-\eqref{eq:admm2}, converge to a Karush-Kuhn-Tucker (KKT) point of~\eqref{eq:consensus-problem}, which is sufficient for optimality.
A KKT point can be found from the Lagrangian of~\eqref{eq:consensus-problem},
\begin{equation}\label{eq:lagrangian}
    \mathcal{L}(\mathbf{r}, z, \mathbf{u}) = f(\mathbf{r}) + \frac{1}{\sigma^2} \sum_{i=1}^N \langle u_i, r_i - z \rangle,
\end{equation}
where $\mathbf{u}$ is the dual variable, and $\langle \cdot, \cdot \rangle$ is the inner-product. 
Specifically, the point $(\mathbf{r}^{*}, z^{*}, \mathbf{u}^{*})$ is a KKT point if and only if
\begin{align}
    -\frac{1}{\sigma^2}u_i^* \in &\partial f_i(r_i^*), \\
    r_{i} ^{*} &= z ^{*},
\end{align}
for $i=1,\dots,N$.

\begin{theorem}\label{thm:admm-converge}
    Assume the hypotheses of Theorem~\ref{THM:MAJORIZATION},  and that~\eqref{eq:consensus-problem} has a KKT point. Then the sequence $\left\{\left(\mathbf{r}^{(k)}, z^{(k)}, \mathbf{u}^{(k)}\right)\right\}$ generated by the iterations~\eqref{eq:admm1}-\eqref{eq:admm2} converges to a KKT point of~\eqref{eq:consensus-problem}.
\end{theorem}

To prove Theorem~\ref{thm:admm-converge}, we first use Proposition~\ref{prop:subgradients} to characterize the subgradients of the surrogate functions.  Then in Lemma~\ref{lem:fejer}, we state the vital property that the distance between the iterates and a KKT point decreases at each iteration.
\begin{prop}\label{prop:subgradients}
    Assume the hypotheses of Theorem~\ref{THM:MAJORIZATION} and let $\left\{\left(\mathbf{r}^{(k)}, z^{(k)}, \mathbf{u}^{(k)}\right)\right\}$ be generated by~\eqref{eq:admm1}-\eqref{eq:admm2}.  Then 
    \begin{align}
    \label{eq:subgradient_f}
    -\frac{1}{\sigma^{2}} \left(r_{i}^{(k+1)}-z^{(k)} + u_{i}^{(k)} \right) &\in \partial \Hat{f}_{i}^{(k)}(r_{i}^{(k+1)}),
    \end{align}
    for $i=1,\ldots,N$.
\end{prop}
\begin{proof}
    This follows from the definition of the proximal map in~\eqref{eq:admm1}.
\end{proof}
To simplify notation, we define 
\begin{equation}
    \label{eq:lambda-half2}
    u_i^{\left(k+\frac{1}{2}\right)} \triangleq u_i^{(k)} + \left(r_i^{(k+1)} - z^{(k)}\right)
\end{equation}
which can be thought of as an intermediate update of the dual variable after the update of $\mathbf{r}$ in~\eqref{eq:admm1} but before the update of $z$ in~\eqref{eq:admmz}. We can then rewrite the subgradients in~\eqref{eq:subgradient_f} as
\begin{align}
    \label{eq:subgradient_f2}
    - \frac{1}{\sigma^{2}} u_i^{\left(k+\frac{1}{2}\right)} &\in \partial \hat{f}^{(k)}_{i}(r_{i}^{(k+1)}). %
\end{align}

\begin{lemma}\label{lem:fejer}
    Assume the hypotheses of Theorem~\ref{THM:MAJORIZATION},  and that~\eqref{eq:consensus-problem} has a KKT point, $\left(\mathbf{r}^*, z^*, \mathbf{u}^*\right)$.
    Then the sequence $\left\{\left(\mathbf{r}^{(k)}, z^{(k)}, \mathbf{u}^{(k)}\right)\right\}$ generated by the iterations~\eqref{eq:admm1}-\eqref{eq:admm2} satisfies
    \begin{equation}
        E^{(k+1)}\left(\mathbf{r}^*, z^*, \mathbf{u}^*\right) \leq E^{(k)}\left(\mathbf{r}^*, z^*, \mathbf{u}^*\right),
    \end{equation}
    where
    \begin{align}
       E^{(k)}\left(\mathbf{r}^*, z^*, \mathbf{u}^*\right) & = \sum_{i=1}^{N} \Big( L_i \left\lVert {r_{i}^{(k)} - r_{i}^*} \right\rVert^{2} \\ 
       & \quad + \frac{1}{\sigma^2} \left\lVert {z^{(k)} - z^*} \right\rVert^{2} + \frac{1}{\sigma^2} \left\lVert {u_i^{(k)} - u_i^*} \right\rVert^{2} \Big).
    \end{align}
\end{lemma}
\begin{proof}
    Proof is given in the supplemental material.
\end{proof}

We are now ready to state and prove convergence of the iterations~\eqref{eq:admm1}-\eqref{eq:admm2}.

\begin{proof}[Proof of Theorem~\ref{thm:admm-converge}]
    Let $(\mathbf{r}^*, z^*, \mathbf{u}^*)$ be a KKT point of~\eqref{eq:consensus-problem}.
    By Lemma~\ref{lem:fejer}, we have that $\left\{\left(\mathbf{r}^{(k)}, z^{(k)}, \mathbf{u}^{(k)}\right)\right\}$ is bounded.
    Let $\left\{\left(\mathbf{r}^{(k_j)}, z^{(k_j)}, \mathbf{u}^{(k_j)}\right)\right\}$ be a convergent subsequence with limit point $\left(\mathbf{r}^{\dagger}, z^{\dagger}, \mathbf{u}^{\dagger}\right)$.
    We will show that this limit point is also a KKT point of~\eqref{eq:consensus-problem}.

    From (43)-(44) in the proof of Lemma~\ref{lem:fejer}, the iterates satisfy
    \begin{align}
        \frac{1}{2\sigma^2} & \sum_{i=1}^{N} \left\lVert {z^{(k+1)} - z^{(k)}} \right\rVert ^{2} + \left\lVert {u_i^{(k+1)} - u_i^{(k)}} \right\rVert ^{2} \label{eq:z-u-diff} \\
        & \leq \frac{1}{2} \left( E^{(k)}\left(\mathbf{r}^*, z^*, \mathbf{u}^*\right) - E^{(k+1)}\left(\mathbf{r}^*, z^*, \mathbf{u}^*\right) \right) \label{eq:E-diff} .
    \end{align}
    Using Lemma~\ref{lem:fejer}, summing both sides of~\eqref{eq:z-u-diff}-\eqref{eq:E-diff} over all $k$, and discarding the telescoping terms, we get $\left\lVert z^{(k+1)} - z^{(k)} \right\rVert \to 0$ and $\left\lVert u_i^{(k+1)} - u_i^{(k)} \right\rVert \to 0$ as $k \to \infty$. 
    Since $z^{(k_j)} \rightarrow z^\dagger$ as $j \rightarrow \infty$, we can set $k=k_j$ and take $j \rightarrow \infty$ to get $z^{(k_j+1)} \rightarrow z^\dagger$.  Likewise, $u_i^{(k_j)} \rightarrow u_i^\dagger$.
    
    Rearranging \eqref{eq:admm2}, we have 
    \begin{equation*}
        r_i^{(k_j+1)} = u_i^{(k_j+1)} - u_i^{(k_j)} + z^{(k_j+1)},
    \end{equation*}
    so $r_i^{(k_j+1)}$ converges to $z^\dagger$, which we denote also as $r_i^\dagger$.  
    Lastly, $u_i^{(k_j+\frac{1}{2})} \to u_i^{\dagger}$ can be seen by taking the limit of \eqref{eq:lambda-half2} along the subsequence $k = k_j$.
    
    Since $-\frac{1}{\sigma^2} u_i^{(k_j+\frac{1}{2})} \in \partial \Hat{f}_i^{(k_j)}(r_i^{(k_j+1)})$, we have by Lemma~\ref{lemma:surrogate} that for all $\mathbf{r}$,
    \begin{equation}
        f(\mathbf{r}^{(k_j+1)}) - f(\mathbf{r}) + \frac{1}{\sigma^2} \sum_{i=1}^{N} \left\langle u_i^{(k_j+\frac{1}{2})}, r_i^{(k_j+1)} - r_i \right\rangle \leq 0.
    \end{equation}
    Since $f$ is convex and finite on ${\mathbb R}^n$ and $r_i^{(k_j+1)} \to r_i^{\dagger}$, we have that $f(\mathbf{r}^{(k_j+1)}) \to f(\mathbf{r}^{\dagger})$, and so taking the limit of the above inequality along the subsequence gives
    \begin{equation}
        f(\mathbf{r}^{\dagger}) - f(\mathbf{r}) + \frac{1}{\sigma^2} \sum_{i=1}^{N} \left\langle u_i^{\dagger}, r_i^{\dagger} - r_i \right\rangle \leq 0,
    \end{equation}
    which implies that $-\frac{1}{\sigma^2} u_i^{\dagger} \in \partial f_i(r_i^{\dagger})$ for $i=1,\dots,N$.
    This completes the proof that $\left(\mathbf{r}^{\dagger}, z^{\dagger}, \mathbf{u}^{\dagger} \right)$ is a KKT point.

    Finally, we show that $\left(\mathbf{r}^{\dagger}, z^{\dagger}, \mathbf{u}^{\dagger} \right)$ is the only limit point of the sequence $\left\{\left(\mathbf{r}^{(k)}, z^{(k)}, \mathbf{u}^{(k)}\right)\right\}$.
    First, we replace $\left(\mathbf{r}^*, z^*, \mathbf{u}^*\right)$ with $\left(\mathbf{r}^{\dagger}, z^{\dagger}, \mathbf{u}^{\dagger}\right)$ in Lemma~\ref{lem:fejer}.
    Then for any $k > k_j$, we have
    \begin{equation}
        E^{(k)}\left(\mathbf{r}^{\dagger}, z^{\dagger}, \mathbf{u}^{\dagger}\right) \leq E^{(k_j)}\left(\mathbf{r}^{\dagger}, z^{\dagger}, \mathbf{u}^{\dagger}\right).
    \end{equation}
    Since the right hand side converges to 0 as $j\to \infty$, the definition of $E^{(k)}$ implies that $\left(\mathbf{r}^{(k)}, z^{(k)}, \mathbf{u}^{(k)}\right) \to \left(\mathbf{r}^{\dagger}, z^{\dagger}, \mathbf{u}^{\dagger}\right)$ as $k \to \infty$.
\end{proof}

To end, we prove Theorem~\ref{THM:MAJORIZATION} part~\ref{THM:PART2}.
\begin{proof}[Proof of Theorem~\ref{THM:MAJORIZATION} part~\ref{THM:PART2}]
    As noted after equation \eqref{eq:dr-admm2}, by introducing the variables
        $u_i^{(k)} = \overline{\mathbf{r}}^{(k)} - w_i^{(k)}$, and 
        $z^{(k+1)} = \overline{\mathbf{r}}^{(k+1)} + \overline{\mathbf{u}}^{(k)}$,
    the iterates of~\eqref{eq:aug-mann} are equivalent to the iterates~\eqref{eq:admm1}-\eqref{eq:admm2}.
    By Theorem~\ref{thm:admm-converge}, $\left(\mathbf{r}^{(k)}, z^{(k)}, \mathbf{u}^{(k)}\right)$ converges to a KKT point, $\left(\mathbf{r}^{*}, z^{*}, \mathbf{u}^{*}\right)$, of~\eqref{eq:consensus-problem}. 
    Resubsituting the change of variables, we get $(\mathbf{w}^*, \mathbf{r}^*)$ is a fixed point of the system~\eqref{eq:aug-mann}, which by Theorem~\ref{THM:MAJORIZATION} part~\ref{THM:PART1} is a solution to the exact MACE formulation.
\end{proof}

\end{appendices}

\bibliography{zotero}
\bibliographystyle{ieeetr}

\clearpage

\title{Supplemental Material for ``CLAMP: Majorized Plug-and-Play for Coherent 3D Lidar Imaging''}

{\maketitle}
\setcounter{page}{1}  %

\section*{Proof of Lemma~\ref{cor:subgrad}}

\noindent
{\bf Lemma~\ref{cor:subgrad}. }
{\em
    Let $\hat{f}\in \mathcal{S}_{L,p} \left( f, \xi \right)$.
    Then $\partial f(\xi) = \partial \hat{f}(\xi)$.
}
\begin{proof}[Proof of Lemma 1]
    The smoothness condition in Definition~\ref{def:surrogate} states $\nabla e (\xi) = 0$, where $e \triangleq f - \hat{f}$.
    Hence, the directional derivatives of $f$ and $\hat{f}$ at $\xi$ must be equal in all directions.
    That is, for all directions $h$, 
    \begin{equation}
        f^{\prime}(\xi; h) = \hat{f}^{\prime}(\xi; h)
    \end{equation}
    where
    \begin{equation}
        f'(\xi; h) = \lim_{t\to 0^{+}} \frac{f(\xi + th) - f(\xi)}{t},
    \end{equation}
    and
    \begin{equation}
        \hat{f}^{\prime}(\xi; h) = \lim_{t\to 0^{+}} \frac{\hat{f}(\xi + th) - \hat{f}(\xi)}{t}.
    \end{equation}
    Since subgradients can be specified by the set of all directional derivatives~\citeSupp[Theorem 8.8 and 8.10]{recht-wright}, we have
    \begin{align}
        \partial f (\xi) &= \left\{ u \mid \langle u, h \rangle \leq f'(\xi; h), \forall h\right\} \\
        &= \left\{ u \mid \langle u, h \rangle \leq \hat{f}^{\prime}(\xi; h), \forall h\right\} \\
        & = \partial \hat{f} (\xi).
    \end{align}
\end{proof}

\section*{Proof of Lemma~\ref{lemma:surrogate}}

Before proving Lemma~\ref{lemma:surrogate}, we state in Lemma~\ref{prop:lipschitz} a useful property of the error function, $e \triangleq f - \hat{f}$.
\begin{lemma}\label{prop:lipschitz}
If $e$ is differentiable with an $L$-Lipschitz continuous gradient, then for any $\xi,y$,
\begin{equation}
    \lvert e(y) - e(\xi) - \left\langle \nabla e(\xi), y-\xi \right\rangle \rvert \leq \frac{L}{2} \left\lVert y-\xi \right\rVert^2.
\end{equation}
\end{lemma}
\begin{proof}
This is Lemma 1.2.3 in~\citeSupp{nesterovIntroductoryLecturesConvex2004}. 
\end{proof}

\noindent
{\bf Lemma~\ref{lemma:surrogate}. }
{\em 
    Let $\hat{f}\in \mathcal{S}_{L,p} \left( f, \xi \right)$ (see Definition~\ref{def:surrogate}), and let $u\in \partial \hat{f}(x)$ be any subgradient of $\hat{f}$ at $x$.
    Then for any $y$,
    \begin{align*}
    f(x) - f(y) - &\left \langle u, x-y \right \rangle \\ 
    & \leq \frac{1}{2}  \left( L\left\lVert {y-\xi} \right\rVert^{2} - p \left\lVert {x-y} \right\rVert ^{2} \right) .
    \end{align*}
}
    
\begin{proof}[Proof of Lemma 2]
    We begin with the fact that $\hat{f}$ is $p$-strongly convex. 
    There are several equivalent definitions of strong convexity, see for example~\citeSupp[Section~9.1.2]{boydConvexOptimization2004}. 
    One such definition is that for all $x,y$ and subgradients $u\in \partial \hat{f}(x)$, we have 
    \begin{equation}
    \label{eq:leamma1-eq0}
        \hat{f}(y) \geq \hat{f}(x) - \left \langle u, x-y \right \rangle +\frac{p}{2}\left\lVert {y-x} \right\rVert ^{2}.
    \end{equation}
    Using this, we have
    \begin{align}
    \label{eq:leamma1-eq1}
        f(x) &\leq \hat{f}(x) \\
        \label{eq:leamma1-eq2}
        & \leq \hat{f}(y) + \left \langle u, x-y \right \rangle - \frac{p}{2}\left\lVert {y-x} \right\rVert ^{2},
    \end{align}
    where~\eqref{eq:leamma1-eq1} is the majorization property of $\hat{f}$, and~\eqref{eq:leamma1-eq2} follows from Equation~\eqref{eq:leamma1-eq0}.
    Now, defining $e(y) \triangleq \hat{f}(y) - f(y)$, the fact that $\hat{f}\in \mathcal{S}_{L,p} \left( f, \xi \right)$ implies that $e(\xi) = 0$ and $\nabla e(\xi) = 0$.  Hence Lemma~\ref{prop:lipschitz} implies that $e(y) \leq \frac{L}{2}\left\lVert {y-\xi} \right\rVert^2$. Therefore, continuing the inequalities in~\eqref{eq:leamma1-eq1}-\eqref{eq:leamma1-eq2}, we have
    \begin{align*}
        f(x) &\leq f(y)  + \frac{L}{2}\left\lVert {y-\xi} \right\rVert ^{2} + \left \langle u, x-y \right \rangle - \frac{p}{2}\left\lVert {y-x} \right\rVert ^{2},
    \end{align*}
    completing the proof.
\end{proof}

\section*{Proof of Lemma~\ref{lem:fejer}}

Before we prove Lemma~\ref{lem:fejer}, we prove two additional useful lemmas.

\begin{lemma}\label{lem:saddle-point-bnd}
    Assume the hypotheses of Proposition~\ref{prop:subgradients}. Then for any $(\mathbf{r}, z )$, we have
    \begin{align*}
    & f(\mathbf{r}^{(k+1)})-f(\mathbf{r})  \\
    & \quad + \frac{1}{\sigma^2} \sum\limits_{i=1}^{N}\left \langle u_i^{\left(k+\frac{1}{2}\right)}, \left( r_{i}^{(k+1)} - r_i \right) - \left( z ^{(k+1)} -z \right) \right \rangle\\
    & \leq \frac{1}{2} \sum\limits_{i=1}^{N} L_{i}  \left(\left\lVert {r_{i}^{(k)} - r_{i}} \right\rVert ^{2} -  \left\lVert {r_{i}^{(k+1)} - r_{i}} \right\rVert ^{2}  \right) \\ 
    & \quad + \frac{N}{2\sigma^2}\left\lVert {z^{(k)} - z} \right\rVert ^{2} - \frac{N}{2\sigma^2}\left\lVert {z^{(k+1)} - z} \right\rVert ^{2} \\
    & \quad - \frac{N}{2\sigma^2}\left\lVert {z^{(k+1)} - z^{(k)}} \right\rVert ^{2}.
    \end{align*}
\end{lemma}
\begin{proof}
    Since $f(\mathbf{r}) = \sum_{i=1}^N f_i(r_i)$, we create a bound for each $f_i$, and then sum over $i=1,\ldots, N$ to get the final result.
    First, we apply Lemma~\ref{lemma:surrogate} to $f_i$, mapping the variables $\xi$ to $r_i^{(k)},$ $x$ to $r_i^{(k+1)},$ $y$ to $r_i,$ and $u$ to the subgradient in Equation~\eqref{eq:subgradient_f2}. 
    This gives
    \begin{align}
    \label{eq:ineq-f0}
    f_{i}(r_{i}^{(k+1)}) & - f_{i}(r_{i}) + \frac{1}{\sigma^2}\left \langle u_i^{\left(k+\frac{1}{2}\right)}, r_{i}^{(k+1)} - r_{i} \right \rangle \\
    &\leq  \frac{1}{2}  \left( L_{i} \left\lVert {r_{i}^{(k)} - r_{i}} \right\rVert ^{2} - p_{i} \left\lVert {r_{i}^{(k+1)} - r_{i}} \right\rVert ^{2} \right)\\
    \label{eq:ineq-f2}
    &\leq \frac{L_{i}}{2}  \left(  \left\lVert {r_{i}^{(k)} - r_{i}} \right\rVert ^{2} - \left\lVert {r_{i}^{(k+1)} - r_{i}} \right\rVert ^{2} \right),
    \end{align}
    where we have used the assumption $p_{i}\geq L_{i}$ from Definition~\ref{def:surrogate} to obtain the second inequality.
    
    Next, we obtain a similar expression with $z$.
    First, note that by~\eqref{eq:lambda-half2} and \eqref{eq:admmz}, we have
    \begin{equation} \label{eq:sum_ui_half}
        \frac{1}{N}\sum_{i=1}^{N} u_i^{\left(k+\frac{1}{2}\right)} = z^{(k+1)} - z^{(k)}.
    \end{equation}
    We then use \eqref{eq:sum_ui_half} to rewrite the inner product of the subgradient in~\eqref{eq:subgradient_f2} with $\left(z ^{(k+1)} - z \right)$ as
    \begin{align}\label{eq:z-norm1}
        \frac{1}{\sigma^2} \sum\limits_{i=1}^{N} &  \left \langle u_i^{\left(k+\frac{1}{2}\right)}, z ^{(k+1)} - z  \right \rangle \\
        &= \frac{N}{\sigma^2} \left \langle z^{(k+1)} - z ^{(k)}, z^{(k+1)} - z \right \rangle \\
        &= \frac{N}{2\sigma^2}\Big(\left\lVert {z^{(k+1)} - z} \right\rVert ^{2} - \left\lVert {z^{(k)}-z} \right\rVert ^{2} \\
        \label{eq:z-norm2}
        & \quad  + \left\lVert {z^{(k+1)} - z^{(k)}} \right\rVert ^{2} \Big),
    \end{align}
    where we have used the identity $2\left \langle a-b, c - d\right \rangle =  \left\lVert {a-d} \right\rVert ^{2}-  \left\lVert {a-c} \right\rVert ^{2} -  \left\lVert {b-d} \right\rVert ^{2} +  \left\lVert {b-c} \right\rVert ^{2}$. 

    To prove the lemma, we sum the inequalities in~\eqref{eq:ineq-f0}-\eqref{eq:ineq-f2} over $i=1, \ldots, N$ and subtract~\eqref{eq:z-norm1}-\eqref{eq:z-norm2}.
\end{proof}

\begin{lemma}\label{lem:inner-prod}
    Assume the hypotheses of Proposition~\ref{prop:subgradients}. Then for any $\mathbf{u}$, 
    \begin{align}
        \frac{1}{\sigma^2}\sum_{i=1}^{N} & \left\langle u_i^{(k+\frac{1}{2})}-u_i, r_i^{(k+1)} - z^{(k+1)} \right\rangle \\
        &= \frac{1}{2\sigma^2} \sum_{i=1}^{N} \Big( \left\lVert {u_i^{(k+1)} - u_i} \right\rVert ^{2} - \left\lVert {u_i^{(k)}- u_i} \right\rVert ^{2} \\ 
        & \mbox{\hspace{3.9cm}} + \left\lVert {u_i^{(k+1)} - u_i^{(k)}} \right\rVert^{2} \Big)
    \end{align}
\end{lemma}
\begin{proof}
    From~\eqref{eq:admm2}, we have $u_i^{(k+1)} - u_i^{(k)} = r_i^{(k+1)} - z^{(k+1)}$. Using this right hand side directly plus solving for $u_i^{(k+1)}$ to simplify~\eqref{eq:lambda-half2}, we can rewrite the inner products in the lemma statement as
    \begin{align}
        \Bigl< u_i^{(k+\frac{1}{2})}& -u_i, r_i^{(k+1)} - z^{(k+1)} \Bigr> \\
        \label{eq:u-bounds}
        &= \left\langle u_i^{(k+1)} - u_i, u_i^{(k+1)} - u_i^{(k)} \right\rangle \\
        \label{eq:zero-inner-prod}
        & \quad + \left\langle z^{(k+1)} - z^{(k)}, u_i^{(k+1)} - u_i^{(k)} \right\rangle.
    \end{align}
    As noted just before \eqref{eq:admm1}, $\overline{{\mathbf u}}^{(k)} = 0$ for $k \geq 2$, so summing \eqref{eq:zero-inner-prod} over $i$ as in the lemma statement yields 0, hence this term may be dropped.  
    The term in~\eqref{eq:u-bounds} can be rewritten as
    \begin{align}
        \Bigl< u_i^{(k+1)} & - u_i, u_i^{(k+1)} - u_i^{(k)} \Bigr> \\
        &= \frac{1}{2} \Big( \left\lVert {u_i^{(k+1 )}- u_i} \right\rVert ^{2} - \left\lVert {u_i^{(k)} - u_i} \right\rVert ^{2} \\ 
        & \quad + \left\lVert {u_i^{(k+1)} - u_i^{(k)}} \right\rVert ^{2} \Big),
    \end{align}
    which follows from the identity $2\left \langle a-b, c - d\right \rangle =  \left\lVert {a-d} \right\rVert ^{2}-  \left\lVert {a-c} \right\rVert ^{2} -  \left\lVert {b-d} \right\rVert ^{2} +  \left\lVert {b-c} \right\rVert ^{2}$. Summing over $i=1,\dots, N$ and multiplying by $\frac{1}{\sigma^2}$ completes the proof.
\end{proof}

Now, we prove Lemma~\ref{lem:fejer}.

\noindent
{\bf Lemma~\ref{lem:fejer}.}
{\em
    Assume the hypotheses of Theorem~\ref{THM:MAJORIZATION},  and that~\eqref{eq:consensus-problem} has a KKT point, $\left(\mathbf{r}^*, z^*, \mathbf{u}^*\right)$.
    Then the sequence $\left\{\left(\mathbf{r}^{(k)}, z^{(k)}, \mathbf{u}^{(k)}\right)\right\}$ generated by the iterations~\eqref{eq:admm1}-\eqref{eq:admm2} satisfies
    \begin{equation}
        E^{(k+1)}\left(\mathbf{r}^*, z^*, \mathbf{u}^*\right) \leq E^{(k)}\left(\mathbf{r}^*, z^*, \mathbf{u}^*\right),
    \end{equation}
    where
    \begin{align}
       E^{(k)}& \left(\mathbf{r}^*, z^*, \mathbf{u}^*\right) = \\
       & \sum_{i=1}^{N} \Big( L_i \left\lVert {r_{i}^{(k)} - r_{i}^*} \right\rVert^{2}  
       + \frac{1}{\sigma^2} \left\lVert {z^{(k)} - z^*} \right\rVert^{2} \\
       & \mbox{\hspace{3.15cm}} + \frac{1}{\sigma^2} \left\lVert {u_i^{(k)} - u_i^*} \right\rVert^{2} \Big).
    \end{align}
}
\begin{proof}
    Since $\left(\mathbf{r}^*, z^*, \mathbf{u}^*\right)$ is a KKT point, it satisfies the saddle point condition
    \begin{equation}
      \mathcal{L}(\mathbf{r}^*, z^*, \mathbf{u}) \leq \mathcal{L}(\mathbf{r}^*, z^*, \mathbf{u}^*) \leq \mathcal{L}(\mathbf{r}, z, \mathbf{u}^*)
    \end{equation}
    for any $(\mathbf{r}, z, \mathbf{u})$, where $\mathcal{L}$ is the Lagrangian function in~\eqref{eq:lagrangian}.
    Specifically at $(\mathbf{r},z) = ( \mathbf{r}^{(k+1)}, z^{(k+1)})$, we have
    \begin{align} \label{eq:thm-lhs1}
        &f(\mathbf{r}^{(k+1)}) - f(\mathbf{r}^*)\\ 
        \label{eq:thm-lhs2}
        & \quad + \frac{1}{\sigma^2} \sum_{i=1}^{N} \left\langle u_i^*, r_i^{(k+1)} - z^{(k+1)} \right\rangle \geq 0.
    \end{align}
    We now use Lemma~\ref{lem:saddle-point-bnd} and Lemma~\ref{lem:inner-prod} to produce an upper-bound for this quantity. 
    First we consider the left side of the inequalities in each lemma.
    By subtracting the left side of the inequality in Lemma~\ref{lem:inner-prod} evaluated at $\mathbf{u} = \mathbf{u}^*$ from the left side of the inequality in Lemma~\ref{lem:saddle-point-bnd} evaluated at $\left(\mathbf{r}, z \right) = \left(\mathbf{r}^*, z^* \right)$, we get
    
    \begin{align} \label{eq:lemma3-1}
    & f(\mathbf{r}^{(k+1)})-f(\mathbf{r}^{*})  \\
    & \quad + \frac{1}{\sigma^2} \sum\limits_{i=1}^{N}\left \langle u_i^{\left(k+\frac{1}{2}\right)}, \left( r_{i}^{(k+1)} - r_i^{*} \right) - \left( z ^{(k+1)} -z^{*} \right) \right \rangle\\
    & \quad - \frac{1}{\sigma^2}\sum_{i=1}^{N} \left\langle u_i^{(k+\frac{1}{2})}-u_i^{*}, r_i^{(k+1)} - z^{(k+1)} \right\rangle. \label{eq:lemma3-2}
    \end{align}
    By the feasibility condition of the KKT point, $r_i^* = z^{*}$ for each $i$, 
    and hence~\eqref{eq:lemma3-1}--\eqref{eq:lemma3-2} simplifies to the left side of~\eqref{eq:thm-lhs1}.
    Now, by subtracting the respective right sides of the inequalities in Lemma~\ref{lem:saddle-point-bnd} and Lemma~\ref{lem:inner-prod}, and combining with~\eqref{eq:thm-lhs1}, we get
    \begin{align}\label{eq:thm-bound1}
        0 & \leq f(\mathbf{r}^{(k+1)}) - f(\mathbf{r}^*)\\ 
        & \quad + \frac{1}{\sigma^2} \sum_{i=1}^{N}  \left\langle u_i^*, r_i^{(k+1)} - z^{(k+1)} \right\rangle \\
        \label{eq:telescope1}
        &\leq \frac{1}{2} \sum\limits_{i=1}^{N}  \Biggl[ L_{i} \left( \left\lVert {r_{i}^{(k)} - r_{i}^*} \right\rVert^{2} -  \left\lVert {r_{i}^{(k+1)} - r_{i}^*} \right\rVert ^{2} \right) \\
        & \quad + \frac{1}{\sigma^2} \left( \left\lVert {z^{(k)} - z^*} \right\rVert ^{2} - \left\lVert {z^{(k+1)} - z^*} \right\rVert ^{2} \right) \\
        \label{eq:telescope2}
        & \quad +  \frac{1}{\sigma^2} \left( \left\lVert {u_{i}^{(k)} - u_{i}^*} \right\rVert ^{2}  -  \left\lVert {u_{i}^{(k+1)} -u_{i}^*} \right\rVert^{2}\right) \\
        \label{eq:thm-bound2}
        & \quad - \frac{1}{\sigma^2} \left( \left\lVert {z^{(k+1)} - z^{(k)}} \right\rVert ^{2} + \left\lVert {u_i^{(k+1)} - u_i^{(k)}} \right\rVert ^{2} \right)\Biggr].
    \end{align}
    Finally, dropping the final non-positive term and rearranging the remaining terms gives
    \begin{align}\label{eq:lyapunov1}
        & \sum_{i=1}^{N} L_i \left\lVert {r_{i}^{(k)} - r_{i}^*} \right\rVert^{2} +  \frac{1}{\sigma^2} \left\lVert {z^{(k)} - z^*} \right\rVert^{2} + \frac{1}{\sigma^2} \left\lVert {u_i^{(k)} - u_i^*} \right\rVert^{2} \\
        &\geq  \sum_{i=1}^{N} L_i \left\lVert {r_{i}^{(k+1)} - r_{i}^*} \right\rVert ^{2} + \frac{1}{\sigma^2} \left\lVert {z^{(k+1)} - z^*} \right\rVert ^{2} \\ 
        \label{eq:lyapunov2}
        & \quad + \frac{1}{\sigma^2} \left\lVert {u_{i}^{(k+1)} - u_{i}^*} \right\rVert ^{2},
    \end{align}
    which completes the proof.
\end{proof}

\bibliographySupp{zotero}
\bibliographystyleSupp{ieeetr}
\end{document}